\documentclass[11pt]{article}

\usepackage{geometry}
\makeatletter

\def\@listI{
	\leftmargin\leftmargini
	\parsep .25ex \@plus .1ex
	\topsep .25ex \@plus .1ex
	\itemsep \parsep
}
\let\@listi\@listI
\@listi

\makeatother

\usepackage{amsmath,amssymb,amsfonts,mathrsfs}
\usepackage{amsthm}
\usepackage{bm}
\usepackage{xfrac}
\usepackage{nicefrac}

\newtheorem{proposition}{Proposition}

\allowdisplaybreaks

\let\emptyset\varnothing
\newcommand{\norm}[1]{\left\lVert#1\right\rVert}

\usepackage{graphicx}
\graphicspath{{Figures/}}

\usepackage{booktabs}
\usepackage{multirow}
\usepackage{array}
\usepackage{multicol}
\usepackage{longtable}
\usepackage{caption}

\usepackage{enumitem}

\usepackage[ruled]{algorithm}
\usepackage{algpseudocode}
\usepackage{comment}
\usepackage{xcolor}
\usepackage{soul}
\usepackage{ulem}

\definecolor{gris}{rgb}{0.44,0.44,0.44}

\newcommand{\tocheck}{\textcolor{black}}
\newcommand{\tocheckTwo}{\textcolor{black}}

\usepackage[colorinlistoftodos,textsize=scriptsize]{todonotes}

\usepackage{natbib}
\bibpunct[, ]{(}{)}{,}{a}{}{,}

\PassOptionsToPackage{hyphens}{url}
\usepackage{hyperref}
\hypersetup{
	colorlinks=true,
	citecolor=blue,
	urlcolor=blue,
	linkcolor=blue,
	breaklinks=true
}

\RequirePackage{tgtermes}
\RequirePackage{endnotes}

\usepackage{tikz}
\usetikzlibrary{calc,decorations.markings}

\usepackage{fancyhdr}
\usepackage{authblk}

\title{\textbf{Modern column generation for estimating single- and multi-purchase ranked list choice models}}

\author[1]{Luciano Costa
\thanks{Corresponding author.\\
	Email addresses: \texttt{luciano.costa@academico.ufpb.br} (Luciano Costa), \texttt{g.berbeglia@mbs.edu} (Gerardo Berbeglia), \texttt{claudio.contardo@concordia.ca} (Claudio Contardo),
	\texttt{jean-francois.cordeau@hec.ca} (Jean-Fran{\c c}ois Cordeau),
	}
}
\author[2]{Gerardo Berbeglia}
\author[3]{Claudio Contardo}
\author[4]{Jean-Fran{\c c}ois Cordeau}

\affil[1]{\small Department of Production Engineering and LOG-UFPB, Federal University of Paraíba, João Pessoa, Brazil, 55014-900}
\affil[2]{\small Melbourne Business School, The University of Melbourne, Carlton, Victoria 3053, Australia}
\affil[3]{\small Gina Cody School of Engineering and Computer Science, Concordia University, CIRRELT \& GERAD, H3G 2W1, Montr{\'e}al (Qu{\'e}bec) Canada}
\affil[4]{\small Department of Logistics and Operations Management, HEC Montr{\'e}al, CIRRELT \& GERAD, H3T 2A7, Montr{\'e}al, Canada}

\date{\today}

\begin{document}
	
\maketitle
	
\paragraph{Abstract:}
This paper studies the estimation of ranked-list discrete choice models with single and multiple purchases.
In this setting, each consumer type is characterized by a ranking over a subset of products and a desired number of purchases, and the estimation task is to identify the set of consumer types and their probabilities that best explain the observed transactional data.
This problem is computationally challenging due to the exponential number of possible consumer types and becomes more difficult when multiple purchases are allowed.
We propose a column generation framework for this problem.
Our main contribution is a dynamic programming algorithm for the column generation subproblem.
This subproblem generalizes the linear ordering problem and incorporates acceleration techniques to improve computational efficiency.
To the best of our knowledge, this is the first dynamic programming-based approach for generating consumer types in non-parametric models.
The proposed framework supports multiple model variants with minor modifications.
Computational experiments on synthetic and real data show substantial speedups over existing methods while maintaining high solution quality, and demonstrate effectiveness in both estimation and assortment optimization.
	


\section{Introduction}\label{sec:introction}

Discrete choice models (DCMs) are a class of stochastic models used to understand and predict the choices of consumers when faced with a finite set of alternatives.
Most DCMs are based on the notion of random utility, an idea originally proposed by \citet{thurstone1927law} and mathematically formalized by \citet{Block1959}.
Under the random utility framework, alternatives are endowed with a probability distribution of utility.
When a consumer observes the alternatives, the random utilities are realized, and the person selects the one with the highest utility.
The richness of DCMs in modeling complex choice behavior lies in the flexibility of specifying the joint probability distribution of utilities.
DCMs have been studied extensively across different fields, including mathematical psychology, economics, operations, revenue management, and marketing.
In operations and revenue management, applications of DCMs range from
predicting demand for transportation systems \citep{Cirillo2011},
estimating consumer demand for inventory management \citep{Lee2016},
assortment optimization \citep{gallego2019assortment, Heger2024},
to pricing of products and services \citep{Talluri2004}
and the design of emergency systems \citep{Lovreglio2014}.

One of the most studied applications of DCMs in revenue management is the \textit{assortment optimization problem}.
In this problem, a firm must determine an optimal subset of products to offer its consumers to maximize expected revenue \citep{Talluri2004}.
The complexity of the assortment problem arises not only from the exponential number of potential assortments, which can be obtained by selecting subsets from all available products, but also from the need to consider multiple aspects when solving the problem.
Simple strategies, such as offering only high-priced or only inexpensive products, are usually suboptimal.
If only expensive products are offered, many consumers may decide not to purchase anything.
On the other hand, focusing solely on low-profit products in an attempt to attract a specific segment may cannibalize demand for highly profitable items, leading to a decrease in the firm's overall profit.
Moreover, the assortment problem may also be subject to additional operational complexities and constraints, such as space capacity restrictions and product positioning, which have an impact on consumer behavior.

Most DCMs assume that consumers purchase only one product at a time. While this assumption can be reasonable for high-priced items, in most day-to-day contexts, consumers typically buy multiple related products within the same category \citep{Tulabandhula2023}. 
More recently, several works modeling multi-purchase behavior have appeared in the literature.
To capture such behavior, models typically rely on extensions of the well-established multinomial logit (MNL) framework \citep{Sun2020, Chen2022, Jasin2024, Tulabandhula2023, Wang2023, Bai2024}.
Very recently, a new approach based on a multi-choice ranked-list model has also been proposed by \citet{Lin2025}.
In a multi-purchase ranked-list model, a consumer type is represented with two components:
(i) a preference list, in which items, including the no-purchase option, are ranked in decreasing order of desirability, and
(ii) a maximum purchase limit, which specifies the largest number of products the consumer is willing to purchase.
Faced with an assortment, the consumer chooses the highest-ranked available items from their list, up to their predefined limit.

In the context of assortment optimization problems, when attempting to determine the best possible offer set, it is important to know, or at least accurately estimate, consumer preferences since they directly affect the products' demand.
Firms typically rely on historical transactional data to understand consumer behavior in response to various alternatives.
However, accurately estimating demand can be challenging.
In the past, several applications assumed that item demands were independent, i.e., unaffected by the availability of other products \citep{Strauss2018}. This assumption, however, ignores \textit{substitution effects}, leading to poor outcomes.
The substitution effect plays an important role in demand estimation:
if consumers cannot find their preferred item, they may opt for a similar alternative if available.
Independent demand models are also incapable of capturing \textit{demand spilling}, which refers to the demand lost when a consumer's first choice is not available.
Ignoring substitution and spilling effects can result in significant under- or over-estimation of demand.
In this context, DCMs are particularly suitable for demand estimation, as they explicitly model consumer preferences over sets of alternatives and account for these effects.
When considering ranked-list-based models, demand estimation involves determining a discrete \emph{probability mass function} (pmf) associated with the various consumer types, making them a generic and powerful tool in choice modeling.
Although \citet{Block1959} showed that this approach contains all utility-based choice models, such as MNL models, a significant drawback lies in the exponential number of consumer types. 
Nonetheless, this issue has been effectively addressed in the literature through various strategies, including robust optimization \citep{Farias2013}, \emph{column generation} (CG) \citep{Vulcano2015, bertsimas2016data, Lin2025}, and Benders decomposition \citep{Bertsimas2019}.

\paragraph{\textbf{Contributions}}
The primary contribution of this work is the development of an efficient unified framework for estimating general single- and multi-purchase DCMs.
Inspired by the works of \citet{Vulcano2015}, \citet{bertsimas2016data}, and \citet{Lin2025}, we provide a scalable and exact CG procedure for estimating non-parametric single- and multi-purchase choice models from transactional data. 
Beyond this central contribution, the paper introduces several additional technical advances:
\begin{enumerate}
	\item \textbf{Dynamic programming for solving the CG subproblem}:
	We develop a novel DP algorithm for solving the CG subproblem, a generalized linear ordering problem (GLOP).
	To the best of our knowledge, this is the first DP-based approach for generating consumer types in non-parametric ranked-list models.
	
	\item \textbf{Acceleration strategies for solving CG subproblem exactly}:
	We propose several acceleration techniques (e.g., completion bounds, unreachable products, and DP-based heuristic pricing) that substantially reduce computational effort while maintaining exact pricing guarantees.
	
	\item \textbf{Unified framework for multiple settings}:
	Our method can easily handle single- and multi-purchase model estimation and settings where consumers have limited consideration sets \citep{Feldman2019} with minor modifications to the DP, providing a unified framework that covers a broad range of choice model variants.
	Moreover, our experiments show that, when applied to estimating single- and multi-purchase ranked-list models, our algorithm is substantially faster than the state of the art.
	It is worth noting that the DP method can exactly solve the problem of generating consumer types under both maximum-likelihood and minimum $\ell_1$-error objectives.
	This is not possible, for instance, when relying on existing mixed-integer linear programming (MILP) formulations to enumerate consumer types, as they assume non-negative revenues and, therefore, can only be used within EM estimation models.
\end{enumerate}

In addition to these methodological contributions, we have conducted extensive experiments using synthetic and real datasets to assess the performance of the methodology developed.
We observed that the method performs well even for instances with a relatively large number of transactions, a setting that poses challenges for several existing methods in the literature \citep{Vulcano2015}.
\tocheck{When compared to the methodology of \cite{Vulcano2015} ---a CG framework that relies on the maximum-likelihood estimation and solves the pricing subproblem via a MIP solver--- our approach delivers speedups of nearly $40\times$ when we consider instances with 10 products and 6,000 transactions.
	This speed-up could have been greater, had we not excluded the instances for which the benchmark method did not reach a solution within the time limit.}
The experiments also highlight the practical relevance of multi-purchase behavior: estimating a model that allows consumers to purchase multiple products, rather than restricting them to a single item, improves predictive accuracy and can increase revenues in assortment optimization.
\tocheck{In particular, when considering synthetic data, we observe that} accounting for multi-purchase behavior can reduce prediction errors by up to 70\% compared to a single-purchase setting, while increasing revenues by up to 20\%, \tocheck{and ensuring speedups greater than $100\times$ when compared to a standard market-discovery method similar to that of \citet{Vulcano2015}}.
Moreover, using the same synthetic dataset as \citet{Lin2025}, our multi-purchase method delivers higher predictive accuracy and higher revenues, despite relying on a simpler behavioral model that ignores the richer attribute set considered in \citet{Lin2025}, which considers behavior-revealed-preference (BRP) rules in the estimation framework.

\paragraph{\textbf{Organization of the paper}}

The remainder of this paper is organized as follows:
In Section \ref{sec:lit_review}, we provide an overview of some important works proposing both single- and multi-purchase DCMs.
Section \ref{sec:problem_description} introduces some notation and formulates the estimation problems that we address.
In Section \ref{sec:column_generation}, we describe our CG method, along with the newly developed DP algorithm proposed to tackle the CG subproblems in this context, namely the problem of finding consumer types of maximum profit for addressing the single- and multi-purchase estimation problems.
Section \ref{sec:additional} is devoted to adapting the proposed DP to address additional attributes and particular cases, namely the single-purchase setting, consumer types with limited-sized lists (i.e., small consideration sets).
Section \ref{sec:acceleration_techiniques} is devoted to presenting several acceleration techniques designed to speed up the resolution of the DP, and consequently, the overall estimation procedure.
Section \ref{sec:computational_experiments} presents numerical results demonstrating the efficiency of our method and the accuracy in the estimation and revenue generation steps.
Finally, Section \ref{sec:conclusions} provides some conclusions and insights derived from the computational results presented.

\section{Literature review}\label{sec:lit_review}

In this section, we provide a brief overview of parametric and non-parametric DCMs in the context of both single- (Section \ref{sec:single_purchase_models}) and multi-purchase (Section \ref{sec:multi_purchase_models}) models.

\subsection{Single-purchase models}\label{sec:single_purchase_models}

In the literature on single-purchase DCMs, both parametric and non-parametric approaches have been extensively explored.
Parametric models rely on analytical expressions with a limited number of parameters, enabling the inclusion of covariates such as product features and prices \citep{Vulcano2010, Berbeglia2022}.
The most classical formulation is the MNL model \citep{Luce1965}, whose closed-form probabilities and concave log-likelihood allow efficient estimation.
However, MNL suffers from the independence-from-irrelevant-alternatives (IIA) property \citep{Train2009, Gallego2019}, which can produce unrealistic substitution patterns.
To mitigate the effects of IIA, several extensions of the MNL model have been developed over the years, including:
latent-class MNL \citep{Bhat1997},
nested MNL \citep{Train2009},
mixed MNL \citep{McFadden2000},
and Markov chain choice models \citep{Blanchet2016}.
Although these models alleviate the IIA limitation, they typically increase computational burden.
Recent surveys provide comprehensive overviews of parametric formulations \citep{Strauss2018, Berbeglia2022}.

Advances in data and computation have motivated an increase in interest in non-parametric approaches \citep{Misic2020, Berbeglia2022}.
Among these, ranked-list (stochastic preference) models form a fundamental class introduced by \citet{Block1959} and later adopted for assortment optimization by \citet{Mahajan2001}.
These models rely solely on observed behavior rather than predefined utility structures, making them particularly appealing in operational settings \citep{Honhon2012, Farias2013, Vulcano2015, bertsimas2016data, VanRyzin2017, Jena2020}.
\citet{Honhon2012} study assortment optimization under several structured substitution patterns, highlighting that offering the most profitable product is not necessarily optimal.
\citet{Farias2013} estimate sparse non-parametric demand distributions via robust optimization and constraint sampling, producing accurate revenue predictions. A complementary line of work considers structured probability models over rankings to retain tractability \citep{DesirGoyalJagabathulaSegev2021, feng2025mallows}.

An important methodological advance is the market-discovery framework of \citet{Vulcano2015}, which employs a CG algorithm under maximum likelihood to estimate consumer types under a general ranked-list model from sales transactions and item availability data.
The generation of new consumer types is done by solving a GLOP via MILP.
Although effective, this approach scales to only a few hundred transactions.
\citet{bertsimas2016data} propose an alternative based on $\ell_1$-error minimization with heuristic pricing.
While both studies aim to estimate the pmf of consumer types consistent with observed transactions, they differ in their estimation frameworks and in the strategies used to solve the GLOP subproblem.
Relative to \citet{Vulcano2015}, in our method, we replace the MILP-based pricing with an exact DP pricing algorithm, allowing us to solve instances with tens of thousands of transactions.
Compared to \citet{bertsimas2016data}, who employ heuristic pricing, our method maintains exact pricing guarantees while optionally allowing heuristics to speed up early iterations.

Estimators based on expectation-maximization (EM) have also been explored.
\citet{VanRyzin2017} build upon the method of \citet{Dempster1977} to develop an EM algorithm and apply it to recover both arrival rates and type probabilities.
\citet{HajMirzaei2020} propose an enhancement to the market-discovery approach of \citet{Vulcano2015} and \citet{VanRyzin2017}, initializing the algorithm with consumer types generated by an evolutionary algorithm rather than singleton consumers.
From a methodological perspective, \citet{MendezDiaz2019} develop a branch-and-cut (BC) framework for the GLOP.
When embedded into an estimation setting, the resulting framework yields solutions with better statistical quality than those obtained using a simple greedy heuristic, at the expense of a computational time that is approximately one order of magnitude larger.

Often, consumers do not consider for purchase more than just a small number of products. Limited-consideration set models have been investigated by \citet{Feldman2019}, who design LP-based rounding algorithms for assortment problems under short preference lists.
Unlike \citet{Feldman2019}, who develop approximation results for limited-consideration consumers, our framework incorporates restricted preference lists within an exact estimation algorithm, requiring only minor modifications to our DP.
Other important contributions in the literature include Benders-based formulations for product-line design \citep{Bertsimas2019} and worst-case performance guarantees for revenue-ordered assortments  \citep{Rusmevichientong2014, aouad2018approximability, Berbeglia2020}.
More recently, \citet{Jena2020} introduce partially ranked consumer preferences and estimate them under an $\ell_1$-error framework using CG over preference trees.
\citet{Jena2022} further extend this idea to the generalized stochastic preference (GSP) model \citep{Berbeglia2018a}, capturing behavioral biases and efficiently scaling to irrational choice data.
Recent work also studies how to learn ranked preferences reliably from noisy or strategic platform data;
\citet{Golrezaei2022a}, for instance, develop methods to learn product rankings that are robust to manipulation by fake users. Finally, \citet{aouad2025representing} showed that every random utility discrete choice model (or equivalently every ranked list model) can be represented by a well-defined class of neural networks.

\subsection{Multi-purchase models}\label{sec:multi_purchase_models}

In contrast to single-purchase models, multi-purchase DCMs have only recently gained attention in the literature.
Most contributions extend parametric MNL-based structures, though semi- and non-parametric ranked-list models have also begun to emerge.

One of the earliest works is \citet{Tulabandhula2023}, motivated by online recommendation systems.
The authors propose BundleMVL-$K$, an extension of the MNL model that allows consumers to purchase up to $K$ items.
The model captures richer substitution patterns without substantially increasing complexity, remaining computationally tractable.
In practice, BundleMVL-2 yields revenue improvements of at least 5\% in instances with roughly 1,500 products and 17\% on average across six real datasets.
To solve the associated assortment problem, the authors design a binary-search-based algorithm exploiting structural properties of the optimal solution.

\citet{Jasin2024} study uncapacitated and capacitated assortment problems under the Multivariate MNL (MVMNL) model, highlighting differences relative to BundleMVL.
Whereas BundleMVL fixes bundle size exogenously, MVMNL determines purchases endogenously through utility maximization and allows grouping of products to capture substitution and complementarity effects.
The authors show that, unlike in single-purchase settings, revenue-ordered assortments may not always be optimal; however, in uncapacitated problems, the optimal solution can consist of revenue-ordered local assortments by category.
\citet{Chen2022} also work with MVMNL, proving strong NP-hardness and proposing a $\frac{1}{2}$-approximation via adjusted revenue ordering, strengthened to $0.74$ through LP relaxation.

\citet{Wang2023} introduce a multi-purchase MNL model where consumers buy at most two distinct products.
This assumption is supported by data from a Chinese retailer showing that fewer than 3\% of consumers purchase more than two items and fewer than 4\% purchase more than one unit of a product.
The resulting log-likelihood function is concave, enabling efficient estimation.
The authors study scenarios with full and partial observability of purchase sequences and solve the corresponding assortment problem via MILP, validating their approach on real and synthetic datasets.
\citet{Bai2024} extend the random-utility MNL framework to multi-purchase settings (MP-MNL) with an endogenous purchase limit.
A consumer selects all items whose utilities exceed the no-purchase option, subject to capacity.
In contrast to \citet{Tulabandhula2023}, who rely on particular structural assumptions, \citet{Bai2024} develop a recursive probability computation and propose polynomial-time approximation schemes for both bounded and general settings.


	\citet{he2026mip} study logit-based multi-purchase choice models and propose MIP formulations for assortment optimization problems.
	The authors introduce a hypergraph-based representation that captures general bundle-based choice structures and unifies several models in the literature, including BundleMVL-$K$ and MVMNL.
	This representation enables the derivation of strong formulations and shows that assortment optimization under general Logit-MP models is NP-hard to approximate.
	Their approach yields tighter LP relaxations than the classical Big-M formulation and characterizes conditions under which the formulations are locally sharp.
	To ensure computational tractability, they develop a cutting-plane algorithm and demonstrate significant computational improvements over existing methods.

Beyond parametric approaches, \citet{Sun2020} propose marginal distribution models (MDMs), a semiparametric framework that generalizes MNL, multiple-discrete choice models, and threshold-utility models.
They introduce a mixed-MDM version to account for heterogeneity and show that the framework sits strictly between RUMs and regular choice models, achieving strong empirical performance.
Moving toward non-parametric methods, \citet{Lin2025} propose a multi-choice ranked-list model (MC-RLM) in which a preference list and a purchase capacity define each consumer type.
Consumers select up to $\eta$ products according to their rankings.
The authors extend the market-discovery method of \citet{Vulcano2015}, incorporate BRP rules (e.g., views, clicks, cart additions), and generate new consumer types using an MILP.
They further propose a revenue-ordered heuristic with performance guarantees and demonstrate strong predictive accuracy relative to independent and MNL-based models.
Although \citet{Lin2025} also consider multi-purchase ranked-list models, our approach differs in that the pricing is determined via an exact DP rather than using an MILP solver, and our focus is on efficient estimation directly from transactional data rather than incorporating BRP rules.

Overall, when compared to other works in the literature, our framework is general and scalable: it handles single-purchase, multi-purchase, and limited-consideration settings with minimal adjustments, while providing significant computational improvements through an exact DP pricing algorithm enhanced with multiple acceleration strategies.

\section{Problem description}\label{sec:problem_description}

In this section, we describe our approach to estimating pmfs over consumer types from historical transaction data.
In Section \ref{sec:model_basis}, we describe the problem that we study and introduce the notation that is considered throughout the paper.
Then, in Section \ref{sec:estimation_problem}, we describe the estimation methods that we employ to solve the problem.

\subsection{Notation}\label{sec:model_basis}

Let $\mathcal{N} = \{1, \ldots, n\}$ be a set of $n$ substitutable items and let $\mathcal{T}$ denote a set of $T$ transactions.
Each transaction $t \in \mathcal{T}$ is associated with a non-empty offer set $S_t \subseteq \mathcal{N}$ and a bundle $B_t \subseteq S_t$ of purchases.
Based on $\mathcal{T}$, we can compute the empirical probabilities for purchasing a bundle as follows. Let us define
$\mathcal{T}(S) = \{t \in \mathcal{T} : S_t = S\}$ and
$\mathcal{T}(S, B) = \{t\in\mathcal{T}: S_t = S, B_t = B\}$.
For any given assortment $S \subseteq \mathcal{N}$ appearing in $\mathcal{T}$, and any bundle $B \subseteq S$, we define the empirical probability that bundle $B$ is selected when assortment $S$ is offered as:
\begin{equation}
	\hat{\mathbb{P}}(B | S) = \frac{|\mathcal{T}(S, B)|}{|\mathcal{T}(S)|}. \label{eq:empirical_probability}
\end{equation}

We assume that consumer demand can be represented as a distribution over a set of \textit{consumer types}.
A consumer type is represented by a tuple $c = (\sigma, \eta)$, where $\sigma$ is a sequence of products in $\mathcal{N}$ without repetition, and $\eta \in \{1, \ldots, n\}$.
A consumer of type $c = (\sigma, \eta)$ strictly prefers item $\sigma_i$ over item $\sigma_{i+1}$ for every $i \in \{1, \ldots, |\sigma| - 1\}$, where $|\sigma|$ denotes the length of the list.
When faced with an assortment $S$, the consumer first constructs a sub-list $\iota(\sigma, S)$ by removing from $\sigma$ all products not in $S$.
Then, the consumer purchases the bundle consisting of the top $\eta$ products in $\iota(\sigma, S)$.
If $|\iota(\sigma, S)| < \eta$, the consumer buys all products in $\iota(\sigma, S)$, and if $|\iota(\sigma, S)| = 0$, the consumer buys nothing.
Let $\pi(c, S)$ denote the bundle purchased by consumer $c$ when facing assortment $S$, according to this definition.
We say that a consumer type $c = (\sigma, \eta)$ is \textit{compatible} with a transaction $(S, B)$ if $\pi(c, S) = B$.

Let us denote $\mathcal{C}$ the set of all possible consumer types $c = (\sigma, \eta)$, whose cardinality is of exponential size in $n$.
By looking at the information from $\mathcal{T}$, it might not be possible to identify specifically which consumer types were involved in the transactions.
Yet, given the full set of consumer types $\mathcal{C}$, we can indicate which ones are \emph{compatible} with a transaction $t \in \mathcal{T}$.
This compatibility implies that these consumers would select all the products in the bundle $B_t$ and none from $S_t \setminus B_t$ when presented with an assortment $S_t$.
For each transaction $t \in \mathcal{T}$, we define
$\mathcal{C}_t = \{c = (\sigma, \eta) \in \mathcal{C} : \pi(c, S_t) = B_t \}$, which is the set of consumers compatible with transaction $t$.
By defining $\textbf{x} = \mathbb{P}(\boldsymbol{c})$ as the pmf associated with consumer types in $\mathcal{C}$, the probability that a consumer will select bundle $B_t$ when offered the assortment in $S_t$ is as follows:
\begin{equation}
	\mathbb{P}(B_t | S_t, \textbf{x}) = \sum\limits_{c \in \mathcal{C}_t} x_{c}\label{eq:conditional_probability}.
\end{equation}
For the sake of conciseness, expression \eqref{eq:conditional_probability} can be expressed simply as $\mathbb{P}(B_t | S_t)$ when it is clear that the probability results from the DCM.

\subsection{Formulation of the estimation problem}\label{sec:estimation_problem}

The estimation problem for a ranked-list based choice model consists of finding the pmf $\textbf{x} = \mathbb{P}(\boldsymbol{c})$ associated with consumer types in $\mathcal{C}$.
In this section, we show how the estimation problem can be defined when two of the most important estimation methods are considered, namely:
\emph{maximum likelihood} and
\emph{minimum $\ell_1$ predictive errors}.
Even though ranked-list-based choice models are a non-parametric approach, for simplicity, we use the term \emph{parameter} when referring to the estimated values from the training set.

\subsubsection{Maximum likelihood}\label{sec:MLE}

The maximum likelihood estimator aims to find the set of parameters associated with a model for which the collected data is the most likely to be observed.
Assuming that the full set of observations (transactions) is independent and identically distributed (\emph{i.i.d.}), it is possible to derive the likelihood function by simply multiplying the probability for each observation occurring.
Due to the properties associated with the $\log$ function, it is preferable to work with the data log-likelihood function.
For a ranked-list-based choice model, the log-likelihood function can be obtained as follows.

Let $\mathcal{P}$ be the set of periods over which transactions in $\mathcal{T}$ take place.
We do not make any assumptions about the duration of the periods.
Each period could correspond to a day, a week, or even chronological intervals with different durations.
As discussed by \citet{Vulcano2015}, one may partition the planning horizon into segments so that the consumer arrival rate is constant and equal over all segments.
In our case, we assume that a period is a chronological interval in which one transaction takes place.
It is presumed that consumer arrivals occur according to a homogeneous Bernoulli process, where at each period $p \in \mathcal{P}$, an arrival occurs with probability $0 < \lambda < 1$.
The value of $\lambda$ can be estimated from the data.

For analysis purposes, we partition the set of periods $\mathcal{P}$ into three subsets, namely:
$\mathcal{P}_p$, set of periods with selections (purchases);
$\mathcal{P}_{0}$, set of periods with arrivals but no selections (no-purchases);
and $\mathcal{P}_{\bar{\lambda}}$, set of periods with no arrivals.
Due to the definition of period that we have adopted, the probability associated with the transaction in a given period is the same as in \eqref{eq:conditional_probability}.
Thus, we adopt the same approach employed in \citet{Vulcano2015} and write the incomplete log-likelihood function as follows:
\begin{align}
	\begin{split}
		& \log \mathscr{L}(\textbf{x}, \lambda) = \sum_{t \in \mathcal{P}_p} \log \lambda + \log \mathbb{P}(B_t | S_t, \textbf{x}) \\
		& + \sum_{t \in \mathcal{P}_0} \log \lambda + \log \mathbb{P}(\emptyset | S_t, \textbf{x}) + \sum_{t \in \mathcal{P}_{\bar{\lambda}}} \log (1 - \lambda).
		\label{eq:log_likelihood_expanded}
	\end{split}
\end{align}
When performing a maximum likelihood estimation, one seeks to maximize the function \eqref{eq:log_likelihood_expanded}.
Given that \eqref{eq:log_likelihood_expanded} is separable in \textbf{x} and $\lambda$, and globally concave in $(\textbf{x}, \lambda)$, by computing $\sfrac{\partial \mathscr{L}(\textbf{x},\lambda)}{\partial \lambda} = 0$, we can find that $\lambda^{*} = \sfrac{(|\mathcal{P}_p| + |\mathcal{P}_0|)}{(|\mathcal{P}_p| + |\mathcal{P}_0| + |\mathcal{P}_{\bar{\lambda}}|)}$.
The remaining optimization problem can be formulated as follows:
\begin{equation}
	\begin{array}{ll}
		\max\limits_{\textbf{x} \geq 0} & \mathscr{L}(\textbf{x}) \\
		\text{s.t.} & \textbf{1}^\intercal \textbf{x} = 1. \\
	\end{array} \label{eq:mle}
\end{equation}
At this point, we can exploit the knowledge that we have acquired from the available data and apply a variable transformation to write the log-likelihood function $\mathscr{L}(\textbf{x})$ in terms of conditional probabilities \eqref{eq:conditional_probability}.
To do so, we define new variables $y_t$, $\forall t \in \mathcal{T}$, which are defined as $y_t = \sum_{c \in \mathcal{C}_t} x_{c}$, and correspond to the probability that a selection takes place in the period.
The data log-likelihood function can then be rewritten as: 
\begin{align}
	\mathscr{L}(\textbf{y}) = \sum_{t \in \mathcal{P}_p} \log y_t + \sum_{t \in \mathcal{P}_{0}} \log y_t,
\end{align}
and the estimation problem \eqref{eq:mle} becomes:
\begin{equation}
	\begin{array}{ll}
		\max\limits_{\textbf{x}, \textbf{y} \geq 0} & \mathscr{L}(\textbf{y}) \\
		\text{s.t.}	& \textbf{A} \textbf{x} - \textbf{y} = 0 \\
		& \textbf{1}^\intercal \textbf{x} = 1, \\
	\end{array} \label{eq:mle_y}
\end{equation}
where $\textbf{A} \in \{0, 1\}^{|\mathcal{T}| \times |\mathcal{C}|}$ is a binary matrix, where $a_{tc} = 1$, if $c \in \mathcal{C}_t$, or $a_{tc} = 0$, otherwise.

\subsubsection{Minimum $\ell_1$ predictive errors}

Let $\mathcal{M} = \{(S_t, B_t): t\in\mathcal{T}\}$ be the set of all distinct transactions, this is without considering repetitions.
Another way of estimating parameters $\textbf{x} = \mathbb{P}(\boldsymbol{c})$ associated with ranked-list based choice models is to minimize the estimation error given by the difference between the observed probabilities $\hat{\mathbb{P}}(B | S)$ and the estimated probabilities $\mathbb{P}(B | S, \textbf{x})$ for all possible distinct transactions $(S, B)\in\mathcal{M}$.
This strategy has been adopted by \citet{bertsimas2016data} and \citet{Jena2020}.

For the sake of conciseness, let us define a vector $\textbf{v}$ containing all empirical probabilities $v_{SB} > 0$, $\forall (S, B) \in \mathcal{M}$, where $v_{SB} = \hat{\mathbb{P}}(B|S)$.
The compatibility matrix $\mathbf{A}$ defined in the previous section is slightly modified, and is such that $\mathbf{A} \in \{0,1\}^{|\mathcal{M}| \times |\mathcal{C}|}$, where $a_{mc} = 1$, if $\pi(c, S) = B$ for $m = (S, B)$, or $a_{mc} = 0$, otherwise.
The estimation problem can then be formulated as follows:
\begin{equation}
	\begin{array}{ll}
		\min \limits_{\textbf{x} \geq 0} & \norm{A \textbf{x} - \textbf{v}}_{1} \\
		\text{s.t.}	&  \textbf{1}^\intercal \textbf{x} = 1.\\
	\end{array} \label{eq:l1_error}
\end{equation}
Problem \eqref{eq:l1_error} can be re-written as a linear problem by considering vectors $\boldsymbol{\epsilon}^{+}, \boldsymbol{\epsilon}^{-} \in \mathbb{R}_{+}^{|\mathcal{M}|}$ of artificial variables.
The estimation problem becomes:
\begin{equation}
	\begin{array}{ll}
		\min \limits_{\textbf{x}, \boldsymbol{\epsilon}^{+}, \boldsymbol{\epsilon}^{-} \geq 0} &  \textbf{1}^\intercal \boldsymbol{\epsilon}^{+} + \textbf{1}^\intercal \boldsymbol{\epsilon}^{-}\\
		\text{s.t.} & A \textbf{x} + \boldsymbol{\epsilon}^{+} - \boldsymbol{\epsilon}^{-} = \textbf{v} \\
		&  \textbf{1}^\intercal \textbf{x} = 1.\\
	\end{array} \label{eq:l1_error_lin}
\end{equation}

\section{Column generation}\label{sec:column_generation}

Both estimation problems \eqref{eq:mle_y} and \eqref{eq:l1_error_lin} are associated with extended formulations containing an exponential number of variables (parameters).
When solving these problems, handling all their variables at once may be computationally prohibitive.
In this context, CG arises as a suitable technique to tackle problems with such structures.
CG is an iterative approach that alternates between solving a \emph{restricted master problem} (RMP) and the CG subproblem.
Both problems \eqref{eq:mle_y} and \eqref{eq:l1_error_lin} constitute the so-called CG \emph{master problems} (MPs).
Due to their large number of variables, when applying CG, one considers at each time only a small subset of variables from the MP, forming the RMP.
At each iteration, new variables are generated by solving the CG subproblem, which considers dual information from the RMP in the process.
When the CG subproblem can no longer find variables with negative reduced cost (for a minimization problem), the CG algorithm stops.

For the sake of comprehensiveness, we organize the remainder of this section as follows. 
Section \ref{sec:master_problem} presents the MPs associated with \eqref{eq:mle_y} and \eqref{eq:l1_error_lin}. 
Section \ref{sec:pricing_suproblem} then describes the proposed DP algorithm for solving the CG subproblems, which consist of finding promising consumer types to be incorporated into the RMPs.

\subsection{Master problem}\label{sec:master_problem}

The RMP is defined by considering a subset $\mathcal{C}' \subseteq \mathcal{C}$ of variables.
To make this section self-contained, we revisit some of the notation already defined in Section \ref{sec:problem_description}.
Moreover, when formulating the RMPs, we adopt a component-wise notation for representing the matrices and vectors.

\subsubsection{Maximum-likelihood master problem}\label{sec:loglikelihood_mp}

Let $a_{tc}$ be the binary parameter indicating whether a consumer type $c = (\sigma, \eta) \in \mathcal{C}$ is compatible with transaction $t \in \mathcal{T}$, or not.
Variables $x_{c}$, $c \in \mathcal{C}$, indicate the probability of a consumer, arriving randomly, to be of type $c$.
Variables $y_t$, $t \in \mathcal{T}$, correspond to the conditional probabilities defined in terms of $x_{c}$, and that express the probability of a bundle $B_t$ being selected when an assortment $S_t$ of items is available (Section \ref{sec:model_basis}).
The RMP associated with \eqref{eq:mle_y} can be explicitly expressed as:
\begin{align}
	\max & \sum_{t \in \mathcal{T}} \log y_t && \label{eq:obj_func_CG} \\
	\text{s.t.} & \sum_{c\in \mathcal{C}'} a_{tc} x_{c} - y_t = 0 && \forall t \in \mathcal{T} \label{eq:constraints_CG} \\
	& \sum_{c \in \mathcal{C}'} x_{c} = 1 && \label{eq:convexity_CG} \\
	& x_{c} \geq 0 && c \in \mathcal{C}' \label{eq:vars_x_CG}\\  
	& y_t \geq 0 && t \in \mathcal{T}. \label{eq:vars_y_CG}
\end{align}
The objective function \eqref{eq:obj_func_CG} to be maximized is the log-likelihood function of the maximum likelihood estimator. 
Constraints \eqref{eq:constraints_CG} are linking constraints that enable the computation of the conditional probabilities of a transaction taking place (values of the $y_t$ variables). 
Constraint \eqref{eq:convexity_CG} imposes that the sum of probabilities of all consumer types must be equal to 1. 
Moreover, \eqref{eq:convexity_CG} corresponds to the CG convexity constraint. 
Finally, constraints \eqref{eq:vars_x_CG} and \eqref{eq:vars_y_CG} define the domain of the variables.
\tocheck{We note that, although the variable $y_t$ is formally allowed to take value~0, which would render the objective function undefined, this situation never occurs in our implementation.
	Indeed, the master problem is initialized with a set of consumer types constructed so that at least one type is compatible with every transaction.
	Following the recommendation of \citet{Vulcano2015}, the initial consumer set is based on an independent-demand model and consists of a simple set of $n$ consumer types, each ranking exactly one product above the no-purchase alternative.}

Let $\mu_t$ and $\gamma$ be the dual variables associated with Constraints \eqref{eq:constraints_CG} and \eqref{eq:convexity_CG}, respectively.
As discussed in \citet{Vulcano2015}, a solution to the Lagrangian problem $\mathcal{L}(\textbf{x}, \textbf{y}, \boldsymbol{\mu}, \gamma)$ arising from \eqref{eq:obj_func_CG}--\eqref{eq:vars_y_CG} can be obtained by applying \emph{Karush-Kuhn-Tucker} (KKT) conditions to $\mathcal{L}(\textbf{x}, \textbf{y}, \boldsymbol{\mu}, \gamma)$:
\begin{align}
	& \frac{\partial}{\partial y_t} \mathcal{L}(\textbf{x}, \textbf{y}, \boldsymbol{\mu}, \gamma) = \frac{1}{y_t} - \mu_t = 0 && \forall t  \in \mathcal{T} \label{eq:kkt_eml_y}\\
	& \frac{\partial}{\partial x_{c}} \mathcal{L}(\textbf{x}, \textbf{y}, \boldsymbol{\mu}, \gamma) = \sum_{t \in \mathcal{T}} a_{tc} \mu_t + \gamma  \leq 0 && \forall c  \in \mathcal{C}. \label{eq:kkt_eml_x}
\end{align}
From \eqref{eq:kkt_eml_y}, one can compute $\mu_t = \sfrac{1}{y_t}$, $\forall t \in \mathcal{T}$. 
From \eqref{eq:kkt_eml_x}, \citet{Vulcano2015} derive a condition to determine whether a newly generated column should be incorporated into the RMP.
The authors show that if a column $c \in \mathcal{C}$ satisfies $\sum_{t \in \mathcal{T}} a_{tc} \mu_t > |\mathcal{T}|$, it constitutes an improvement direction for the Lagrangian dual function and, hence, can be added to the RMP. 
However, even if this condition provides an improvement direction for the maximum likelihood log-likelihood function, a newly generated consumer type should only be added to \eqref{eq:obj_func_CG}--\eqref{eq:vars_y_CG} if it is deemed statistically significant according to a log-likelihood ratio test. 
The use of such a test aims at avoiding overfitting.

Since function \eqref{eq:obj_func_CG} is non-linear, solving RMP \eqref{eq:obj_func_CG}--\eqref{eq:vars_y_CG} directly using an off-the-shelf solver is difficult.
To circumvent this issue, the EM method of \citet{VanRyzin2017} can be employed to estimate the probabilities associated with consumer types.
The EM method can be adopted as the RMP solver in the CG algorithm applied over \eqref{eq:obj_func_CG}--\eqref{eq:vars_y_CG} with dual variables $\mu_t, \forall t \in \mathcal{T}$ being estimated from \eqref{eq:kkt_eml_y}. 

\subsubsection{$\ell_1$-error master problem}\label{sec:l1_error_mp}

The notation employed to formulate the RMP associated with problem \eqref{eq:l1_error_lin} is similar to that adopted in Section \ref{sec:loglikelihood_mp}. 
For simplicity, we abuse the notation by referring to an element of $\mathcal{M}$ simply by its index $m$. 
Hence, a binary parameter $a_{mc}$ indicates whether a consumer type $c \in \mathcal{C}$ is compatible with a transaction $m = (S, B)$. 
Variables $\epsilon_m^{+}$ and $\epsilon_m^{-}$, $m \in \mathcal{M}$, are used to compute predictive errors between the estimated and observed probabilities. 
Problem \eqref{eq:l1_error_lin} can then be explicitly formulated as:
\begin{align}
	\min & \sum_{m \in \mathcal{M}} \epsilon_{m}^{+} + \epsilon_{m}^{-} && \label{eq:obj_func_CG_l1norm} \\
	\text{s.t.}  & \sum_{c \in \mathcal{C}'} a_{mc} x_{c} + \epsilon_{m}^{+} - \epsilon_{m}^{-} =  v_{m} && \forall m \in \mathcal{M} \label{eq:constraints_CG_l1norm} \\
	& \sum_{c \in \mathcal{C}'} x_{c} = 1 && \label{eq:convexity_CG_l1norm} \\
	& x_{c} \geq 0 && c \in \mathcal{C}' \label{eq:vars_x_CG_l1norm}\\  
	& \epsilon_{m}^{+}, \epsilon_{m}^{-} \geq 0 && m \in \mathcal{M}. \label{eq:vars_eps_CG_l1norm}
\end{align}
The objective function \eqref{eq:obj_func_CG_l1norm} minimizes the sum of errors arising from the differences between the estimated and observed probabilities.
Constraints \eqref{eq:constraints_CG_l1norm} account for the predictive errors generated in the estimation process. 
Constraint \eqref{eq:convexity_CG_l1norm} is equivalent to Constraint \eqref{eq:convexity_CG}. 
Finally, Constraints \eqref{eq:vars_x_CG_l1norm} and \eqref{eq:vars_eps_CG_l1norm} define the domain of the variables.

Let $\mu_m$ and $\gamma$ be the dual variables associated with constraints \eqref{eq:constraints_CG_l1norm} and \eqref{eq:convexity_CG_l1norm}, respectively.
A variable associated with a consumer $c \in \mathcal{C}$ generated by the CG subproblem is added to the RMP if its reduced cost satisfies
$\gamma + \sum_{m \in \mathcal{M}} a_{mc} \mu_m > 0$.
\citet{bertsimas2016data} highlight several advantages of applying CG to \eqref{eq:obj_func_CG_l1norm}--\eqref{eq:vars_eps_CG_l1norm} rather than to \eqref{eq:obj_func_CG}--\eqref{eq:vars_y_CG}.
In particular, the number of basic solutions obtained when solving the former is typically smaller than when solving the latter.
At optimality, this means fewer consumer types are associated with non-zero probabilities, which reduces the risk of overfitting.
Moreover, when CG is applied to \eqref{eq:obj_func_CG_l1norm}--\eqref{eq:vars_eps_CG_l1norm}, there is no need for an additional test to decide whether a column should be added to the RMP.
The initialization of the RMP associated with the minimization of predictive errors is also simpler: while an RMP derived from a maximum likelihood estimator requires an initial set of compatible consumers (at least one consumer type compatible with each historical transaction), an RMP related to error minimization can start empty.
This is due to the presence of variables $\epsilon_m^{+}$ and $\epsilon_m^{-}$, $m \in \mathcal{M}$, which act as artificial variables.

\subsection{Subproblem: dynamic programming}\label{sec:pricing_suproblem}

The CG subproblem arising from the process of estimating consumer types can be modeled as a \emph{generalized linear ordering problem} \citep[GLOP, ][]{MendezDiaz2019}.
The GLOP is NP-hard, as it generalizes the classical linear ordering problem, which is known to be NP-hard \citep{Garey1979}.
Since this paper focuses on the estimation of DCMs, we present, for the sake of conciseness, the formal definition of the GLOP together with the formulations used to solve the problem for both single- and multi-purchase consumers \citep{Vulcano2015, Lin2025} in Sections \ref{sec:customer_generation_single_purchase_case} and \ref{sec:customer_generation_multi_purchase_case} of Appendix \ref{sec:formulations_GLOP}, respectively. 
The remainder of this section is devoted to introducing a DP approach for modeling and solving the pricing subproblem arising in multi-purchase ranked-list DCMs.

For clarity, when describing the algorithm, we adopt the notation and terminology introduced in Section \ref{sec:model_basis}.
Further, we assume w.l.o.g. that the set of transactions $\mathcal{T}$ satisfies $|\mu_t| > 0$ for every $t \in \mathcal{T}$.
Transactions with null rewards can be safely ignored in the pricing subproblem, as they do not affect the computation of the reduced costs of the ranked lists.
We also define $\mathcal{T}_+ = \{t \in \mathcal{T}: \mu_t > 0\}$ and $\mathcal{T}_- = \{t \in \mathcal{T}: \mu_t < 0\}$.
For a transaction $t \in \mathcal{T}$, let $R_t = S_t \setminus B_t$.
Our algorithm is very flexible and can accommodate different aspects that may arise when defining consumer types, without requiring drastic changes in its implementation.
These extensions will be further discussed in Section \ref{sec:additional}.

\subsubsection{Label definition}

The first step in introducing the DP method is to define a label, or state, of the DP.
We define a label $L$ as a tuple $L = (v(L), N(L), \eta(L), p(L), \{\kappa_t(L): t\in\mathcal{T}\}, \mathtt{pred}(L))$.
The first component, $v(L) \in \mathcal{N}$, corresponds to the last product appended to the consumer's preference list.
The second component, $N(L) \subseteq \mathcal{N}$, denotes the set of products already appended to the consumer's preference list.
The third component, $\eta(L) \in \{1, \ldots, |\mathcal{N}|\}$, indicates the maximum number of purchases desired by the consumer.
The fourth component, $p(L)$, is the total profit collected by the consumer type.
The fifth component, $\{\kappa_t(L): t \in \mathcal{T}\}$, tracks the compatibility status of the consumer type with each transaction in $\mathcal{T}$.
A status of $-1$ is used to mark that the consumer type can no longer add or subtract the reward from transaction $t$.
Other values (between $0$ and $\eta(L) - 1$) represent intermediate states of the transaction, some of which allow rewards to be collected or subtracted. 
Finally, the last component, $\mathtt{pred}(L)$, stores the predecessor label of $L$ and is used to reconstruct the full consumer type $(\sigma, \eta)$ via backtracking.

Our algorithm is initialized with labels encoding empty consumer types for every $\eta = 1, \ldots, |\mathcal{N}|$, namely $L^{\eta}_0 = (-1, \emptyset, \eta, p_0, \{\kappa_{t\eta}^0: t\in\mathcal{T}\}, \mathtt{null})$ where:
\begin{align*}
	p_0 &= \sum_{t\in\mathcal{T}, B_t = \emptyset}\mu_t\\
	\kappa_{t\eta}^0 & = \begin{cases}
		0 & \text{if $|B_t|\leq\eta$}\\
		-1 & \text{if $|B_t| > \eta$}
	\end{cases} & t\in\mathcal{T}.
\end{align*}

\subsubsection{Label extension}

When a consumer type encoded by a label $L$ is extended to a product $j \in \mathcal{N} \setminus N(L)$, to form, say, a new label $L' \leftarrow L \oplus j$---where the operator $\oplus$ denotes appending product $j$ at the end of the consumer's preference list---we compute $\kappa_t(L')$ and $p(L')$ according to the pseudocode in Algorithm~\ref{pseudo:extension_sets}.
We then update the components of the new label as follows: 
$v(L') \leftarrow j$, 
$N(L') \leftarrow N(L) \cup \{j\}$, 
$\eta(L') \leftarrow \eta(L)$, 
and $\mathtt{pred}(L') \leftarrow L$.
\begin{algorithm}[!ht]
	\caption{Computation of $\{\kappa_t(L'):t\in\mathcal{T}\}, p(L')$ upon extending $L$ to a product $j\in\mathcal{N}\setminus N(L)$\label{pseudo:extension_sets}}
	\begin{algorithmic}[1]
		\State $\kappa_t(L') \leftarrow \kappa_t(L)$ for every $t\in\mathcal{T}$, $p(L')\leftarrow p(L)$
		\ForAll{$t\in\mathcal{T}$ such that $\kappa_t(L) \geq 0$}
		\If{$j\in R_t$}
		\If{$\kappa_t(L) \geq |B_t|$}
		\State $p(L')\leftarrow p(L') - \mu_t$
		\EndIf
		\State $\kappa_t(L')\leftarrow -1$ 
		\ElsIf{$j\in B_t$}
		\State $\kappa_t(L')\leftarrow \kappa_t(L) + 1$
		\If{$\kappa_t(L') = |B_t|$}
		\State $p(L')\leftarrow p(L') + \mu_t$\label{pseudo:extension:add_reward}
		\EndIf
		\If{$\kappa_t(L') = \eta(L)$}
		\State $\kappa_t(L')\leftarrow -1$
		\EndIf
		\EndIf
		\EndFor
		\State\Return $\{\kappa_t(L'):t\in\mathcal{T}\}, p(L')$
	\end{algorithmic}
\end{algorithm}

The interpretation of the values $\kappa_t(L)$ is as follows.
They track the number of elements from the bundle $B_t$ encountered during the successive extensions, up to the $(\eta - 1)$-\textit{th} coincidence.
When the value reaches $|B_t|$, the consumer type becomes compatible.
However, it may become incompatible again if the consumer type purchases an additional product from $R_t$.
When the number of purchases in a transaction reaches $\eta$, the transaction is marked as closed, and no further reward can be collected or subtracted.
Transactions from which no reward can be collected or subtracted are marked with the value $-1$.

\subsubsection{Dominance}

A key component of every DP algorithm is its ability to detect 
dominated
labels during execution and discard them from further extensions.
Intuitively, a label $L'$ is declared
\textit{dominated}
if there exists another label $L$ whose extensions can only lead to consumer types with higher accumulated profits.
When two labels $L$ and $L'$, possibly encoding different consumer types (different preference lists and/or different values of $\eta$), are extended over a common sequence of products $(j_1, j_2, \ldots, j_l)$, each may trigger additions or subtractions of the rewards $\mu_t$, depending on the values taken by $\kappa_t(L)$, $\kappa_t(L')$, $\eta(L)$, and $\eta(L')$.
We specify these different possibilities to introduce the associated \emph{dominance rules}.

\paragraph{Rewards that may be added to $L$ and subtracted from $L'$:}

A reward $\mu_t$ may be added to label $L$ and subtracted from $L'$ when the following three conditions hold simultaneously:
\begin{align}
	& 0 \leq \kappa_t(L) < |B_t|,\label{dom:cond:T2:1}\\
	&\kappa_t(L') \geq |B_t|,\label{dom:cond:T2:2}\\
	&\eta(L') - \kappa_t(L') + \kappa_t(L) > |B_t|.\label{dom:cond:T2:3}    
\end{align}
For every $s \in \{+, -\}$, we define $\mathcal{T}_s^2(L, L') = \{t\in \mathcal{T}_s: \text{$t, L$, and $L'$ satisfy conditions \eqref{dom:cond:T2:1}--\eqref{dom:cond:T2:3}}\}$.

\paragraph{Rewards that may be added to $L$ and not to $L'$:}

Let us consider the following three conditions involving two labels $L$, $L'$, and a transaction $t \in \mathcal{T}$:
\begin{align}
	&0\leq \kappa_t(L) < |B_t|,\label{dom:cond:T1a:1}\\
	&\kappa_t(L') = -1,\label{dom:cond:T1a:2}\\
	&\kappa_t(L') \geq |B_t|.\label{dom:cond:T1a:3}
\end{align}
A reward $\mu_t$ may be added to a label $L$ but not to a label $L'$ when condition \eqref{dom:cond:T1a:1} holds, together with either \eqref{dom:cond:T1a:2} or \eqref{dom:cond:T1a:3}.
We then define $\mathcal{T}_s^{1a}(L, L') = \{t\in\mathcal{T}_s\setminus\mathcal{T}_s^2(L, L'): \text{$t, L$, and $L'$ satisfy condition \eqref{dom:cond:T1a:1} and [\eqref{dom:cond:T1a:2} or \eqref{dom:cond:T1a:3}]}\}$.

\paragraph{Rewards that may be subtracted from $L'$ and not from $L$:}

Let us consider the following conditions involving two labels $L$ and $L'$, and a transaction $t \in \mathcal{T}$:
\begin{align}
	&\kappa_t(L')\geq |B_t|,\label{dom:cond:T1b:1}\\
	&\kappa_t(L) < |B_t|.\label{dom:cond:T1b:2}
\end{align}
For $s \in \{+, -\}$, the set of transactions in $\mathcal{T}_s$ for which the reward may be subtracted from $L'$ but not from $L$ is denoted $\mathcal{T}_s^{1b}(L, L') = \{t\in\mathcal{T}_s\setminus\mathcal{T}_s^2(L, L'): \text{$t, L, L'$ satisfy conditions \eqref{dom:cond:T1b:1}-\eqref{dom:cond:T1b:2}}\}$.

\paragraph{Dominance criterion:}

For every $s \in \{+, -\}$, we define $\mathcal{T}_s^{1}(L, L') = \mathcal{T}_s^{1a}(L, L')\cup \mathcal{T}_s^{1b}(L, L')$.
We further define:
\begin{equation}
	\Delta^s(L, L') = 2\sum_{t\in\mathcal{T}_s^2(L, L')}\mu_t + \sum_{t\in\mathcal{T}_s^{1}(L, L')}\mu_t.
\end{equation}
We say that a label $L$ dominates another label $L'$ if the following two conditions are met:
\begin{align}
	& p(L) + \Delta^-(L, L') - \Delta^+(L', L) \geq p(L'),\label{eq:dom:profit}\\
	& N(L)\subseteq N(L').\label{eq:dom:elem}
\end{align}
The following result formalizes the \textbf{dominance criterion} between two labels $L$ and $L'$.

\begin{proposition}\label{prop:dom}
	Let $L$ and $L'$ be two labels that satisfy conditions \eqref{eq:dom:profit}--\eqref{eq:dom:elem}.
	Then, every extension of $L'$ applied to $L$ leads to a consumer type with the same or higher total profit.
	In this case, we say that label $L$ dominates $L'$, denoted by $L \prec L'$.
\end{proposition}
\proof{Proof}
From the definition of $\Delta^s(L, L')$ for $s \in \{+, -\}$, it follows that the profits of the labels $L$ and $L'$ satisfy $p(L) \geq p(L')$.
The term $\Delta^-(L, L')$ represents the maximum difference in favor of $L'$ that can arise when negative rewards are added to $L$ but not to $L'$, or subtracted from $L'$ but not from $L$.
Similarly, $\Delta^+(L', L)$ represents the maximum difference, again in favor of $L'$, obtained when positive rewards are added to $L'$ but not to $L$, or subtracted from $L$ but not from $L'$.
Equation \eqref{eq:dom:profit} guarantees that these differences in favor of $L'$ are insufficient to overturn the inequality.
Finally, equation \eqref{eq:dom:elem} ensures that any extension feasible for $L'$ is also feasible for $L$.
\qed
\endproof

\subsubsection{Final algorithm}

In Algorithm \ref{pseudo:labeling_algorithm}, we provide a high-level description of the proposed DP.
It consists of three main steps:
an initialization step (Line~\ref{alg:initialization}), 
a dominance check (Line~\ref{alg:dominance}), and 
an extension step (Line~\ref{alg:extension}).
\begin{algorithm}[!ht]
	\caption{Labeling algorithm \label{pseudo:labeling_algorithm}}
	\begin{algorithmic}[1]
		\State $\mathcal{Q} = \{L_0^{\eta}: \eta = 1\ldots |\mathcal{N}|\}$, $\mathcal{U}\leftarrow\emptyset$, $L^*\leftarrow \mathtt{null}$ \label{alg:initialization}
		\While{$\mathcal{Q}\neq\emptyset$}
		\State Let $L$ be a label in $\mathcal{Q}$, and make $\mathcal{Q}\leftarrow\mathcal{Q}\setminus\{L\}$ 
		\If{$\nexists\,\, L'\in\mathcal{U}$ s.t. $L'\prec L$} \label{alg:dominance}
		\State Set $\mathcal{U}\leftarrow\mathcal{U}\cup\{L\}$
		\ForAll{$j \in \mathcal{N}\setminus N(L)$}\label{alg:2:elemext}
		\State{Let $L'\leftarrow L\oplus j$} \label{alg:extension}
		\If{$L^* = \mathtt{null}$ \textbf{or} $p(L') > p(L^*)$}
		\State Set $L^*\leftarrow L^{\prime}$ \label{pseudo:collect_best_profit}
		\EndIf
		\State Set $\mathcal{Q} \leftarrow \mathcal{Q} \cup \{L'\}$\label{pseudo:add_to_L}
		\EndFor
		\EndIf
		\EndWhile
		\State\Return $L^*$ \label{pseudo:return_best_L}
	\end{algorithmic}
\end{algorithm}

\section{Additional attributes and particular cases}\label{sec:additional}

We devote this section to introducing several relevant attributes and particular cases of multi-purchase ranked-list DCMs.
These cases can be addressed using the same DP approach introduced in Section~\ref{sec:pricing_suproblem}, with only simple and often straightforward modifications to the base algorithm.

\subsection{Addressing elementarity in the single-purchase case}\label{sec:elementarity}

We refer to the requirement that the preference list of a consumer type must not contain repeated products as \textit{elementarity}, following the terminology commonly used in vehicle routing \citep{Feillet2004exact, irnich2005shortest}.
Enforcing elementarity comes at the expense of a worst-case complexity that is exponential in the cardinality of $\mathcal{N}$, due to the dominance condition \eqref{eq:dom:elem}.
We will show that, for the single-purchase setting, the previous algorithm can be trivially simplified to address the elementarity requirement.
Let us temporarily relax the elementarity condition and assume that Algorithm~\ref{pseudo:labeling_algorithm} is modified as follows:
\begin{enumerate}
	\item Algorithm~\ref{pseudo:extension_sets} is applied without modification when a label $L$ is extended to a product $j \in N(L)$.
	\item The condition at Line~\ref{alg:2:elemext} of Algorithm~\ref{pseudo:labeling_algorithm} is extended to products $j \in N(L)$.
	\item Condition~\eqref{eq:dom:elem} is ignored when assessing dominance.
\end{enumerate}

For a label $L$ and a sequence of products $\sigma = (j_1, \ldots, j_k)$, possibly containing repetitions, let $L \oplus \sigma$ denote the label obtained by sequentially appending the products in $\sigma$ to the consumer type encoded by $L$.  
The following result formalizes the dominance relationship between two labels $L$ and $L'$.

\begin{proposition}\label{prop:dom:sp}
	Let $L$ and $L'$ be two labels satisfying condition \eqref{eq:dom:profit}.
	Then every feasible extension $\sigma = (j_1, j_2, \ldots, j_k)$ of $L'$ is also feasible for $L$ and their extensions using $\sigma$ are such that $p(L' \oplus \sigma) \leq p(L\oplus\sigma)$.
\end{proposition}

\begin{proof}
The feasibility of label $L \oplus \sigma$ is obvious, since once the elementarity requirement is dropped, appending any sequence to the end of $L$ or $L'$ remains feasible.
Moreover, as both $L$ and $L'$ satisfy condition~\eqref{eq:dom:profit}, the same reasoning used in the proof of Proposition~\ref{prop:dom} can be applied to $L$ and $L'$ to conclude that $p(L' \oplus \sigma) \leq p(L \oplus \sigma)$.
\end{proof}

\begin{proposition}
	Let $L$ be a label such that $\eta(L) = 1$.
	Then, for every $j\in N(L)$, $L'=L\oplus j$ is dominated by $L$.
\end{proposition}

\begin{proof}
From Algorithm~\ref{pseudo:extension_sets}, we observe that the first time a product $j \in S_t$ for a transaction $t \in \mathcal{T}$ is added to the consumer type' preference list, it triggers either the addition of the reward $\mu_t$ (if $B_t = \{j\}$) or its subtraction (if $j \notin B_t$).
In both cases, the value of $\kappa_t$ becomes $-1$.  
When the product $j$ is added a second time, all transactions $t \in \mathcal{T}$ with $j \in S_t$ are also associated with $\kappa_t = -1$.
It follows that the resulting label $L'$ is trivially dominated by its parent label $\mathtt{pred}(L')$, since they satisfy the conditions of Proposition~\ref{prop:dom:sp}.
\end{proof}

These two results imply that one can simply drop the elementarity requirement in the dominance rule and prevent any extension of a label $L$ to products already in $N(L)$.
Unfortunately, this result does not extend to the multi-purchase case, since transactions may not close (i.e., change from $\kappa_t(L) \geq 0$ to $\kappa_t(L') = -1$) upon adding an item once to the preference list of the consumer type.

\subsection{Consumer types with small consideration sets}\label{sec:limited_lists}

The proposed modeling approach can also handle the setting in which all consumer types must have a preference list of size at most $q$, with $1 \leq q < n$.
This condition is relevant in situations where consumers have small consideration sets, such as environments in which brand reputation or habitual preferences lead them to focus on a limited subset of well-known or trusted alternatives \citep{Feldman2019}.

The adaptations required apply to both the extension and dominance rules, as described next.
We add an additional field to the label definition, namely $\ell(L)$, which is initially set to zero at the initialization step.
If $\ell(L) = q$, then no further extensions are allowed.
Otherwise, for the resulting label $L' = L \oplus j$, we set $\ell(L') \leftarrow \ell(L) + 1$.
The dominance rule in Proposition~\ref{prop:dom} must also be extended by including the following necessary condition for a label $L$ to dominate another label $L'$:
$\ell(L) \leq \ell(L')$.
Since the list size is limited, shorter lists are favored, as they allow more potential extensions.

\section{Accelerating the dynamic programming}\label{sec:acceleration_techiniques}

We devote this section to introducing several techniques aimed at accelerating the execution of the DP.
We consider completion bounds (Section \ref{sec:completion_bounds}), unreachable products (Section \ref{sec:unreachable_prodcts}), and heuristic pricing (Section \ref{sec:heuristic_pricing}).

\subsection{Completion bounds}\label{sec:completion_bounds}

Despite the application of dominance rules in our algorithm, the number of generated labels can still be very large.  
To reduce the number of extensions performed by the labeling algorithm, and hence speed up the method, we apply completion bounds to discard non-promising extensions, a common feature of modern solvers from the vehicle routing literature \citep{Baldacci2008, irnich2010path, martinelli2014efficient, pecin2017improved}.  
Before extending a label $L$ to products $j \in \mathcal{N}\setminus N(L)$, we check whether the rewards associated with the transactions in  
$\mathcal{T}_+'(L) = \{t \in \mathcal{T}_+ : 0 \leq \kappa_t(L) < |B_t|\}$ and  
$\mathcal{T}_-'(L) = \{t \in \mathcal{T}_- : \kappa_t(L) \geq |B_t|\}$  
are sufficient to produce a consumer type better than the incumbent.  

Let $\mathtt{LB}$ denote the total profit collected by the best consumer type constructed so far by the labeling algorithm.  
If
\begin{equation}
	p(L) + \sum_{t \in \mathcal{T}_+'(L)} \mu_t - \sum_{t \in \mathcal{T}_-'(L)} \mu_t \leq \mathtt{LB},
\end{equation}
then label $L$ can be discarded, since any further extensions will no longer collect sufficient rewards to generate a consumer type with a profit higher than $\mathtt{LB}$.

\subsection{Unreachable products}\label{sec:unreachable_prodcts}

Assume for a moment that we can determine, for a given label $L$ and product $j \in \mathcal{N} \setminus N(L)$, that appending the product $j$ to $L$ at any point in the future would only lead to a worsening of the total profit.  
In this case, we can mark product $j$ as unreachable and prevent any future extension of $L$ from appending $j$ to its associated preference list.  
Let us denote
$\mathcal{T}_1(L, j) = \{t \in \mathcal{T} : j \in B_t, \, 0 \leq \kappa_t(L) < |B_t|\}$ and  
$\mathcal{T}_2(L, j) = \{t \in \mathcal{T}_- : j \in R_t, \, \kappa_t(L) \geq |B_t|\}$.
If $\mathcal{T}_1(L, j) = \mathcal{T}_2(L, j) = \emptyset$, then product $j$ can indeed be marked as unreachable from $L$.
Let $U(L) = \{j\in\mathcal{N}\setminus N(L): \mathcal{T}_1(L, j) = \mathcal{T}_2(L, j) = \emptyset\}$.
Note that a product marked as unreachable never loses this status, and therefore any label extension $L \oplus j$ satisfies $U(L) \subseteq U(L \oplus j)$.
We then use the unreachable set to:  
(i) prevent any extensions to products in $N(L) \cup U(L)$; and  
(ii) strengthen pairwise dominance by replacing equation~\eqref{eq:dom:elem} with  
\begin{equation}
	N(L) \cup U(L) \subseteq N(L') \cup U(L').
\end{equation}  
This procedure follows similar principles to the concept of \textit{unreachable nodes} from the vehicle routing literature \citep{Feillet2004exact}.

\subsection{Heuristic pricing}\label{sec:heuristic_pricing}

Despite the efficiency of the proposed DP algorithm and the speedups provided by the acceleration techniques described so far, the use of exact DP algorithms may still be time-consuming.
Since exact methods are only required in the context of CG to prove the optimality of a given solution, fast and effective heuristics can be employed to generate columns in the initial iterations of the CG algorithm.
In fact, depending on their effectiveness, exact methods may not even be needed when applying CG to the estimation of DCMs, as suggested by \citet{bertsimas2016data}.
The purpose of using pricing heuristics is to reduce the number of calls to exact DP algorithms, often leading to substantial decreases in total computational time.  

Several heuristic strategies for solving the GLOP can be applied.
For instance, \citet{MendezDiaz2019} describe two greedy heuristics and three MILP-based heuristics.
Another common approach is to adapt DP algorithms, either by relaxing dominance rules, heuristically reducing the network size, or limiting the number of labels associated with specific items and states (a strategy known as bucket pruning), a practice widely adopted in the vehicle routing literature \citep{Costa2019}.
In our approach, we specifically implemented DP-based bucket pruning heuristics that restrict the number of labels per item and state to 2 or 5.
Preliminary results show that this method achieves a better trade-off between solution quality and computational time compared to the heuristics of \citet{MendezDiaz2019}.
Moreover, the DP-based heuristic builds directly on the exact DP implementation, eliminating the need for additional development effort while still generating good-quality solutions.  

Our experiments also highlight differences between estimation frameworks.
Under maximum likelihood estimation, heuristic pricing did not considerably benefit the CG algorithm, likely because the re-optimization step is already very fast due to the EM method.
In contrast, in the context of $\ell_1$-error estimation, heuristic pricing proved highly beneficial: the CG algorithm starts with heuristic pricing until no further consumer types can be generated, and only then calls the exact DP.
Nevertheless, heuristic pricing remains useful under maximum likelihood estimation to generate LBs for the completion bounds.
In some cases, the exact DP requires significant time before producing any valid LB, during which the cardinality of $\mathcal{Q}$ may grow excessively and increase CPU time.
To circumvent this, we run heuristic pricing prior to the CG step to generate valid LBs from the early stages of the labeling algorithm.

\section{Computational experiments} \label{sec:computational_experiments}

We have conducted comprehensive computational experiments to assess the performance of the proposed methodology. The algorithms were implemented in C++20 and compiled with GCC 14.2.0 using CMake 3.29.5.
To maintain flexibility in the implementation and to facilitate reproducibility, mathematical formulations, when required, were implemented using the MathOpt interface from Google OR-Tools.
For benchmarking purposes, we used Gurobi 11 as the linear and MILP solver, as it represents a state-of-the-art commercial solver and provides a strong baseline for comparison.
However, open-source solvers can be employed simply by specifying them via the command line.
All experiments were conducted on a computing grid composed of 11 Dell PowerEdge R740 machines, each equipped with 512 GB of RAM and two Intel Xeon Gold 6258R CPUs @ 2.70GHz (56 cores total per machine).

The section is organized as follows.
In Section~\ref{sec:instances}, we provide a brief description of the instances considered in our experiments.
In Section~\ref{sec:comp_DP_MIP}, we evaluate the performance of the proposed DP algorithm against a MILP solver.
In Section~\ref{sec:benchmarking_lin}, we compare the performance of our approach with the method proposed by \citet{Lin2025}, which addresses a problem with similar characteristics.
In Section~\ref{sec:estimation_results_MP}, we assess the estimation accuracy and revenue generation capabilities of the multi-purchase model.
In Section~\ref{sec:small_consideration}, we analyze the performance of our approach when dealing with small consideration sets.
Finally, in Section~\ref{sec:real_semisynthetic}, we evaluate the estimation performance of the method on real and semi-synthetic instances.

{
	\color{black}
	
	\subsection{Considered datasets}\label{sec:instances}
	
	In this section, we briefly describe the datasets used in our computational experiments.
	We consider both synthetic and real preference-based datasets.
    A more detailed description of the instances can be found in Appendix \ref{sec:instances_appendix}, and all instances used in the experiments are publicly available for reproducibility purposes at 
	\url{https://github.com/lucianoccosta/single-multi-purchase-cg}.
    
	\subsubsection{Synthetic datasets}\label{sec:syntetic_dataset}
	
	We consider four families of synthetic datasets, each designed to capture different aspects of consumer choice behavior.
	
	\begin{description}
		
		\item[\texttt{SinglePurchase}]
		This dataset, introduced by \citet{Berbeglia2022}, focuses on single-purchase behavior and is divided into \emph{random} (10x1) and \emph{structured} (5x3) instances.
		Random instances assume non-hierarchical product preferences, whereas structured instances impose a natural ordering among products (e.g., price levels).
		Across both categories, the datasets vary with respect to
		the number of products ($|\mathcal{N}|$),
		the number of consumer types in the ground truth (\textit{List}),
		market-share levels (\textit{Share}),
		and number of periods ($\mathcal{P}$).
		These instances serve as a benchmark for evaluating estimation performance in classical single-purchase ranked-list choice models.
		
		\item[\texttt{MultiPurchaseProbit}]
		This dataset is based on the synthetic instances proposed by \citet{Lin2025} and captures consumer behavior where multiple products may be purchased in the same transaction.
		\tocheck{Transactions are generated from consumers whose utilities are sampled from a normal distribution (\textit{Probit model}) and who select at most $\eta = 2$ products from a randomly generated offer set.}
		The instances vary primarily with respect to the number of products ($|\mathcal{N}|$).
		A fixed number of in-sample and out-of-sample transactions is considered for all instances.
		By using this dataset, we are able to compare our method with that of \citet{Lin2025}.
		
		\item[\texttt{MultiPurchaseRankedList}]
		To further investigate the impact of explicitly modeling multi-purchase behavior, we generate an additional family of synthetic instances that allow consumers to purchase more than two products ($\eta \geq 2$) in a single transaction.
		These instances are defined through ground-truth consumer types characterized by preference rankings and heterogeneous purchase capacities.
		They vary with respect to the number of products ($|\mathcal{N}|$),
		market-share levels (\textit{Share}),
		number of consumer types (\textit{List}),
		maximum purchase capacity ($\eta$), and
		number of periods ($\mathcal{P}$).
		This dataset is designed to assess how increasing the flexibility of multi-purchase behavior affects estimation accuracy and revenue performance.
		
		\item[\texttt{LimitedList}]
		This dataset, based on the instance-generation procedure of \citet{Feldman2019}, focuses on settings where consumers have small consideration sets.
		Preference lists are constructed from vertically differentiated products and are subject to random removals and local perturbations, yielding limited-size ranked lists.
		The instances vary according to the number of consumer types, the strength of substitution effects, and the number of observed transactions.
		This dataset is used to evaluate the impact of explicitly constraining preference-list sizes during model estimation.
		
	\end{description}
	
	\subsubsection{Datasets derived from real preferences}
	
	In addition to synthetic data, we also consider two datasets derived from real consumer preferences.
	\begin{description}
		\item[\texttt{hotels}]
		This dataset is constructed from hotel booking data originally analyzed by \citet{Bodea2009} and subsequently processed by \citet{Berbeglia2022}.
		Products correspond to room types, and transactions reflect observed booking decisions augmented with synthetic no-purchase observations.
		The instances vary across hotels with different numbers of products ($\mathcal{N}$) and transaction volumes ($\mathcal{T}$) and contain both in-sample and out-of-sample data.
		These datasets allow us to assess model performance in a real-world setting where no ground-truth model is available.
		
		\item[\texttt{sushi}]
		The sushi dataset is based on preference rankings collected by \citet{Kamishima2003}, where consumers provide complete rankings over a fixed set of products.
		Synthetic transactions are generated from these rankings under different truncation schemes, capturing varying degrees of consideration-set behavior.
		The instances vary with respect to the number of periods ($|\mathcal{P}|$) and the structure of the preference lists (full and truncated rankings).
		This dataset provides a semi-synthetic benchmark where ground-truth preferences are known.
		
	\end{description}
}

\subsection{Comparison between the DP and MILP solvers}\label{sec:comp_DP_MIP}

In this section, we report in Tables~\ref{tab:dp_vs_mip:singlep}--\ref{tab:dp_vs_mip:multip} the (time) performance of our DP algorithm compared with a commercial MILP solver when solving the formulations presented in Appendix~\ref{sec:formulations_GLOP} for both the single- and multi-purchase settings.
Our DP algorithm applies completion bounds and unreachable products simultaneously as acceleration techniques.
Preliminary experiments, reported in Appendix~\ref{sec:dp:accel}, show that combining both mechanisms yields the best overall performance.
In this experiment, we consider the \texttt{SinglePurchase} and the \texttt{MultiPurchaseRankedList} datasets, both detailed in Appendix \ref{sec:instances_appendix}.
Specifically for the multi-purchase setting, showing all possible combinations of attributes would result in excessively large tables.
Therefore, we report aggregate results in this section, while more detailed results for the different parameter combinations are provided in Appendix~\ref{sec:detailled_multipurchase_setting}.

We report the shifted geometric means of the computing times (in seconds), using a shift value of $\sigma = 1$, required to solve the estimation problems.
The shifted geometric mean of a vector $\mathbf{x} = (x_1, \ldots, x_n)\geq 0$ and a shift value of $\sigma > 0$ is computed as
\begin{equation}\label{eq:shifted_geomean}
	sgm(\mathbf{x}, \sigma) = \sqrt[n]{\prod_{i = 1}^n (x_i + \sigma)} - \sigma.
\end{equation}
The geometric mean is preferred because it provides a more representative summary of computing times that vary widely across instances, reducing the influence of extreme values.

We restrict our analysis to dataset configurations where at least one of the methods did not time out in at least one problem. In addition, for the single-purchase case, we restrict our analysis to the configurations where the mean computing times across both algorithms took $> 60$ seconds.
Timeouts are given 3 hours for the purpose of computing the geometric means.
We highlight in boldface the method that achieves the best performance (i.e., the smallest mean computing time across all problems from a given dataset configuration).
Because the $\ell_1$-error model may yield negative dual values during CG, it cannot be used with the MILP pricing formulation, which assumes non-negative rewards (duals).
Therefore, in this experiment, we only report results for the EM-based estimation approach.

\begin{table}[!hbtp]
	\caption{Runtime comparison between the DP and MILP approaches for estimating the single-purchase problems}
	\label{tab:dp_vs_mip:singlep}
	\begin{minipage}[t]{0.5\linewidth}
		\centering
		\begin{tabular}{ccccrr}
			\textit{Share} & \textit{Structure} & \textit{List} & $|\mathcal{P}|$ & DP & MILP\\ \hline
			\multirow{8}{*}{20\%} & \multirow{3}{*}{10x1} & 10  & 600 & \textbf{1.7} & 218.6 \\
			&                       & 100 & 300 & \textbf{0.9} & 137.1 \\
			&                       & 100 & 600 & \textbf{2.9} & 1701.3 \\
			\cline{2-6}
			& \multirow{5}{*}{5x3}  & 10  & 300 & \textbf{37.4} & 198.8 \\
			&                       & 10  & 600 & \textbf{120.5} & 2660.5 \\
			&                       & 100 & 150 & \textbf{16.2}  & 61.1 \\
			&                       & 100 & 300 & \textbf{57.9}  & 1624.4 \\
			&                       & 100 & 600 & \textbf{190.2} & 10800.0 \\
			\hline
			\multirow{8}{*}{50\%} & \multirow{5}{*}{10x1} & 10  & 300 & \textbf{1.2} & 130.8 \\
			&                       & 10  & 600 & \textbf{3.2} & 728.9 \\
			&                       & 100 & 150 & \textbf{0.6} & 111.0 \\
			&                       & 100 & 300 & \textbf{1.8} & 1138.4 \\
			&                       & 100 & 600 & \textbf{4.9} & 10800.0 \\
			\cline{2-6}
			& \multirow{3}{*}{5x3}  & 10 & 150 & \textbf{21.4} & 103.2 \\
			&                      & 10 & 300 & \textbf{67.6} & 1351.5 \\
			&                      & 10 & 600 & \textbf{167.2} & 6285.7 \\
			\hline
		\end{tabular}
	\end{minipage}%
	\begin{minipage}[t]{0.5\linewidth}
		\centering
		\begin{tabular}{ccccrr}
			\textit{Share} & \textit{Structure} & \textit{List} & $|\mathcal{P}|$ & DP & MILP\\ \hline
			\multirow{3}{*}{50\%} & \multirow{3}{*}{5x3} & 100 & 150 & \textbf{34.3} & 915.3 \\
			&                      & 100 & 300 & \textbf{93.5} & 10800.0 \\
			&                      & 100 & 600 & \textbf{270.7} & 10800.0 \\
			\hline
			\multirow{12}{*}{90\%} & \multirow{5}{*}{10x1} & 10  & 300 & \textbf{1.5} & 203.0 \\
			&                      & 10  & 600 & \textbf{4.3} & 1015.6 \\
			&                      & 100 & 150 & \textbf{1.0} & 265.8 \\
			&                      & 100 & 300 & \textbf{2.6} & 1894.2 \\
			&                      & 100 & 600 & \textbf{9.5} & 10800.0 \\
			\cline{2-6}
			& \multirow{7}{*}{5x3}  & 10  & 150 & \textbf{35.0} & 381.6 \\
			&                      & 10  & 300 & \textbf{109.7} & 4634.1 \\
			&                      & 10  & 600 & \textbf{273.4} & 10160.1 \\
			&                      & 100 & 75  & \textbf{19.0} & 196.0 \\
			&                      & 100 & 150 & \textbf{61.6} & 6497.8 \\
			&                      & 100 & 300 & \textbf{170.0} & 10800.0 \\
			&                      & 100 & 600 & \textbf{498.8} & 10800.0 \\
			\hline
			&&&&&
		\end{tabular}
	\end{minipage}%
\end{table}

\begin{table}[!hbtp]
	\caption{Runtime comparison between the DP and MILP approaches for estimating the multi-purchase problems}
	\label{tab:dp_vs_mip:multip}
	\begin{minipage}[t]{0.5\linewidth}
		\centering
		
		\begin{tabular}{cccrr}
			$|\mathcal{N}|$   & \textit{Share} & $\eta$ & DP & MILP\\ \hline
			\multirow{12}{*}{10} & \multirow{4}{*}{20\%} & 2 & \textbf{4.3} & 440.1 \\
			& & 3 & \textbf{7.9} & 227.5 \\
			& & 4 & \textbf{11.2} & 72.8 \\
			& & 5 & \textbf{15.5} & 22.7 \\
			\cline{2-5}
			& \multirow{4}{*}{50\%} & 2 & \textbf{6.9} & 842.5 \\
			& & 3 & \textbf{12.9} & 578.7 \\
			& & 4 & \textbf{17.5} & 202.0 \\
			& & 5 & \textbf{25.8} & 62.8 \\
			\cline{2-5}
			& \multirow{4}{*}{80\%} & 2 & \textbf{9.5} & 1065.2 \\
			& & 3 & \textbf{19.1} & 888.0 \\
			& & 4 & \textbf{24.9} & 382.7 \\
			& & 5 & \textbf{34.8} & 90.8 \\
			\hline
		\end{tabular}
		
	\end{minipage}%
	\begin{minipage}[t]{0.5\linewidth}
		\centering
		
		\begin{tabular}{cccrr}
			$|\mathcal{N}|$  & Share & $\eta$ & DP & MILP\\ \hline
			\multirow{12}{*}{15} & \multirow{4}{*}{20\%} & 2 & \textbf{62.1} & 10701.6 \\
			&  & 3 & \textbf{158.1} & 5677.5 \\
			&  & 4 & \textbf{265.7} & 1623.6 \\
			&  & 5 & \textbf{443.0} & 742.3 \\
			\cline{2-5}
			& \multirow{4}{*}{50\%} & 2 & \textbf{109.5} & 10800.0 \\
			&  & 3 & \textbf{271.3} & 8376.4 \\
			&  & 4 & \textbf{422.6} & 4743.3 \\
			&  & 5 & \textbf{683.0} & 2277.2 \\
			\cline{2-5}
			& \multirow{4}{*}{80\%} & 2 & \textbf{134.3} & 10313.0 \\
			&  & 3 & \textbf{345.8} & 8991.0 \\
			&  & 4 & \textbf{514.7} & 5621.1 \\
			&  & 5 & \textbf{744.5} & 3462.5 \\
			\hline
		\end{tabular}
		
	\end{minipage}%
\end{table}

We observe that the performance of the DP dominates that of the MILPs both in the single- and in the multi-purchase cases, resulting in computing times that are on average shorter by an order of magnitude. Although the MILP seems to scale better as $\eta$ increases, it never gets to the point of becoming competitive with the DP.

{
	\color{black}
	
	\subsection{Performance benchmarking of multi-purchase methods}\label{sec:benchmarking_lin}
	
	In this section, we compare the performance of our method against that obtained by the approach of \citet{Lin2025}.
	Their method follows a market-discovery framework in which new consumer types are generated by solving a MILP-based pricing problem that incorporates BRP rules.
	In this experiment, we consider the instance set \texttt{MultiPurchaseProbit} introduced by \citet{Lin2025}.
	The instances were obtained
	from the authors and are detailed in Appendix \ref{sec:instances_appendix}.
	We compare the performance of our method against that of \citet{Lin2025} using the same error metric used in their work, namely an adapted version of the \textit{hard}-RMSE.
	The metric is adapted because, although the data generation process allows for multi-purchase behavior, the evaluation is performed by considering the marginal contribution of single-product purchases.
	We refer to this metric as the \emph{Marginal-RMSE} (MRMSE).
	Appendix~\ref{sec:performance_metrics} provides details on how this metric is computed, and discusses alternative error measures.
	
	In Table~\ref{tab:comparision_lim_rmse}, we compare the performance of all the methods that we implemented against that of \citet{Lin2025} with respect to MRMSE, aggregating results for instances with the same number of products ($|\mathcal{N}|$).
	We report results for both configurations of our method:
	(i) solving the pricing subproblem using the proposed DP algorithm combined with a maximum-likelihood master problem solved using the EM method (\texttt{EM-DP}), and
	(ii) solving the pricing subproblem using DP combined with an $\ell_1$-error master problem (\texttt{L1-DP}).
	Also, for the sake of completeness, we also consider a configuration that combines a maximum-likelihood master problem,  solved using the EM method, with an MILP-based CG subproblem (\texttt{EM-MIP}), which can be viewed as a proxy for the method of \citet{Lin2025}.
	The best-performing method in each group is highlighted in bold.
	We observe that, except for instances with 30 products, the methods relying on the proposed DP algorithm yield the smallest errors.
	Yet, we point out that even when the DP-based method did not entail the best results, the difference in favor of \citet{Lin2025}'s method is very small.
	Furthermore, it is also worth noting that, even though the method of \citet{Lin2025} relies on richer behavioral information through BRP rules, our approach achieves similar or better predictive accuracy.
	A possible explanation is that the additional flexibility introduced by BRP rules may increase the risk of overfitting.
	Since the evaluation is conducted using out-of-sample transactions, this can limit the ability of the model to generalize and accurately predict consumer behavior.
	
	\begin{table}[htbp]
		\centering
		\caption{Performance comparison for different numbers of products with respect to MRMSE}
		\label{tab:performance_num_products}
		\begin{tabular}{ccccc}
			\toprule
			$|\mathcal{N}|$ & \textbf{\texttt{EM-DP}} & \textbf{\texttt{L1-DP}} & \textbf{\texttt{EM-MIP}} & \citet{Lin2025} \\
			\midrule
			5  & 0.4216 & \textbf{0.4134} & 0.4319 & 0.4212 \\
			10 & \textbf{0.3775} & 0.3881 & 0.4049 & 0.3803 \\
			15 & \textbf{0.3121} & 0.3996 & 0.3300 & 0.3162 \\
			20 & \textbf{0.2835} & 0.3596 & 0.3030 & 0.2904 \\
			25 & \textbf{0.2732} & 0.3490 & 0.2844 & 0.2733 \\
			30 & 0.2475 & 0.3047 & 0.2570 & \textbf{0.2474} \\
			\bottomrule
		\end{tabular}
		\label{tab:comparision_lim_rmse}
	\end{table}
	
	In addition to analyzing prediction errors (MRMSE), we report in Table~\ref{tab:comparision_cpu_multipurchase} the computing times required by \texttt{EM-DP}, \texttt{L1-DP}, and \texttt{EM-MIP}.
	We do not report the computing times of \citeauthor{Lin2025}'s method, as they were not provided in their work.
	However, since their approach relies on an MILP solver to generate consumer types, our \texttt{EM-MIP} implementation can be regarded as a good proxy for their method.
	The results show that the DP-based methods are able to estimate multi-purchase DCMs within reduced computing times.
	In particular, \texttt{L1-DP} solves all instances in less than 6 seconds on average.
	When considering \texttt{EM-DP}, the method is able to solve all instances in slightly more than 100 seconds, on average, while still yielding the lowest prediction errors for most instance groups.
	In contrast, \texttt{EM-MIP} is scalable only for instances with a small number of products ($n \leq 10$).
	For instances with $n = 15$, the method requires more than 1.5 hours to estimate the models.
	For instances with $n \geq 20$, the method reaches the 2-hour time limit imposed for this set of experiments.
	The results show that the DP-based methods provide a good balance between solution quality and computational time, making them suitable for practical applications.
	
	\begin{table}[htbp]
		\centering
		\caption{Computational time comparison for different numbers of products}
		\label{tab:time_num_products}
		\begin{tabular}{cccc}
			\toprule
			$|\mathcal{N}|$ & \textbf{\texttt{EM-DP}} & \textbf{\texttt{L1-DP}} & \textbf{\texttt{EM-MIP}} \\
			\midrule
			5  & 0.07   & \textbf{0.05} & \textbf{0.05}   \\
			10 & \textbf{0.42}   & 0.47 & 161.37 \\
			15 & 1.36   & \textbf{0.29} & 6069.46 \\
			20 & 5.64   & \textbf{0.79} & 7200.00 \\
			25 & 33.05  & \textbf{3.52} & 7200.00 \\
			30 & 100.66 & \textbf{5.14} & 7200.00 \\
			\bottomrule
		\end{tabular}
		\label{tab:comparision_cpu_multipurchase}
	\end{table}
	
	Finally, we assess the ability of the estimated models to produce revenue in an assortment optimization context.
	Specifically, for both \texttt{EM-DP} and \texttt{L1-DP}, we solve the \textit{associated assortment optimization} (AAO) problem using an adapted version of the formulation by \citet{Lin2025}, which is presented in Appendix \ref{sec:milp_assortement_recommendation}.
	Table \ref{sec:revenue_assort} reports the average AAO by our method in comparison with the approach proposed by \citeauthor{Lin2025}.
	Then, for each resulting assortment, the average revenue is estimated through a Monte Carlo simulation, following a procedure similar to that described by \citet{Lin2025}.
	For clarity, the complete procedure is described in Appendix \ref{sec:monte_carlo_simulation}.
	Overall, the results show that, except for the group of instances with $n=20$, the DP-based methods achieve higher average revenues than those reported by \citet{Lin2025}.
	
	\begin{table}[htbp]
		\centering
		\caption{Average AAO comparison for different numbers of products}
		\label{tab:avg_revenue_num_products}
		\begin{tabular}{cccc}
			\toprule
			$|\mathcal{N}|$
			& \textbf{\texttt{EM-DP}} 
			& \textbf{\texttt{L1-DP}} 
			& \textbf{\citet{Lin2025}} \\
			\midrule
			5  & 5.0199 & \textbf{5.8343} & 3.9944 \\
			10 & 5.5187 & \textbf{6.1368} & 4.4519 \\
			15 & 5.1010 & \textbf{5.1989} & 4.9827 \\
			20 & 4.1457 & 4.1786 & \textbf{4.8262} \\
			25 & \textbf{5.1269} & 3.1724 & 4.3844 \\
			30 & \textbf{5.7982} & 3.3740 & 4.5456 \\
			\bottomrule
		\end{tabular}
		\label{sec:revenue_assort}
	\end{table}
}

\subsection{The value of multi-purchase modeling}\label{sec:estimation_results_MP}

In this experiment, we consider problems from the \texttt{MultiPurchaseRankedList} dataset, described in Appendix \ref{sec:instances_appendix}, with ground-truth models containing consumer types $i$ such that
$1 \le \eta(i) \le 5$.
The objective of this experiment is to assess the importance of capturing this behavior via ad-hoc (i.e., multi-purchase) models.
To that end, we consider five possible settings for the estimation models, namely, restricting consumer types $i$ such that $1 \le \eta(i) \le \eta$, for $\eta \in \{1,2,3,4,5\}$.
We use the proposed DP algorithm as the engine to enumerate consumer types, and evaluate the performance under the two estimation models considered in this work, namely the EM and the $\ell_1$-error models.
We restrict our analysis to the problems in which all settings are able to produce a final choice model, either because optimality is attained (for the $\ell_1$-error model) or because the statistical test is triggered (for the EM model).


We summarize our results in Table~\ref{tab:multip}.
For each value of $\eta \in \{1, 2, 3, 4, 5\}$, we report shifted geometric means of the computing times (in seconds) required by the estimation models, using a shift value of $\sigma_{CPU} = 1$,
as well as shifted geometric means of both \textit{soft-RMSE} (SRMSE) \tocheckTwo{and \textit{marginal-RMSE} (MRMSE) computed with a shift value of $\sigma_{SRMSE} = \sigma_{MRMSE} = 0.01$.}
The SRMSE provides a measure of how close a model is to the ground truth and is the best way to assess model strength, as it is agnostic to the historical transactions.
This contrasts with the \textit{hard}-RMSE (HRMSE), which evaluates a model's ability to explain a fixed set of out-of-sample transactional data, particularly useful in data-driven contexts where no ground truth model is available.
\tocheckTwo{With respect to MRMSE, while SRMSE provides an appropriate metric when the maximum number of purchases allowed by the estimated model matches that observed in the data, it may become less informative when evaluating simpler models that restrict the number of purchases, as in the experiments reported here. 
	For this reason, we also report MRMSE in addition to SRMSE to allow a fair comparison across models with different purchase-capacity assumptions.}
More details about these metrics are presented in Appendix~\ref{sec:performance_metrics}.
We also evaluate the ability of the estimated models to capture revenue in an assortment optimization context.
Similar to what was done in Section \ref{sec:benchmarking_lin}, we solve the AAO problem using an adapted version of the formulation by \citet{Lin2025}, which is presented in Appendix \ref{sec:milp_assortement_recommendation}.
The resulting assortment is then evaluated under the ground-truth model to compute an empirical revenue $\rho_e$, which is compared to the theoretical optimal revenue $\rho_t$, obtained by assuming full knowledge of the ground truth.
For each instance, we compute the ratio $\rho_e/\rho_t \in [0,1]$ and report shifted geometric means with a shift value of $\sigma_{AAO} = 0.01$.

{
	\color{blue}
	
	\begin{table}[!hbtp]
		\caption{Estimation and assortment optimization on multi-purchase problems}
		\label{tab:multip}
		\centering
		\begin{tabular}{crrrrrrrr}
			\multirow{2}{*}{$\eta$} & \multicolumn{4}{c}{$\ell_1$-error} & \multicolumn{4}{c}{EM}\\
			\cline{2-5} \cline{6-9}\\
			& CPU & SRMSE & MRMSE & AAO & CPU & SRMSE & MRMSE & AAO\\ \hline
			1 & 0.6 & 0.0212 & 0.3518 & 0.8369 & 1.9 & 0.0264 & 0.3555 & 0.9119\\
			2 & 11.5 & 0.0106 & 0.3395 & \textbf{0.9417} & 4.4 & 0.0188 & 0.3451 & 0.8480\\
			3 & 22.8 & 0.0080 & 0.3362 & 0.9388 & 11.9 & 0.0120 & 0.3384 & 0.8423\\
			4 & 30.6 & 0.0066 & 0.3350 & 0.9301 & 25.4 & 0.0089 & 0.3354 & 0.8367\\
			5 & 42.9 & \textbf{0.0060} & 0.3348 & 0.9227 & 37.8 &  0.0069 & \textbf{0.3347} & 0.8333\\
			\hline
		\end{tabular}
	\end{table}
}

\tocheckTwo{We observe that, in this experiment, the $\ell_1$-error model yields the best results in terms of SRMSE and AAO, although its performance with respect to MRMSE is very close to that of the EM model.}
As expected, for both the $\ell_1$-error and EM models, the best accuracy is obtained when considering $\eta = 5$.
However, estimating the model with smaller values of $\eta$, i.e., $\eta \in \{3, 4\}$, still yields acceptable accuracies.
These results indicate that the proposed models are capable of capturing the essence of multi-purchase behavior.
When analyzing the revenue obtained from solving the instances, we observe that the highest revenue is attained with $\eta = 2$ under the $\ell_1$-error model.
For the EM method, the highest revenue corresponds to $\eta = 1$.
However, this value is already smaller than almost all revenues obtained under the $\ell_1$-error model, except when using $\eta = 1$, indicating that the EM method may be more sensitive to model specification in this setting.
We also emphasize that the computational effort increases substantially, as much as by orders of magnitude, as the maximum purchase size grows, particularly when moving from single-purchase models to settings where $\eta = 2$.

Finally, we highlight that our results are consistent with those reported by \citet{Wang2023}.
In their work, the authors argue that modeling consumers as purchasing at most two products is already sufficient to capture multi-purchase behavior in practice.
Our experiments support this insight.
Even though the ground-truth model assumes that consumers may purchase up to five products, the best revenue performance is obtained when setting $\eta = 2$ (under $\ell_1$-error).
Interestingly, both $\eta = 3$ and $\eta = 4$ also yield revenues higher than those associated with $\eta = 5$.
These findings corroborate that, in many cases, a less complex model can effectively capture multi-purchase behavior, even in environments where consumers exhibit higher multi-purchase capacity.

\subsection{Performance in settings with small consideration sets}\label{sec:small_consideration}

In this section, we assess the performance of our method and the modeling accuracy when consumers are subject to small consideration sets.
We consider five different values for parameter $q$ as defined in Section \ref{sec:limited_lists}, namely $q \in \{4, 5, 7, 10, +\infty\}$.
We restrict our analysis to the instances from the \texttt{LimitedList} dataset, which is detailed in Appendix \ref{sec:instances_appendix}.
In this experiment, we make use of the proposed DP to enumerate consumer types during the CG process.

In Table~\ref{tab:limlists}, we report the same type of summary results as in Table~\ref{tab:multip}.
For each value of $q \in \{4, 5, 7, 10, +\infty\}$, we report shifted geometric means of the computing times (in seconds) required by the estimation models, using a shift value of $\sigma_{CPU} = 1$, as well as shifted geometric means of the \textit{soft-RMSE} (SRMSE), computed with a shift value of $\sigma_{SRMSE} = 0.01$.
To ensure a fair comparison, all estimation models were re-run with a six-hour time limit.
We restrict our analysis to those problem classes for which no time-out occurred across any of the ten settings, ensuring that the reported SRMSE and AAO values are representative of each model's performance.

\begin{table}[!hbtp]
	\caption{Estimation and assortment optimization with small consideration sets}
	\label{tab:limlists}
	\centering
	\begin{tabular}{crrrrrr}
		\multirow{2}{*}{$q$} & \multicolumn{3}{c}{$\ell_1$-error} & \multicolumn{3}{c}{EM}\\
		\cline{2-4} \cline{5-7}\\
		& CPU & SRMSE & AAO & CPU & SRMSE & AAO\\ \hline
		4 & 11.7 & 0.2583 & 0.4523 & 108.3 & 0.0291 & 0.9674\\
		5 & 23.2 & 0.2582 & 0.4768 & 614.2 & \textbf{0.0267} & \textbf{0.9679}\\
		7 & 22.2 & 0.2562 & 0.4888 & 2331.4 & 0.0270 & 0.9466\\
		10 & 23.0 & 0.2573 & 0.4997 & 2126.0 & 0.0281 & 0.9144\\
		$+\infty$ & 19.9 & 0.2574 & 0.503 & 2089.4 & 0.0285 & 0.8712\\
		\hline
	\end{tabular}
\end{table}

We observe that the EM method captures the revenue stream with very high accuracy, especially for low values of $q$ (namely, $q = 4, 5$).
We also note that this accuracy does not align perfectly with the SRMSE, as the model estimated with $q = 4$ performs worst in terms of SRMSE, yet ranks second, very close to the best, in terms of revenue generation.
Furthermore, we observe a steep degradation in captured revenue as $q$ increases, demonstrating the importance of considering small consideration sets during the model estimation when consumers behave in such a way.
The results obtained with the $\ell_1$-error model are unsatisfactory, as it consistently performs poorly across all values of $q$, with all models capturing 50\% of the revenue or less.

\tocheck{Upon inspecting the behavior of the $\ell_1$-error model, we offer the following possible explanation for its poor performance.
	The transaction set in the \texttt{LimitedList} dataset is constructed by randomly picking a subset $S\subseteq N$ and a consumer type from the ground truth according to the associated probabilities. 
	Given the combinatorial number of possible offer sets $S$, we have observed that very rarely does an offer set appear more than once in the training set.
	In most other datasets \citep[e.g.,][]{Berbeglia2022}, instances are generated for a smaller subset of offer sets, but by sampling several consumer types according to their probabilities, thus generating a much larger sample of transactions for each offer set in the training data.
	For this reason, the empirical probabilities using expression \eqref{eq:empirical_probability} in the first case are often 0, 1, and a handful of times $\tfrac{1}{2}$.
	In the second case, these probabilities provide more precise information about the ground truth, but only on a much smaller sample of offer sets.
	Generating too few transactions for many different offer sets is what prevents the $\ell_1$-error model from providing a competitive performance.
	We believe that this has nothing to do with the use or non-use of limited lists, and possibly will extend to other estimation problems (single-purchase, multi-purchase) if the training set has a similar structure.   
}
Finally, we observe that the computational complexity increases with $q$ for small values of this parameter, but stabilizes around $q = 7$, after which the computing times appear largely unaffected.

\subsection{Performance on real and semi-synthetic instances}\label{sec:real_semisynthetic}

In this section, we present results obtained by solving the \texttt{hotel} and \texttt{sushi} datasets, \tocheck{which correspond to single-purchase choice settings}.
Since the problems from the \texttt{hotel} dataset are constructed from actual consumer preferences and do not include information about a ground-truth model, we report the \textit{hard}-RMSE (HRMSE) for this case.
The \texttt{sushi} instances, however, are semi-synthetic and the transactions are derived from a known ground-truth model; therefore, we report the \textit{soft}-RMSE.
Also note that the \texttt{sushi} problems contain pricing data, so we additionally report the AAO metric.
For both datasets, we also report mean computing times.
We report shifted geometric means for the computing times in seconds (using a shift value of 1.0), for the \textit{soft}-RMSE and for the \textit{hard}-RMSE (using a shift value of 0.01), and for the AAO (using a shift value of 0.01).

\begin{table}[!hbtp]
	\caption{Estimation using the two models on the \texttt{hotel} dataset}
	\label{tab:hotel_dl}
	\centering
	\begin{tabular}{crr}
		Model & CPU & HRMSE\\ \hline
		EM & 0.24 & 0.332 \\
		L1 & 0.10 & 0.331 \\
		\hline
	\end{tabular}
\end{table}

\begin{table}[!hbtp]
	\caption{Estimation using the two models on the \texttt{sushi} dataset}
	\label{tab:sushi_dl}
	\centering
	\begin{tabular}{crrr}
		Model & CPU & SRMSE & AAO\\ \hline
		EM &  0.83 & 0.204 & 0.862 \\
		L1 & 2.69 & 0.188 & 0.919 \\
		\hline
	\end{tabular}
\end{table}

The results reported in Tables~\ref{tab:hotel_dl}--\ref{tab:sushi_dl} show that the $\ell_1$-error estimation framework yields a slightly better performance over the EM method on both the error metrics and the revenue generation.

\section{Conclusions} \label{sec:conclusions}

We have introduced a novel DP method to address the CG subproblem of generating consumer types in single- and multi-purchase ranked-list choice models.
The proposed DP builds upon techniques from the vehicle routing literature and incorporates several acceleration mechanisms to effectively reduce the computational complexity associated with estimation problems.

We have shown that the proposed method integrates well with existing market-discovery algorithms and yields substantial computational gains compared with MILP-based approaches from the literature.
In particular, it achieves up to a $70\times$ speed-up over recent methods for multi-purchase ranked-list models and up to a $30\times$ speed-up for single-purchase models.
Moreover, the DP framework can naturally handle modeling variants such as limited consideration sizes.

\tocheck{The proposed approach has also been shown to outperform recent multi-purchase methods relying on more complex models, including behavior-revealed-preference (BRP) rules.
	Our method achieves lower prediction errors and higher revenues while requiring less computational time for the estimation.}
A comprehensive computational study also demonstrates that the resulting demand estimates lead to significant improvements in expected revenue when embedded in assortment optimization models.

We have also assessed the performance of our approach on estimating models with two key behavioral features:
(1) consumers with limited consideration sets,
and (2) consumers with multi-purchase behavior.
In both cases, our results indicate that appropriately chosen model specifications improve both predictive accuracy and revenue performance.

Although our approach is capable of estimating complex demand behaviors in which preferences are represented as distributions over large sets of ranked lists, it cannot capture consumer behavior that deviates from the classical utility maximization framework. In view of the recent interest in general choice models that explain systematic departures from rationality \citep[e.g.,][]{chen2022decision, Berbeglia2018a, chen2019use, Jena2022}, a natural direction for future work is to extend our methodology to support the efficient estimation of these richer behavioral models. A further promising direction is to introduce regularization into the estimation procedure to reduce overfitting, especially when working with small datasets.



\section*{Acknowledgements}
We would like to thank Xiaobo Li and Hongyuan Lin for sharing detailed information about the experiment used in this paper.
\section*{Notes}

This work was conducted before Claudio Contardo joined Amazon and is not related in any way to his work at the company.

	\bibliographystyle{apalike}
	\bibliography{library}

\appendix
		\section{MILP formulations for the generalized linear ordering problem}\label{sec:formulations_GLOP}

In this section, we present the problem that we solve to generate new consumer types.
The problem arising from the process of estimating consumer types can be modeled as a \emph{generalized linear ordering problem} \citep[GLOP, ][]{MendezDiaz2019}.
The GLOP is NP-hard, as it generalizes the classical linear ordering problem, which is known to be NP-hard \citep{Garey1979}.
Since we study the GLOP in the context of consumer-type generation, we revisit some of the notation previously introduced in Section~\ref{sec:model_basis} when defining the problem.
This allows us to clearly highlight how the GLOP relates to the CG subproblem that arises in the estimation framework addressed in this paper.

The GLOP can be formally defined as follows.
Let $\mathcal{N} = \{1, \ldots, n\}$ be a set of nodes (items).
Let $\mathcal{T} = \{(B_t, S_t): B_t \subseteq S_t,\, S_t \subseteq \mathcal{N}\}$ denote a set of directed trees of height one, where each tree $t \in \mathcal{T}$ has root $B_t$ and leaves $S_t$.
Here, $B_t$ represents the set of items (bundle) selected in transaction $t$, which can be possibly empty, i.e., $B_t = \emptyset$.
We allow this case because, in the context of consumer-type generation, the no-purchase option (empty bundle) is always available.
Thus, each tree (transaction) may have an empty root.
Each $t \in \mathcal{T}$ is associated with a non-negative reward $\mu_t$.
In parallel with the notation from Section~\ref{sec:problem_description}, each tree encodes a historical transaction:
the root $B_t$ corresponds to the purchased items (bundle), $S_t$ represents the offer set, and the reward $\mu_t$ corresponds to the dual value associated with Constraint~\eqref{eq:constraints_CG}.
The objective in the GLOP is to find a total order (permutation) $\sigma$ of $\mathcal{N}$ that maximizes the collected reward.
A permutation specifies a ranking of the items, with $\sigma(i)$ denoting the position of item $i$.
The reward $\mu_t$ associated with a tree $t$ is collected if the permutation satisfies
$\sigma(i) < \sigma(j)$ for all $i \in B_t$ and all $j \in S_t \setminus B_t$; that is, all purchased items precede all non-purchased items in that transaction.
Thus, the total reward of $\sigma$ is
$R(\sigma) = \sum \{ \mu_t : t \in \mathcal{T} : \sigma(i) < \sigma(j), i \in B_t, \forall j \in S_t \setminus B_t\}$.

To define the MILP associated with the GLOP, throughout this section, we associated the empty bundle (i.e., the case $B_t = \emptyset$) with a dummy index ~$0$.
For the sake of simplicity, we define an extended product set $\mathcal{N}' = \mathcal{N} \cup \{0\}$ and an extended offer set $S'_t$ including the non-purchase item, i.e., $S'_t = S_t \cup \{0\}$.
To solve the GLOP, one can consider binary variables $x_{ij}$ that assume the value 1 if item $i$ precedes item $j$ in the permutation, and 0 otherwise.
Moreover, binary variables $w_t$, $t \in \mathcal{T}$, take values $1$ if the reward associated with the corresponding tree (transaction) is collected.
Finally, in the definition of the GLOP, a parameter $d$ indicates the position $d + 1$ (or lower) at which the element zero is allowed to be placed in the ranking.

In Section~\ref{sec:customer_generation_single_purchase_case}, we present the formulation proposed by \citet{MendezDiaz2019} to tackle the GLOP.
In turn, in Section~\ref{sec:customer_generation_multi_purchase_case}, we present an adaptation of the formulation proposed by \citet{Lin2025}, which addresses the problem of estimating DCMs under multi-purchase behavior.

\subsection{Generating consumer types for the single-purchase case}\label{sec:customer_generation_single_purchase_case}

In this section, we present the formulation for GLOP, originally proposed by \citet{MendezDiaz2019}.
This formulation is applicable when the transactional data is restricted to unitary bundles, i.e., when each bundle is of the form 
$B_t = \{j_t\}$, for some $j_t \in \mathcal{N}$.
\begin{align}
\max & \sum_{t \in \mathcal{T}} \mu_t w_t && \label{eq:obj_function_glop}
\end{align}
s.t.
\begin{align}
x_{ij} + x_{ji} = 1 && \forall i,j \in \mathcal{N'}, i < j \label{eq:order}\\
x_{ij} + x_{jk} + x_{ki} \leq 2 && \forall i,j,k \in \mathcal{N'}, i \neq j \neq k \label{eq:transitivity}\\
w_t \leq x_{j_tj} && \forall t \in \mathcal{T}, B_t \neq \emptyset , j \in S'_t \setminus \{j_t\} \label{eq:reward_consistency}\\
w_t \leq x_{0j} && \forall t \in \mathcal{T}, B_t = \emptyset, j \in S_t \label{eq:reward_consistency_additional}\\
\sum_{i \in \mathcal{N}} x_{i0} \geq d \label{eq:min_pos_0}\\
\sum_{j \in S_t} x_{j_tj} \leq |S_t| - 1 + w_t && \forall t \in \mathcal{T} \label{eq:redudant}\\
x_{ij} \in \{0,1\} && \forall i,j \in \mathcal{N'} \label{eq:vars_x}\\
w_{t} \in \{0,1\} && \forall t \in \mathcal{T}. \label{eq:vars_w}
\end{align}

Objective function \eqref{eq:obj_function_glop} aims at maximizing the total rewards collected. 
Constraints \eqref{eq:order} impose that either item $i$ must precede item $j$, or item $j$ must precede item $i$.
Constraints \eqref{eq:transitivity} impose a transitive relation among any of the items.
Constraints \eqref{eq:reward_consistency} and \eqref{eq:reward_consistency_additional} ensure that the reward associated with a tree is only collected if the transaction is compatible with the permutation.
Constraint \eqref{eq:min_pos_0} controls the position in which the $0$ item appears in the ranking.
Constraints \eqref{eq:redudant} are valid inequalities used to tighten the formulation.
Finally, Constraints \eqref{eq:vars_x} and \eqref{eq:vars_w} define the domain of the variables.

When solving the GLOP, it is sufficient to consider only meaningful rankings, i.e., rankings truncated at the position of the empty bundle (index~$0$).
This follows from the fact that the relative ordering of items appearing after the empty bundle in the ranking does not affect the objective value.
Consequently, any complete ranking can be modified arbitrarily beyond the position of the empty bundle without changing the collected reward, which implies the existence of multiple symmetric solutions yielding the same objective value.
To improve computational performance when solving the model with a commercial MILP solver, it is therefore beneficial to reduce such symmetries by restricting attention to rankings only up to the position of the empty bundle.
To this end, \citet{MendezDiaz2019} propose adding the following constraint to the model:
\begin{align}
& x_{0j} + x_{0i} \leq 1 + x_{ij} && \forall i,j \in \mathcal{N}, i < j. \label{eq:reduce_symmetry}
\end{align}

In our paper, we also investigate the special case of the GLOP that arises in applications where consumer preference lists are limited to a small set of products.
This feature has been previously explored by \citet{Feldman2019}.
Since the empty bundle marks the end of a meaningful preference list, Equation~\eqref{eq:limit_list} enforces that the empty bundle (index~$0$) cannot appear beyond position $q$, where $q$ represents the maximum length of the consumer's consideration set.
In this setting, $q$ becomes the earliest position at which item $0$ may appear in the ranking, i.e., the model restricts preference lists to at most $q-1$ products before the non-purchase option.
\begin{align}
	& \sum_{i \in \mathcal{N} \setminus \{0\}} x_{i0} \leq q - 1. && \label{eq:limit_list}
\end{align}

\subsection{Generating consumer types for the multi-purchase case}\label{sec:customer_generation_multi_purchase_case}

In this section, we present the formulation used to generate consumer types associated with multiple purchases.
This formulation is an adaptation of the one proposed by \citet{Lin2025}.
In their work, the authors extend the approach of \citet{Vulcano2015} by proposing a market-discovery algorithm to address the case in which consumers are willing to buy up to $\eta$ products.

In addition to transaction data, \citet{Lin2025} also consider consumers' browsing behavior in e-commerce environments.
To develop their model, the authors assume that consumers typically follow a sequence of actions during the purchasing process:
viewing,
clicking,
adding to the shopping cart,
and purchasing.
Based on this behavior, they define behavior-reveals-preference (BRP) rules, which state that products that are clicked are preferred over those that are viewed but not clicked;
products added to the shopping cart are preferred over those that are clicked but not added;
and products that are purchased are preferred over those that are added to the shopping cart but not purchased, as well as over the non-purchase option.
Since this paper considers only information associated with transactions, we adapt the formulation proposed by \citet{Lin2025} to address the more common case in which browsing behavior is not available, as is typical in brick-and-mortar retailing.
Moreover, for the sake of comprehensiveness, we define the formulation using the notation already introduced in this paper.
Also, we rewrite the expressions for completeness.
\begin{align}
\max & \sum_{t \in \mathcal{T}} \mu_t w_t && \label{eq:obj_function_glop_mp}
\end{align}
s.t.
\begin{align}
x_{ij} + x_{ji} = 1 && \forall i,j \in \mathcal{N'}, i < j \label{eq:order_mp}\\
x_{ij} + x_{jk} + x_{ki} \leq 2 && \forall i,j,k \in \mathcal{N'}, i \neq j \neq k \label{eq:transitivity_mp}\\
\sum_{i \in \mathcal{N}} x_{i0} \geq \max \{d, \eta\} \label{eq:min_pos_0_mp}\\
w_t = 0 && t \in \mathcal{T}, |B_t| > \eta \label{eq:incompatibility_bundle_greater_eta_mp} \\
w_t \leq x_{0j} && \forall t \in \mathcal{T}, B_t = \emptyset, j \in S_t \label{eq:reward_consistency_mp}\\
\eta \, w_t \leq \sum_{i \in \mathcal{N} \setminus S_t} x_{ij} && \forall t \in \mathcal{T}, B_t = \emptyset, j \in S'_t \label{eq:reward_consistency_1_mp}\\
w_t \leq x_{0j} && \forall t \in \mathcal{T}, B_t \neq \emptyset, |B_t| < \eta, \notag \\ 
&& j \in S_t \setminus B_t \label{eq:reward_consistency_mp_2}\\
\eta \, w_t \leq \sum_{i \in B_t} x_{ij} + \sum_{i \in \mathcal{N} \setminus S_t} x_{ij} && \forall t \in \mathcal{T}, B_t \neq \emptyset, |B_t| \leq \eta, \notag \\
&& j \in S'_t \setminus B_t \label{eq:reward_consistency_mp_3}\\
x_{ij} \in \{0,1\} && \forall i,j \in \mathcal{N'} \label{eq:vars_x_mp}\\
w_{t} \in \{0,1\} && \forall t \in \mathcal{T}. \label{eq:vars_w_mp}
\end{align}

\clearpage
\pagebreak

\clearpage
\pagebreak

{
\color{black}

\section{MILP formulations for assortment recommendation}\label{sec:milp_assortement_recommendation}

In this section, we present the MILP formulation used to perform assortment recommendation in our study.
The model is adapted from the formulation proposed by \citet{Lin2025}.
This adaptation excludes the components associated with BRP rules, as they are not relevant to the scope of our analysis.

Let $\mathcal{N} = \{1, \ldots, n\}$ be the set of products, indexed by $j$, that may be purchased by consumers.
Let $\lambda_j$ denote the revenue associated with product $j \in \mathcal{N}$.
Let $\mathcal{C}$ be the set of consumer types.
Each consumer type $i \in \mathcal{C}$ is characterized by a preference list $\sigma(i) = (j_1^i, \ldots, j_{t(i)}^i)$ and a purchase capacity $\eta(i) \ge 1$, where $t(i)$ denotes the length of the preference list of consumer type~$i$.
For a given consumer type $i$ and a position $1 \le k \le t(i)$, we denote by $\sigma(i,k) = j_k^i$ the product ranked at position $k$.
The set of products appearing in the preference list of consumer type $i$ is denoted by $\mathcal{N}(i) = \{ j_k^i : 1 \le k \le t(i) \}$.
We also define $\sigma^{-1}(i,j) = k$ if $j = j_k^i$ for some $k$, and $\sigma^{-1}(i,j) = +\infty$ if $j \notin \mathcal{N}(i)$, as the position of product $j$ in the preference list of consumer type~$i$.
For a product $j \in \mathcal{N}(i)$, we define $\mathcal{N}^{+}(i,j) = \{ h \in \mathcal{N}(i) : \sigma^{-1}(i,h) < \sigma^{-1}(i,j) \}$ and $\mathcal{N}^{-}(i,j) = \{ h \in \mathcal{N}(i) : \sigma^{-1}(i,h) > \sigma^{-1}(i,j) \}$ as the sets of products ranked higher and lower than $j$, respectively, in the preference list of consumer type~$i$.
Finally, let $x_i$ denote the probability associated with consumer type $i \in \mathcal{C}$, as estimated by a given discrete choice model.

We consider binary decision variables $\xi_j$, where $\xi_j = 1$ if product $j \in \mathcal{N}$ is included in the selected assortment, and $\xi_j = 0$ otherwise.
We also consider continuous variables $z_{ijk}$ for $i \in \mathcal{C}$, $j \in \mathcal{N}(i)$, and $k = 1,\dots,\eta(i)$, where $z_{ijk}$ indicates whether product $j$ is selected as the $k$-\textit{th} purchase by consumer type~$i$.
Although the variables $z_{ijk}$ have an underlying binary interpretation, \citet{Lin2025} show that they can be modeled as nonnegative continuous variables without loss of optimality.
The resulting MILP formulation is presented below.
\begin{align}
	 \max \ &
	\sum_{i \in \mathcal{C}} 
	\sum_{j \in \mathcal{N}(i)} 
	\sum_{k=1}^{\eta(i)} 
	x_i \lambda_j z_{ijk}
	\label{eq:obj_function_opt_assortement} \\
	\text{s.t.} \
	& \sum_{j \in \mathcal{N}(i)} z_{ijk} \le 1 & \nonumber\\
	&  \qquad \qquad i \in \mathcal{C}, k = 1,\dots,\eta(i) & \label{eq:constr_each_customer_at_most_one_product_per_position} \\
	& \sum_{k=1}^{\eta(i)} z_{ijk} \le \xi_j & \nonumber\\
	&  \qquad \qquad i \in \mathcal{C},\ j \in \mathcal{N}(i) &
	\label{eq:constr_select_product_if_offered} \\
	& \sum_{l \in \mathcal{N}^-(i,j)} z_{ijl} \le 1 - \xi_j & \nonumber\\
	& \qquad \qquad i \in \mathcal{C},\ j \in \mathcal{N}(i) &
	\label{eq:constr_only_selects_most_prefered} \\
	& \sum_{l \in \mathcal{N}^-(i,j)} z_{ilk}
	- \sum_{p=1}^{k-1} z_{ijp}
	\le 1 - \xi_j & \nonumber\\
	& \qquad \qquad  i \in \mathcal{C},\ j \in \mathcal{N}(i),\ k = 2,\dots,\eta(i)
	\label{eq:constr_order_preference} & \\
	& z_{ijk} \ge 0  & \nonumber\\
	& \qquad \qquad i \in \mathcal{C},\ j \in \mathcal{N}(i),\ k = 1,\dots,\eta(i)
	\label{eq:constr_domain_z_assortement} & \\
	&  \xi_j \in \{0,1\} & \nonumber\\
	& \qquad \qquad  j \in \mathcal{N}
	\label{eq:constr_domain_x_assortement}
\end{align}
Objective function \eqref{eq:obj_function_opt_assortement} maximizes the expected revenue generated by the selected assortment across consumer types.
Constraints \eqref{eq:constr_each_customer_at_most_one_product_per_position} ensure that each consumer type $i$ selects at most one product in each of its $\eta(i)$ choice positions.
Constraints \eqref{eq:constr_select_product_if_offered} restrict product selections to offered products and impose that each product can be selected at most once.
Constraints \eqref{eq:constr_only_selects_most_prefered} enforce preference consistency by preventing a consumer type from selecting a less preferred product when a more preferred alternative is offered.
Constraints \eqref{eq:constr_order_preference} ensure that the $k$th choice ($k \geq 2$) is less preferred than the $(k-1)$th choice.
Finally, Constraints \eqref{eq:constr_domain_z_assortement} and \eqref{eq:constr_domain_x_assortement} specify the domains of the decision variables.
}


\section{Considered datasets}\label{sec:instances_appendix}

In this section, we provide a short description of the instances that were considered in our experiments.
We used synthetic and real preference-based datasets, as described in the next two subsections.
For reproducibility purposes, all the instances considered in this study are publicly available in the GitHub repository 
\url{https://github.com/lucianoccosta/single-multi-purchase-cg}.

\subsection{Synthetic datasets}\label{sec:syntetic_dataset_appendix}

We consider four families of synthetic datasets, as described next.

\begin{description}
	\item [\texttt{SinglePurchase}]
	For the single-purchase setting, we rely on the instance set introduced by \citet{Berbeglia2022}, which is grouped into two categories: random and structured.
	In the random category, product selection is non-hierarchical, whereas in the structured category, products follow a natural ordering (e.g., price levels).
	The random instances consider $n = 10$ products (10x1) and two classes of ground-truth models with $k = 10$ and $k = 100$ preference lists.
	\tocheck{In both cases, the ground-truth models are constructed by first considering a passive consumer $c_1 = (\sigma_1, 0)$, where $\sigma_1 = ()$ and $0$ indicates that the consumer is not willing to buy anything.}
	This is a consumer who, by definition, rejects any alternative from any possible assortment.
	This passive consumer type is assigned a probability $p_1 \in \{0.2, 0.5, 0.9\}$, inducing market shares of 80\%, 50\%, and 10\%.
	The data consist of $|\mathcal{P}| \in \{30, 75, 150, 300, 600\}$ periods, each generating 10 transactions, yielding $T \in \{300, 750, 1500, 3000, 6000\}$ total observations.
	In the structured instances, $n' = 5$ item types are considered, each item has three price levels (5x3), resulting in $n = 15$ distinct products.
	Preference lists are constructed by generating random permutations of the 15 products while assuring that the three price levels for every product are ranked accordingly (the cheapest fare is always preferred to the middle-priced one, which in turn is always preferred to the most expensive one).
	The permutation is then truncated to the first $\kappa$ products in the list, where $\kappa$ is an integer uniformly generated in the interval $[1, 15]$.
	The transaction-generation process mirrors that of the random case, except that assortments are built by selecting $3 \le |S'| \le 5$ base items and assigning one of three price levels to each.
	Overall, by varying $n$, $T$, $p_1$, and $k$, the authors generate 60 parameter configurations, and 60 instances per configuration, resulting in a total of 3,600 single-purchase estimation problems.
	Further details can be found in \citet[\S 3.1]{Berbeglia2022}.
	
	{
		\color{black}
        \item[\texttt{MultiPurchaseProbit}]:
		To assess the performance of our approach in a multi-purchase setting, we consider the synthetic instances proposed by \citet{Lin2025}.
        \tocheck{These instances are generated with utilities sampled from a Normal distribution as in the probit model, and allow consumers to purchase multiple products in a single transaction.}
		The instances vary with respect to the number of products, with
		$n \in \{5, 10, 15, 20, 25, 30\}$.
		For each product $i \in \mathcal{N}$, the random utility is defined as
		$U_i = V_i - \beta_i r_i + \varepsilon_i$,
		where $V_i \sim \mathcal{N}(3,1)$ represents the deterministic utility component,
		$r_i \sim \mathcal{U}(1,5)$ denotes the product revenue,
		$\beta_i = -|\tilde{\beta}_i|$ is the price-sensitivity parameter with $\tilde{\beta}_i \sim \mathcal{N}(0,1),$
		and $\varepsilon_i \sim \mathcal{N}(0,1)$ is an i.i.d.\ stochastic noise term.
		The no-purchase option is incorporated by assuming $V_0 = r_0 = 0$.
		Each instance consists of $T = 1{,}000$ transactions, which are split into $80\%$ in-sample (training) data and $20\%$ out-of-sample (testing) data.
		For each transaction, an offer set $S \subseteq \mathcal{N}$ is constructed by uniformly sampling a subset of products whose cardinality satisfies
		$|S| \in \left[\left\lceil \tfrac{n}{3} \right\rceil, \left\lfloor \tfrac{2n}{3} \right\rfloor\right]$,
		with the additional constraint that $|S| \ge 2$ when $n = 5$.
		Given an offer set $S$, an intended purchase quantity $\eta$ is drawn from the discrete support
		$q \in \{0,1,\ldots,Q_{\max}\}$, where $Q_{\max} = 2$.
		The probability of selecting quantity $q$ is proportional to
		$\mathbb{P}(Q = q) \propto \exp(\lambda q),$
		with $\lambda = 0.6$, favoring larger purchase quantities.
		Given the realized $\eta$, the consumer selects up to $q$ products from the offer set according to the MP-MNL choice probabilities induced by the utilities $U_i$, i.e., the $q$ products in $S$ with the highest utilities.
		Thus, a consumer purchases at most $q$ products, but may purchase fewer if fewer acceptable products are available in the assortment.
		For each value of $n$, the authors generate 30 independent instances, resulting in a total of 180 multi-purchase estimation problems.
		Additional details on the instance construction can be found in \citet{Lin2025}.

        To assess our methods using the family of instances proposed by \citet{Lin2025}, we used the same subset of instances used in their experiments, as well as their exact in-sample and out-of-sample split of the transactions in each instance \footnote{All details about their instances were obtained through personal communication with the authors.}.
         
	}
	
	\item [\texttt{MultiPurchaseRankedList}]
	\tocheck{Despite the \texttt{MultiPurchaseProbit} dataset of \citet{Lin2025} already capturing consumer behavior in which more than one product may be purchased within a single transaction, it restricts purchases to at most two products.
		For this reason, in an attempt to gain a deeper understanding of the impact of explicitly modeling multi-purchase behavior in DCM estimation, we additionally generate synthetic instances that allow consumers to purchase more than two products in the same transaction.}
	In these instances, each consumer type is represented by a pair $c = (\sigma,\eta)$, where $\sigma$ is a preference list and $\eta$ denotes the maximum number of products the consumer may buy.
	The instances vary with respect to:
	the number of products $n$,
	the market-share $\{80\%, 50\%, 20\%\}$ (which induce a probability $p_1$ of non-purchase),
	the number of ground-truth types $k$,
	the possible purchase-capacity values $\eta$, and the number of periods $|\mathcal{P}|$, with
	$n \in \{10,15\}$,
	$p_1 \in \{0.2, \, 0.5, \, 0.8\}$,
	$k \in \{25,50,100\}$,
	$\eta \in \{2,3,4,5\}$, and
	$|\mathcal{P}| \in \{30,75,150,300,600\}$.
	By combining these parameter values, 360 configurations are obtained, and 10 instances are generated for each, resulting in a total of 3,600 multi-purchase estimation problems.
	Consumer types in the ground truth are generated by sampling $k$ preference lists uniformly at random, with a passive consumer (i.e., one who purchases nothing) always included.
	After assigning probability $p_1$ to the passive consumer, the remaining probabilities $p_i$ are drawn from a uniform distribution, normalized so that $\sum_{i=1}^k p_i = 1$, and then assigned to the other consumer types.
	For each $\eta \in \{2,3,4,5\}$, a purchase limit $\eta'$ is associated with every list, sampled from $[1,\eta]$ and not exceeding the size of the preference list~$\sigma$, producing consumer types $c = (\sigma,\eta')$.
	This design allows multiple purchase-capacity values to be associated with the same ranking structure, enabling us to study how increased multi-purchase flexibility affects model performance.
	Transactions are generated over $|\mathcal{P}|$ periods, each containing 50 arrivals, resulting in
	$T \in \{1500,3750,7500,15000,30000\}$ observations.
	Each period considers an assortment $S$ with $5 \le |S| \le 10$ items.
	For each transaction, a consumer type $c = (\sigma_c, \eta_c)$ is sampled according to the associated probability distribution of the ground truth, and the purchase set consists of $\pi(c, S)$ as per the definitions from Section \ref{sec:problem_description}, thus allowing no-, single-, and multi-purchase behavior.
	To assess the performance of the model estimation in terms of revenue extraction, when doing assortment optimization, we generate per unit product revenues uniformly at random within the interval $[1,5]$.
	
	\item [\texttt{LimitedList}]
	For the limited-size list setting, we rely on the instance-generation procedure of \citet{Feldman2019}, where consumers purchase vertically differentiated products under single-purchase behavior ($\eta = 1$).
	Products are indexed by decreasing quality, and each consumer type $c$ is associated with a consideration interval $\mathcal{Q}_c = [i_c, \ldots, j_c]$, where $i_c$ captures the consumer's price tolerance and $j_c$ reflects the minimum acceptable quality.
	The corresponding preference list $\sigma_c$ is constructed from $\mathcal{Q}_c$ by randomly removing products with probability $p_d$ and applying up to $\mathcal{F}$ adjacent swaps (each with probability $0.5$) to introduce heterogeneous quality perception.
	We consider two experimental settings.
	In the first, we vary $k \in \{100,500\}$, $p_d \in \{0, 0.25\}$, $\mathcal{F} \in \{1,2,4\}$, and $T \in \{5,000, \, 10,000\}$ with $n = 20$, generating 20 instances per configuration for a total of 240 datasets.
	In the second, we focus on consumers with stronger substitution patterns, retaining only preference lists with at least eight products, and set $n = 20$, $k \in \{100,500\}$, $p_d = 0.25$, $\mathcal{F} = 4$, and $T = 10,000$, generating 10 instances per configuration (20 datasets).
	In all cases, arrival probabilities are sampled uniformly and normalized, assortments are drawn by including each product with probability~0.5 (plus the no-purchase option), and purchases correspond to the first preferred product in $S$ according to $\sigma_c$.
	In total, 260 limited-size list instances are generated.
	In line with the procedure used for the \texttt{MultiPurchaseRankedList} instances, revenues are generated uniformly at random within the interval $[1,5]$ so that we can assess revenue performance in the context of assortment optimization.
\end{description}

\subsection{Datasets derived from real preferences}

In addition to the synthetic datasets described before, we also consider two sets of instances derived from real preferences, as follows.

\begin{description}
	\item[\texttt{hotels}]
	This dataset has been constructed from publicly available booking data originally analyzed in \citet{Bodea2009}.
	The dataset contains booking records from five U.S. hotels over a one-month horizon (March-April 2007), including room availability and rate information at the time of each booking.
	A product corresponds to a room type (e.g., king, queen, suite).
	Products with fewer than 10 purchases are removed, and transactions with missing or inconsistent availability information are discarded \citep{Berbeglia2022}.
	For each remaining booking, four synthetic no-purchase transactions are created using the same offer set, and 80\% of the resulting observations are selected as in-sample data.
	After preprocessing, the five hotels contain $n \in \{4, 6, 8, 10\}$ room types.
	The number of in-sample transactions ranges from approximately $1,000$ to $5,300$ across hotels, with out-of-sample sets ranging from roughly $255$ to $1,325$ transactions.
	More precisely, the datasets contain:
	Hotel~\#1 with $n = 10$ products and $5,290$ in-sample / $1,325$ out-of-sample transactions;
	Hotel~\#2 with $n = 10$ and $1,845$ / $465$ transactions;
	Hotel~\#3 with $n = 8$ and $5,070$ / $1,270$ transactions;
	Hotel~\#4 with $n = 4$ and $1,100$ / $275$ transactions; and
	Hotel~\#5 with $n = 6$ and $1,000$ / $255$ transactions.
	The interested reader is referred to \citet{Berbeglia2022} for full details about the preprocessed datasets used in these experiments.
	
	\item[\texttt{sushi}]
	The sushi dataset is based on the preference data collected by \citet{Kamishima2003}, which contains complete ranked lists from 5,000 consumers over $n = 10$ sushi types.
	Each consumer provides a full ranking of the ten items, yielding one of the $10!$ possible permutations;
	note that the no-purchase option is not explicitly modeled in the survey.
	We consider the same experimental configurations adopted in \citet{Berbeglia2022}.
	The instances vary with respect to the number of periods $|\mathcal{P}| \in \{30, 75, 150, 300, 600\}$.
	In each period, an offer set $S \subseteq \{1,\ldots,10\}$ is sampled such that $3 \le |S| \le 6$, reflecting a small-assortment setting where substitution effects are relevant.
	Two preference specifications are analyzed:
	(i) \texttt{Top10}, where full rankings are used and a consumer purchases if at least one preferred product is present in the assortment; and
	(ii) \texttt{Top3}, where rankings are truncated after the third position, so a purchase occurs only if one of the top three preferred products is available. 
	The \texttt{Top3} configuration induces a stronger relationship between purchase probability and assortment size while still preserving preference heterogeneity.
	For each value of $|\mathcal{P}|$, 60 instances are considered under each specification, resulting in 300 \texttt{Top10} instances and 300 \texttt{Top3} instances overall.
\end{description}
	

\section{Performance metrics}\label{sec:performance_metrics}

To assess the performance of DCMs obtained under different estimation strategies, we consider the \emph{root mean square error} (RMSE) metric.
RMSE is widely used in the choice-modeling literature and evaluates how accurately the estimated choice model predicts the ground-truth model.
Following \citet{Berbeglia2022}, we consider two variants of RMSE depending on the available information:
a \emph{soft RMSE} (SRMSE) and a \emph{hard RMSE} (HRMSE).
The SRMSE is computed using all possible assortments and transactions and evaluates how closely the estimated model recovers the underlying ground-truth probabilities.
When ground-truth probabilities are not available, we instead compute the HRMSE using only the out-of-sample transactions.

Since the metrics SRMSE and HRMSE, commonly used in revenue management (e.g. \citet{Berbeglia2022}) only consider single-purchase transactions, we need to extend them to accommodate the multi-purchase setting.
Consider a ground truth choice model $\theta(\eta)$ where consumers buy at most $\eta$ items.
Let $\mathbb{P}_{\theta(\eta)}(B \mid S)$ denote the ground-truth probabilities under choice model $\theta(\eta)$ that a consumer would buy bundle $B$ when the offer set is $S$.
Consider now a choice model resulting from an estimation procedure, and let us denote by $\mathbb{P}(B \mid S, \mathbf{x})$ the equivalent choice probabilities that are predicted by the estimated model.
We assume that under the estimated choice models, consumers can buy at most $\eta$ items.
Then, the SRMSE is given by
\begin{align}
& SRMSE = \nonumber\\
& \sqrt{\frac{\sum\limits_{S \subseteq \mathcal{N}} \sum\limits_{B\in \mathcal{B}(S, \eta)} \left(\mathbb{P}(B \mid S, \mathbf{x}) - \mathbb{P}_{\theta_{\eta}}(B \mid S) \right)^2}{\sum\limits_{S\subseteq \mathcal{N}} |\mathcal{B}(S, \eta)|
}}, \label{eq:srmse}
\end{align}
where $\mathcal{B}(S, \eta) = \{B\subseteq S: |B|\leq\eta\}$ is the set of bundles of products from $S$ of cardinality at most $\eta$ (including the empty bundle). Observe that when $\eta=1$, this definition coincides with the usual SRMSE employed to assess single purchase models.

Instead of relying on a ground-truth choice model, the HRMSE is computed using a set of out-of-sample transactions.
Let $\mathcal{T}'=\{(B_1,S_1),\hdots,(B_{T},S_T)\}$ denote this set, and let $\mathcal{I}(B, S_t)$ be an indicator function equal to $1$ if the bundle $B \subseteq S_t$ was purchased in transaction
$t \in \mathcal{T}'$
under the offer set $S_t$ (i.e., $B = B_t$), and $0$ otherwise (i.e., $B \neq B_t$).
Let $\eta := \max\{|B_t| : t \in [\mathcal{T}']\}$ denote the maximum cardinality of any purchased bundle across all transactions.
Similar to the SRMSE, consider a choice model resulting from an estimation procedure and denote by $\mathbb{P}(B \mid S, \mathbf{x})$ the predicted probability that a consumer purchases bundle $B$ when the offer set is $S$.
The HRMSE is given by

\begin{align}
& HRMSE = \nonumber\\
& \sqrt{\frac{\sum\limits_{t=1}^T \sum\limits_{B\in\mathcal{B}(S_t, \eta)} \left( \mathcal{I}(B, S_t) - \mathbb{P}(B \mid S_t) \right)^2}{\sum\limits_{t=1}^T |\mathcal{B}(S_t, \eta)|
}}. \label{eq:hrmse}
\end{align}


{
\color{black}
\citet{Lin2025} considered an alternative metric to \eqref{eq:srmse} and \eqref{eq:hrmse} to assess predictions in the multi-purchase setting. 
Instead of evaluating bundles of products, they proposed a metric that predicts the purchase probability of individual products in the offer sets, assigning a predicted probability to each item.
Since this metric captures the marginal contribution of each product within a bundle to its purchase probability, we refer to it as the \textit{Marginal-RMSE} (MRMSE).
The MRMSE is computed using expression \eqref{eq:mrsme}, originally introduced in \citet{Lin2025}, and is given by:
\begin{align}
& MRMSE = \nonumber\\
& \sqrt{
		\frac{1}{\sum_{t \in \mathcal{T}'} |S_t|} 
		\sum_{t \in \mathcal{T}'} \sum_{j \in S_t}
		\left(
		\mathcal{I}(j, S_t) - \mathbb{P}(j \mid S_t)
		\right)^2
	}.
	\label{eq:mrsme}
\end{align}
Functions $\mathcal{I}(\cdot)$ and $\mathbb{P}(\cdot)$ have the same meaning as above, differing only in that they are applied to individual products rather than to bundles.
In Equation \eqref{eq:mrsme}, instead of computing the probability of a complete transaction, the metric evaluates the vector of marginal purchase probabilities $\mathbb{P}(\cdot)$ item by item, without explicitly considering the no-purchase option (which would appear in $S_t'$).
The definition of $\mathcal{I}(j, S_t)$ is straightforward:
$\mathcal{I}(j, S_t) = 1$ if product $j$ is purchased in transaction $t$, and $0$ otherwise.
The value of $\mathbb{P}(j \mid S_t)$ is computed as $\mathbb{P}(j \mid S_t) = \sum_{i \in \mathcal{C}} x_i \, a_{i,j,S_t}$, where $x_i$ is the probability associated with consumer type $i$, and $a_{i,j,S_t}$ is a binary parameter indicating whether consumer type $i$ purchases product $j$ when faced with offer set $S_t$.
}

\tocheckTwo{
MRMSE can be particularly useful when assessing the performance of simpler models that impose restrictions on the maximum number of purchases. 
While SRMSE and HRMSE provide appropriate metrics when the maximum number of purchases allowed by the estimated model matches that observed in the data, they become less informative when estimating simpler models that restrict the number of purchases.
In particular, when the data contain transactions with more than $k$ purchased products, but the estimated model allows at most $k$ purchases, the model will necessarily assign zero probability to all transactions with more than $k$ products. 
In such cases, SRMSE and HRMSE may fail to properly distinguish between models with different predictive capabilities, since all models incur the same error for those transactions.
To address this limitation, MRMSE evaluates the marginal probability of each product being purchased within an offer set instead of the probability of complete purchase bundles. 
This allows the metric to assess the predictive quality of simpler models even when they cannot reproduce the exact bundle sizes observed in the data.
}


{
\color{black}

\section{Assortment Revenue Computation}\label{sec:monte_carlo_simulation}

In this section, we describe the evaluation procedure used to compute the revenue induced by a given assortment.
We follow the same procedure as the one employed by \citet{Lin2025}.
We assess the expected revenue of assortments obtained under different estimation methods using a consumer-based simulation framework.
We attempted to reproduce this procedure as closely as possible and, through personal communication with the authors, obtained additional implementation details.
The evaluation procedure consists of three main steps, which are described in detail below.
    
\paragraph{Step 1: Input Assortments}	
We consider assortments obtained by solving an MILP formulation (as for instance the one described in Appendix \ref{sec:milp_assortement_recommendation}), considering the probability mass functions (PMFs) estimated using a given method.  
For each instance and each estimation setting, namely \texttt{ML\_DP} and \texttt{L1\_DP}, a single optimal assortment is generated and used as input for the evaluation procedure.
	
\paragraph{Step 2: Generation of the Consumer Evaluation Base}
For each instance, we generate an evaluation base consisting of $10{,}000$ synthetic consumers using the product-level information provided by the original authors (instance files obtained through personal communication).
Each instance file contains the following parameters for each product $j \in \mathcal{N}$:
deterministic utilities $V = (V_j)_{j \in \mathcal{N}}$,
price sensitivities $\beta = (\beta_j)_{j \in \mathcal{N}}$, and
prices (revenues) $P = (P_j)_{j \in \mathcal{N}}$.
For each consumer, random utilities are generated according to a probit-type random utility model.
The utility of the outside option (no-purchase) is defined as $U_0 \sim \mathcal{N}(0,1)$.
For each product $j \in \mathcal{N}$, the utility is given by $U_j = V_j - \beta_j P_j + \varepsilon_j$,
where the error terms $\varepsilon_j \sim \mathcal{N}(0,1)$ are independently and identically distributed.
The complete utility vector is therefore defined as $U = (U_0, U_1, \dots, U_{|\mathcal{N}|})$.

Given the utility vector $U$, all alternatives (including the outside option) are sorted in non-increasing order of utility.
This sorted list defines the consumer's preference ranking.
To complete the consumer type definition, each consumer is assigned an intended purchase quantity $\eta$.
The value of $\eta$ is independently drawn from the discrete set $k \in \{0,1,\dots,Q_{\max}\}$, with $Q_{\max}=2$ (as obtained through personal communication with the authors), according to the exponential distribution $\mathbb{P}(\mathrm{IPQ}=k) \propto \exp(\lambda k)$, where $\lambda = 0.6$.
The resulting probabilities are normalized to ensure that the resulting probability mass function is valid.
Each consumer is fully characterized by a preference ranking and an $\eta$ value.
After generating the $10{,}000$ consumers for a given instance, the entire consumer base is stored in an external file to ensure reproducibility of the revenue evaluation process.

\paragraph{Step 3: Revenue Evaluation of a Given Assortment}
Let $S \subseteq \mathcal{N}$ denote the assortment associated with a given instance and estimation method.
To compute the expected revenue induced by $S$, we simulate the purchasing behavior of all consumers in the evaluation base.
For each consumer, the following procedure is applied:
\begin{enumerate}
	\item The assortment $S$ is presented to the consumer.
	\item The consumer scans her/his preference list sequentially.
	\item Products belonging to $S$ are selected until one of the following conditions is met:
	\begin{itemize}
		\item the number of selected products reaches the consumer's $\eta$;
		\item the outside option is encountered;
		\item no additional products from $S$ are available.
	\end{itemize}
\end{enumerate}
Thus, a consumer purchases at most $\eta$ products, but may purchase fewer if fewer acceptable products are available in the assortment.
Given the set of products selected by the consumer, the revenue collected from that consumer is computed as the sum of the corresponding product prices.

Repeating this process for all consumers, the expected revenue of the assortment $S$ is estimated as
\begin{equation}
	\bar{R}(S) = \frac{1}{M} \sum_{i=1}^{M} R_i(S),
\end{equation}
where $M = 10{,}000$ denotes the number of consumers and $R_i(S)$ is the revenue generated by consumer $i$.
    
%

\begin{algorithm}[H]
\caption{Evaluation of Expected Revenue}
\begin{algorithmic}[1]

\For{each instance}
\State Load the optimal assortment $S$
\State Generate $M = 10{,}000$ consumers

\For{each consumer $i = 1,\dots,M$}
    \State Draw $U_{0} \sim \mathcal{N}(0,1)$

    \For{each product $j \in \mathcal{N}$}
        \State Draw $\varepsilon_{j} \sim \mathcal{N}(0,1)$
        \State Compute $U_{j} \gets V_{j} - \beta_{j} P_{j} + \varepsilon_{j}$
    \EndFor

    \State Build the preference ranking by sorting products $j \in \mathcal{N}$, with respect to $U = (U_{0}, U_{1}, \dots, U_{|\mathcal{N}|})$ in non-increasing order

    \State Draw $\eta \in \{0,1,\dots,Q_{\max}\}$ from a truncated discrete distribution, $Q_{\max} = 2$.
\EndFor

\State Simulate purchases from $S$ for all consumers
\State Compute the average revenue $\bar{R}(S)$
\EndFor


\end{algorithmic}
\label{alg:revenue_eval}
\end{algorithm}
}


\section{Impact of the acceleration techniques on the performance of the DP}\label{sec:dp:accel}

In this experiment, we evaluate the effectiveness of two acceleration techniques introduced in Section~\ref{sec:acceleration_techiniques}, namely, completion bounds (Section~\ref{sec:completion_bounds}) and unreachable-products procedure (Section~\ref{sec:unreachable_prodcts}).
We consider four settings for the proposed DP:  
\begin{description}
	\item[None:] neither feature is active;  
	\item[UP:] only unreachable products are active;  
	\item[CB:] only completion bounds are active;  
	\item[Both:] both features are active.  
\end{description}

We compare the performance of these settings under single- and multi-purchase estimation problems, using instances \texttt{SinglePurchase} and \texttt{MultiPurchaseRankedList} (Section \ref{sec:instances}).
Tables~\ref{tab:cpu_times_cb_up_exp1:sp:em}-\ref{tab:cpu_times_cb_up_exp1:mp:l1:share80} report the mean computational times necessary to estimate the EM and the $\ell_1$-error models in the single- and in the multi-purchase settings.
For the computation of geometric means, timeouts are treated as having a computing time of 3 hours.
Moreover, we restrict our analysis to the problem configurations such that the minimum geometric mean across all algorithm variants is less than 3 hours, and such that the maximum geometric mean across all algorithm variants is higher than $60$ seconds for the single-purchase case, and higher than $600$ seconds for the multi-purchase case.
For this reason, the results for certain instance classes in both the \texttt{SinglePurchase} and \texttt{MultiPurchaseRankedList} settings are not presented in the tables.
We observe that, overall, it is the addition of both features that results in the best performance.

\begin{table}[!hbtp]
\caption{Computational time taken in seconds as a function of the two acceleration techniques available for single-purchase problems using the EM estimation model (shifted geometric mean of the duration in seconds)}
\label{tab:cpu_times_cb_up_exp1:sp:em}
\centering
\begin{tabular}{ccccrrrr}
Share & Structure & Lists & $|\mathcal{P}|$ & None & UP & CB & Both \\ \hline
\multirow{6}{*}{20\%} & \multirow{6}{*}{5x3} & \multirow{3}{*}{10} & 150 & 65.2 & 69.5 & \textbf{8.9} & \textbf{8.9} \\
&  & & 300 & 153.0 & 159.9 & 37.6 & \textbf{37.4} \\
&  & & 600 & 359.7 & 357.3 & 124.0 & \textbf{120.5} \\
\cline{3-8}
&  & \multirow{3}{*}{100} & 150 & 71.2 & 76.0 & \textbf{15.7} & 16.2 \\
&  &  & 300 & 157.0 & 164.2 & \textbf{57.2} & 57.9 \\
&  &  & 600 & 364.8 & 354.9 & 194.6 & \textbf{190.2} \\
\hline
\multirow{6}{*}{50\%} & \multirow{6}{*}{5x3} & \multirow{3}{*}{10} & 150 & 91.7 & 104.4 & 22.0 & \textbf{21.4} \\
 &  &  & 300 & 202.3 & 201.8 & \textbf{66.3} & 67.6 \\
 &  &  & 600 & 441.8 & 366.3 & 181.4 & \textbf{167.2} \\
 \cline{5-8}
 &  & \multirow{3}{*}{100} & 150 & 89.6 & 100.0 & \textbf{34.0} & 34.3 \\
 &  &  & 300 & 173.0 & 189.3 & \textbf{88.7} & 93.5 \\
 &  &  & 600 & 393.0 & 414.5 & 288.4 & \textbf{270.7} \\
\hline
\multirow{8}{*}{90\%} & \multirow{8}{*}{5x3} & \multirow{4}{*}{10} & 75 & 70.5 & 73.5 & 11.0 & \textbf{10.2} \\
 &  &  & 150 & 166.0 & 172.6 & \textbf{34.4} & 35.0 \\
 &  &  & 300 & 354.9 & 371.4 & \textbf{103.7} & 109.7 \\
 &  &  & 600 & 717.2 & 754.5 & \textbf{255.5} & 273.4 \\
 \cline{3-8}
 &  & \multirow{4}{*}{100} & 75 & 84.2 & 85.0 & 19.2 & \textbf{19.0} \\
 &  &  & 150 & 165.3 & 183.0 & 62.5 & \textbf{61.6} \\
 &  &  & 300 & 360.9 & 369.7 & 183.5 & \textbf{170.0} \\
 &  &  & 600 & 811.8 & 841.7 & \textbf{497.4} & 498.8 \\
\hline
\end{tabular}
\end{table}

\begin{table}[!hbtp]
\caption{Computational time taken in seconds as a function of the two acceleration techniques available for single-purchase problems using the $\ell_1$-error model (shifted geometric mean of the duration in seconds)}
\label{tab:cpu_times_cb_up_exp1:sp:l1}
\centering
\begin{tabular}{ccccrrrr}
Share & Structure & Lists & $|\mathcal{P}|$ & None & UP & CB & Both \\ \hline
\multirow{10}{*}{20\%} & \multirow{10}{*}{5x3} & \multirow{5}{*}{10} & 30 & 133.5 & 80.9 & 4.6 & \textbf{3.0} \\
 & & & 75 & 249.0 & 216.4 & 22.7 & \textbf{20.5} \\
 & & & 150 & 378.4 & 345.7 & 67.7 & \textbf{61.2} \\
 & & & 300 & 637.7 & 573.1 & 223.2 & \textbf{217.1} \\
 & & & 600 & 1184.5 & 1064.8 & 607.8 & \textbf{550.1} \\
 \cline{3-8}
 & & \multirow{5}{*}{100} & 30 & 146.2 & 89.4 & 6.3 & \textbf{3.8} \\
 & & & 75 & 262.1 & 219.5 & 33.4 & \textbf{30.6} \\
 & & & 150 & 411.6 & 406.0 & 104.5 & \textbf{97.9} \\
 & &  & 300 & 625.1 & 623.1 & 276.8 & \textbf{265.9} \\
 & & & 600 & 1160.3 & 1135.5 & 714.9 & \textbf{712.4} \\
\hline
\multirow{9}{*}{50\%} & \multirow{9}{*}{5x3} & \multirow{5}{*}{10} & 30 & 405.2 & 216.1 & 30.9 & \textbf{14.0} \\
 & & & 75 & 656.5 & 497.0 & 102.5 & \textbf{73.0} \\
 & & & 150 & 959.6 & 839.1 & 279.5 & \textbf{225.6} \\
 & & & 300 & 2109.9 & 1682.3 & 706.6 & \textbf{643.7} \\
 &  & & 600 & 10158.9 & 8957.0 & 3466.5 & \textbf{3288.8} \\
 \cline{3-8}
 & & \multirow{4}{*}{100} & 30 & 522.1 & 282.2 & 48.6 & \textbf{23.1} \\
 &  & & 75 & 830.5 & 673.7 & 156.2 & \textbf{117.0} \\
 &  & & 150 & 1205.6 & 1107.8 & 402.4 & \textbf{349.6} \\
 &  & & 300 & 8678.2 & 6812.9 & 1130.7 & \textbf{977.6} \\
\hline
\multirow{9}{*}{90\%}  & \multirow{9}{*}{5x3}  & \multirow{5}{*}{10} & 30 & 702.1 & 323.8 & 65.0 & \textbf{26.0} \\
 & &  & 75 & 1081.1 & 727.9 & 180.3 & \textbf{114.4} \\
 & &  & 150 & 1766.8 & 1212.8 & 440.4 & \textbf{325.8} \\
 & &  & 300 & 5306.9 & 2851.8 & 926.0 & \textbf{832.4} \\
 & &  & 600 & 10800.0 & 10800.0 & \textbf{5552.9} & 5559.9 \\
  \cline{3-8}
 & & \multirow{4}{*}{100} & 30 & 952.6 & 433.7 & 124.4 & \textbf{46.9} \\
 & &  & 75 & 1441.4 & 970.4 & 279.6 & \textbf{177.9} \\
 & &  & 150 & 4353.9 & 1692.8 & 630.4 & \textbf{475.0} \\
 & &  & 300 & 10800.0 & 10800.0 & 1499.3 & \textbf{1212.4} \\
\hline
\end{tabular}
\end{table}

\begin{table}[!hbtp]
\caption{Computational time taken in seconds as a function of the two acceleration techniques available under multi-purchase problems using the EM estimation model (shifted geometric mean of the duration in seconds) for problems with $|\mathcal{N}| = 15$ and a market share of 20\%.}
\label{tab:cpu_times_cb_up_exp1:mp:em:share20}
\centering
\begin{tabular}{cccrrrr}
$|\mathcal{C}|$ & $\eta$ & $|\mathcal{P}|$ & None & UP & CB & Both \\ \hline
\multirow{3}{*}{25} & \multirow{3}{*}{2} & 150 & 766.4 & 430.3 & 37.2 & \textbf{32.8} \\
 &  & 300 & 1579.4 & 1129.2 & 97.7 & \textbf{95.1} \\
 &  & 600 & 10800.0 & 10800.0 & 312.5 & \textbf{282.9} \\
\hline
\multirow{3}{*}{50} & \multirow{3}{*}{2} & 150 & 609.9 & 633.4 & 85.4 & \textbf{79.2} \\
 &  & 300 & 10800.0 & 6132.3 & \textbf{436.7} & 488.8 \\
 &  & 600 & 10800.0 & 10800.0 & 1117.1 & \textbf{786.4} \\
\hline
\multirow{3}{*}{100} & \multirow{3}{*}{2} & 150 & 626.7 & 541.6 & 126.7 & \textbf{104.1} \\
 &  & 300 & 6115.2 & 1578.0 & 625.2 & \textbf{397.3} \\
 &  & 600 & 10800.0 & 10800.0 & 4617.8 & \textbf{1570.0} \\
\hline
\multirow{3}{*}{25} & \multirow{3}{*}{3} & 150 & 1331.0 & 1082.3 & \textbf{118.4} & 143.1 \\
 &  & 300 & 10800.0 & 10800.0 & 401.2 & \textbf{377.5} \\
 &  & 600 & 10800.0 & 10800.0 & 1666.5 & \textbf{1073.2} \\
\hline
\multirow{3}{*}{50} & \multirow{3}{*}{3} & 150 & 1804.4 & 1381.2 & 322.2 & \textbf{275.6} \\
 &  & 300 & 10800.0 & 10800.0 & 908.7 & \textbf{897.3} \\
 &  & 600 & 10800.0 & 10800.0 & 8851.3 & \textbf{4029.4} \\
\hline
\multirow{3}{*}{100} & \multirow{3}{*}{3} & 150 & 2912.1 & 1116.1 & 326.1 & \textbf{280.9} \\
 &  & 300 & 10800.0 & 10800.0 & 926.0 & \textbf{754.0} \\
 &  & 600 & 10800.0 & 10800.0 & 8933.9 & \textbf{8782.3} \\
\hline
\multirow{4}{*}{25} & \multirow{4}{*}{4} & 75 & 586.5 & 615.9 & \textbf{58.4} & 65.0 \\
 &  & 150 & 1683.0 & 1315.8 & 154.9 & \textbf{144.7} \\
 &  & 300 & 10800.0 & 8988.0 & 576.9 & \textbf{501.8} \\
 &  & 600 & 10800.0 & 10800.0 & 3885.8 & \textbf{1543.1} \\
\hline
\multirow{3}{*}{50} & \multirow{3}{*}{4} & 75 & 769.0 & 731.2 & \textbf{135.5} & 136.1 \\
 &  & 150 & 5249.3 & 5931.6 & 530.6 & \textbf{442.7} \\
 &  & 300 & 10800.0 & 10800.0 & 2844.1 & \textbf{1347.9} \\
\hline
\multirow{3}{*}{100} & \multirow{3}{*}{4} & 75 & 717.6 & 612.7 & 161.8 & \textbf{161.1} \\
 &  & 150 & 8821.7 & 1825.3 & 625.0 & \textbf{597.8} \\
 &  & 300 & 10800.0 & 10800.0 & 6107.2 & \textbf{3346.4} \\
\hline
\multirow{4}{*}{25} & \multirow{4}{*}{5} & 75 & 656.4 & 780.2 & \textbf{130.3} & 135.7 \\
 &  & 150 & 8910.9 & 4845.4 & 374.3 & \textbf{314.6} \\
 &  & 300 & 10800.0 & 10800.0 & 1287.1 & \textbf{941.6} \\
 &  & 600 & 10800.0 & 10800.0 & 10800.0 & \textbf{7495.8} \\
\hline
\multirow{3}{*}{50} & \multirow{3}{*}{5} & 75 & 811.0 & 968.9 & \textbf{217.7} & 224.9 \\
 &  & 150 & 6047.7 & 10800.0 & 854.7 & \textbf{843.4} \\
 &  & 300 & 10800.0 & 10800.0 & 7299.0 & \textbf{3488.4} \\
\hline
\multirow{3}{*}{100} & \multirow{3}{*}{5} & 75 & 1372.0 & 908.0 & \textbf{335.1} & 338.7 \\
 &  & 150 & 10800.0 & 8878.7 & 907.7 & \textbf{844.1} \\
 &  & 300 & 10800.0 & 10800.0 & 8978.5 & \textbf{8691.1} \\
\hline
\end{tabular}

\end{table}

\begin{table}[!hbtp]
\caption{Computational time taken in seconds as a function of the two acceleration techniques available under multi-purchase problems using the EM estimation model (shifted geometric mean of the duration in seconds) for problems with $|\mathcal{N}| = 15$ and a market share of 50\%.}
\label{tab:cpu_times_cb_up_exp1:mp:em:share50}
\centering
\begin{tabular}{cccrrrr}
$|\mathcal{C}|$ & $\eta$ & $|\mathcal{P}|$ & None & UP & CB & Both \\ \hline
\multirow{3}{*}{25} & \multirow{3}{*}{2} & 150 & 668.2  & 642.9  & 49.6   & \textbf{48.7} \\
                    &                    & 300 & 3592.6 & 1511.8 & 114.7  & \textbf{107.3} \\
                    &                    & 600 & 10800.0 & 10800.0 & \textbf{433.4} & 439.6 \\
\hline
\multirow{3}{*}{50} & \multirow{3}{*}{2} & 150 & 916.9  & 828.2  & 132.7  & \textbf{123.6} \\
                    &                    & 300 & 4940.2 & 2749.0 & 449.8  & \textbf{381.9} \\
                    &                    & 600 & 10800.0 & 10800.0 & 2567.8 & \textbf{1601.5} \\
\hline
\multirow{2}{*}{100} & \multirow{2}{*}{2} & 150 & 1214.8 & 1048.8 & \textbf{284.1} & 322.4 \\
                     &                    & 300 & 10800.0 & 10800.0 & 1591.1 & \textbf{1201.9} \\
\hline
\multirow{3}{*}{25} & \multirow{3}{*}{3} & 150 & 3099.8 & 1350.8 & 188.3  & \textbf{178.1} \\
                    &                    & 300 & 10800.0 & 10800.0 & 774.6  & \textbf{639.7} \\
                    &                    & 600 & 10800.0 & 10800.0 & 3257.2 & \textbf{1425.1} \\
\hline
\multirow{4}{*}{50} & \multirow{4}{*}{3} & 75  & 676.0  & 613.0  & \textbf{105.3} & 133.4 \\
                    &                    & 150 & 2656.4 & 1392.3 & 404.6  & \textbf{381.2} \\
                    &                    & 300 & 10800.0 & 10800.0 & 2077.3 & \textbf{1138.1} \\
                    &                    & 600 & 10800.0 & 10800.0 & 6988.0 & \textbf{5592.5} \\
\hline
\multirow{3}{*}{100} & \multirow{3}{*}{3} & 75  & 666.0  & 659.7  & 191.5  & \textbf{168.3} \\
                     &                    & 150 & 7410.3 & 7221.4 & 669.9  & \textbf{620.6} \\
                     &                    & 300 & 10800.0 & 10800.0 & 4631.2 & \textbf{3364.8} \\
\hline
\multirow{4}{*}{25} & \multirow{4}{*}{4} & 75  & 757.3  & 714.4  & 127.8  & \textbf{117.4} \\
                    &                    & 150 & 4015.3 & 1446.0 & 233.5  & \textbf{228.6} \\
                    &                    & 300 & 8979.9 & 10800.0 & 736.3  & \textbf{685.2} \\
                    &                    & 600 & 10800.0 & 10800.0 & 5731.5 & \textbf{3646.7} \\
\hline
\multirow{3}{*}{50} & \multirow{3}{*}{4} & 75  & 1048.4 & 1026.8 & 270.8  & \textbf{254.0} \\
                    &                    & 150 & 7422.1 & 4178.1 & 726.2  & \textbf{598.5} \\
                    &                    & 300 & 10800.0 & 10800.0 & 6985.2 & \textbf{1820.3} \\
\hline
\multirow{2}{*}{100} & \multirow{2}{*}{4} & 75  & 1445.7 & 1265.5 & 447.7  & \textbf{413.5} \\
                     &                    & 150 & 10800.0 & 10800.0 & 1391.4 & \textbf{1308.7} \\
\hline
\multirow{4}{*}{25} & \multirow{4}{*}{5} & 75  & 1085.7 & 1159.1 & \textbf{285.4} & 305.8 \\
                    &                    & 150 & 7263.8 & 7561.9 & 778.8  & \textbf{590.2} \\
                    &                    & 300 & 10800.0 & 10800.0 & 2897.7 & \textbf{2047.3} \\
                    &                    & 600 & 10800.0 & 10800.0 & 10800.0 & \textbf{9016.9} \\
\hline
\multirow{3}{*}{50} & \multirow{3}{*}{5} & 75  & 1365.2 & 1139.8 & 622.9  & \textbf{523.9} \\
                    &                    & 150 & 9005.8 & 8985.3 & 1210.9 & \textbf{871.3} \\
                    &                    & 300 & 10800.0 & 10800.0 & 10800.0 & \textbf{4052.0} \\
\hline
\multirow{2}{*}{100} & \multirow{2}{*}{5} & 75  & 3145.9 & 2454.1 & \textbf{707.3} & 714.0 \\
                     &                    & 150 & 10800.0 & 10800.0 & 6158.5 & \textbf{2824.0} \\
\hline
\end{tabular}
\end{table}

\begin{table}[!hbtp]
\caption{Computational time taken in seconds as a function of the two acceleration techniques available under multi-purchase problems using the EM estimation model (shifted geometric mean of the duration in seconds) for problems with $|\mathcal{N}| = 15$ and a market share of 80\%.}
\label{tab:cpu_times_cb_up_exp1:mp:em:share80}
\centering
\begin{tabular}{cccrrrr}
$|\mathcal{C}|$ & $\eta$ & $|\mathcal{P}|$ & None & UP & CB & Both \\ \hline
\multirow{3}{*}{25} & \multirow{3}{*}{2} & 150 & 641.8   & 575.4   & 69.9   & \textbf{61.8} \\
                    &                    & 300 & 1628.2  & 1249.9  & \textbf{178.9} & 182.6 \\
                    &                    & 600 & 10800.0 & 7512.5  & 485.8  & \textbf{334.3} \\
\hline
\multirow{3}{*}{50} & \multirow{3}{*}{2} & 150 & 1093.2  & 908.7   & 203.5  & \textbf{155.2} \\
                    &                    & 300 & 7234.7  & 4860.8  & 593.3  & \textbf{504.2} \\
                    &                    & 600 & 10800.0 & 10800.0 & 1419.7 & \textbf{1052.8} \\
\hline
\multirow{2}{*}{100} & \multirow{2}{*}{2} & 150 & 2046.8  & 1231.3  & \textbf{357.9} & 407.8 \\
                     &                    & 300 & 10800.0 & 10800.0 & 1757.8 & \textbf{1198.7} \\
\hline
\multirow{4}{*}{25} & \multirow{4}{*}{3} & 75  & 661.1   & 612.8   & 82.2   & \textbf{76.3} \\
                    &                    & 150 & 1742.0  & 1383.6  & 256.2  & \textbf{236.2} \\
                    &                    & 300 & 10800.0 & 10800.0 & 1073.2 & \textbf{826.4} \\
                    &                    & 600 & 10800.0 & 10800.0 & 4772.3 & \textbf{1351.2} \\
\hline
\multirow{3}{*}{50} & \multirow{3}{*}{3} & 75  & 840.3   & 758.5   & 226.4  & \textbf{193.3} \\
                    &                    & 150 & 4018.6  & 2542.0  & 541.9  & \textbf{363.9} \\
                    &                    & 300 & 10800.0 & 10800.0 & 1959.7 & \textbf{1195.4} \\
\hline
\multirow{3}{*}{100} & \multirow{3}{*}{3} & 75  & 841.0   & 841.9   & 267.9  & \textbf{264.9} \\
                     &                    & 150 & 5909.7  & 4557.8  & 813.2  & \textbf{751.2} \\
                     &                    & 300 & 10800.0 & 10800.0 & 7261.4 & \textbf{4870.6} \\
\hline
\multirow{4}{*}{25} & \multirow{4}{*}{4} & 75  & 816.8   & 837.5   & 123.4  & \textbf{123.2} \\
                    &                    & 150 & 5661.9  & 1437.3  & 346.5  & \textbf{291.0} \\
                    &                    & 300 & 10800.0 & 10800.0 & 1114.3 & \textbf{729.7} \\
                    &                    & 600 & 10800.0 & 10800.0 & 5771.2 & \textbf{3595.6} \\
\hline
\multirow{3}{*}{50} & \multirow{3}{*}{4} & 75  & 1473.1  & 1244.0  & 374.5  & \textbf{371.1} \\
                    &                    & 150 & 9013.3  & 8997.3  & 900.8  & \textbf{754.7} \\
                    &                    & 300 & 10800.0 & 10800.0 & 6217.0 & \textbf{5995.7} \\
\hline
\multirow{2}{*}{100} & \multirow{2}{*}{4} & 75  & 2596.4  & 1545.4  & 647.5  & \textbf{614.3} \\
                     &                    & 150 & 10800.0 & 10800.0 & 7468.3 & \textbf{1796.2} \\
\hline
\multirow{3}{*}{25} & \multirow{3}{*}{5} & 75  & 1407.9  & 1148.7  & \textbf{242.3} & 251.4 \\
                    &                    & 150 & 10800.0 & 8943.2  & 644.7  & \textbf{487.3} \\
                    &                    & 300 & 10800.0 & 10800.0 & \textbf{4805.8} & 5613.4 \\
\hline
\multirow{3}{*}{50} & \multirow{3}{*}{5} & 75  & 2007.7  & 1231.4  & 508.6  & \textbf{506.6} \\
                    &                    & 150 & 10800.0 & 10800.0 & 1490.4 & \textbf{1210.7} \\
                    &                    & 300 & 10800.0 & 10800.0 & 10800.0 & \textbf{8746.3} \\
\hline
\multirow{3}{*}{100} & \multirow{3}{*}{5} & 30  & 696.2   & 605.1   & 307.1  & \textbf{262.8} \\
                     &                    & 75  & 5157.0  & 5030.6  & 1193.8 & \textbf{937.4} \\
                     &                    & 150 & 10800.0 & 10800.0 & 10800.0 & \textbf{5092.3} \\
\hline
\end{tabular}

\end{table}

\begin{table}[!hbtp]
\caption{Computational time taken in seconds as a function of the two acceleration techniques available under multi-purchase problems using the $\ell_1$-error model (shifted geometric mean of the duration in seconds) for problems with $|\mathcal{N}| = 15$ and a market share of 20\%.}
\label{tab:cpu_times_cb_up_exp1:mp:l1:share20}
\centering
\begin{tabular}{cccrrrr}
$|\mathcal{C}|$ & $\eta$ & $|\mathcal{P}|$ & None & UP & CB & Both \\ \hline
\multirow{5}{*}{25} & \multirow{5}{*}{2} & 30  & 644.9   & 415.6   & 17.3   & \textbf{8.8} \\
                    &                    & 75  & 6074.3  & 2873.2  & 103.2  & \textbf{55.9} \\
                    &                    & 150 & 10800.0 & 10800.0 & 229.4  & \textbf{139.1} \\
                    &                    & 300 & 10800.0 & 10800.0 & 620.2  & \textbf{443.3} \\
                    &                    & 600 & 10800.0 & 10800.0 & 3079.7 & \textbf{1783.1} \\
\hline
\multirow{4}{*}{50} & \multirow{4}{*}{2} & 75  & 8370.5  & 3009.1  & 125.4  & \textbf{74.0} \\
                    &                    & 150 & 10800.0 & 10800.0 & 545.1  & \textbf{315.8} \\
                    &                    & 300 & 10800.0 & 10800.0 & 1305.8 & \textbf{982.3} \\
                    &                    & 600 & 10800.0 & 10800.0 & 10800.0 & \textbf{9006.4} \\
\hline
\multirow{3}{*}{100} & \multirow{3}{*}{2} & 75  & 2645.1  & 978.7   & 97.5   & \textbf{61.1} \\
                     &                    & 150 & 10800.0 & 8792.4  & 728.6  & \textbf{387.7} \\
                     &                    & 300 & 10800.0 & 10800.0 & 6089.6 & \textbf{1123.7} \\
\hline
\multirow{5}{*}{25} & \multirow{5}{*}{3} & 30  & 1935.3  & 782.6   & 43.5   & \textbf{21.7} \\
                    &                    & 75  & 7342.2  & 8532.4  & 112.5  & \textbf{63.7} \\
                    &                    & 150 & 10800.0 & 10800.0 & 236.7  & \textbf{156.2} \\
                    &                    & 300 & 10800.0 & 10800.0 & 436.0  & \textbf{321.7} \\
                    &                    & 600 & 10800.0 & 10800.0 & 893.9  & \textbf{700.0} \\
\hline
\multirow{5}{*}{50} & \multirow{5}{*}{3} & 30  & 1541.4  & 576.7   & 43.9   & \textbf{25.6} \\
                    &                    & 75  & 8945.7  & 7076.3  & 212.3  & \textbf{128.4} \\
                    &                    & 150 & 10800.0 & 10800.0 & 512.1  & \textbf{385.0} \\
                    &                    & 300 & 10800.0 & 10800.0 & 2791.8 & \textbf{1116.4} \\
                    &                    & 600 & 10800.0 & 10800.0 & \textbf{8506.3} & 8612.5 \\
\hline
\multirow{4}{*}{100} & \multirow{4}{*}{3} & 30  & 749.0   & 372.6   & 50.3   & \textbf{24.1} \\
                     &                    & 75  & 10800.0 & 2488.6  & 289.7  & \textbf{134.3} \\
                     &                    & 150 & 10800.0 & 10800.0 & 949.7  & \textbf{566.0} \\
                     &                    & 300 & 10800.0 & 10800.0 & 5638.4 & \textbf{1533.0} \\
\hline
\multirow{5}{*}{25} & \multirow{5}{*}{4} & 30  & 2785.1  & 1204.5  & 63.6   & \textbf{38.7} \\
                    &                    & 75  & 5957.6  & 3916.7  & 108.1  & \textbf{86.0} \\
                    &                    & 150 & 10800.0 & 10800.0 & 350.5  & \textbf{173.3} \\
                    &                    & 300 & 10800.0 & 10800.0 & 412.9  & \textbf{408.3} \\
                    &                    & 600 & 10800.0 & 10800.0 & 1111.2 & \textbf{940.8} \\
\hline
\multirow{4}{*}{50} & \multirow{4}{*}{4} & 30  & 1790.2  & 853.4   & 73.1   & \textbf{43.7} \\
                    &                    & 75  & 7123.2  & 8353.0  & 396.8  & \textbf{243.6} \\
                    &                    & 150 & 10800.0 & 10800.0 & 1176.4 & \textbf{675.5} \\
                    &                    & 300 & 10800.0 & 10800.0 & 4411.4 & \textbf{1836.1} \\
\hline
\multirow{4}{*}{100} & \multirow{4}{*}{4} & 30  & 627.3   & 472.5   & 49.3   & \textbf{38.2} \\
                     &                    & 75  & 4838.1  & 2356.0  & 363.4  & \textbf{192.4} \\
                     &                    & 150 & 10800.0 & 10800.0 & 3503.8 & \textbf{1141.5} \\
                     &                    & 300 & 10800.0 & 10800.0 & \textbf{8805.9} & 10800.0 \\
\hline
\multirow{5}{*}{25} & \multirow{5}{*}{5} & 30  & 5607.0  & 2289.8  & 119.7  & \textbf{83.6} \\
                    &                    & 75  & 8520.0  & 8884.5  & 273.0  & \textbf{195.5} \\
                    &                    & 150 & 8634.0  & 8784.4  & 429.6  & \textbf{342.7} \\
                    &                    & 300 & 10800.0 & 10800.0 & 733.1  & \textbf{660.8} \\
                    &                    & 600 & 10800.0 & 10800.0 & 1328.8 & \textbf{1283.8} \\
\hline
\multirow{4}{*}{50} & \multirow{4}{*}{5} & 30  & 1332.1  & 959.2   & 148.6  & \textbf{94.1} \\
                    &                    & 75  & 8565.3  & 8631.5  & 512.6  & \textbf{373.5} \\
                    &                    & 150 & 10800.0 & 10800.0 & 2696.7 & \textbf{1535.7} \\
                    &                    & 300 & 10800.0 & 10800.0 & 10800.0 & \textbf{7172.8} \\
\hline
\multirow{3}{*}{100} & \multirow{3}{*}{5} & 30  & 1724.5  & 627.9   & 221.7  & \textbf{133.3} \\
                     &                    & 75  & 10800.0 & 7426.8  & 876.4  & \textbf{591.4} \\
                     &                    & 150 & 10800.0 & 10800.0 & 3545.2 & \textbf{1994.6} \\
\hline
\end{tabular}

\end{table}

\begin{table}[!hbtp]
\caption{Computational time taken in seconds as a function of the two acceleration techniques available under multi-purchase problems using the $\ell_1$-error model (shifted geometric mean of the duration in seconds) for problems with $|\mathcal{N}| = 15$ and a market share of 50\%.}
\label{tab:cpu_times_cb_up_exp1:mp:l1:share50}
\centering
\begin{tabular}{cccrrrr}
$|\mathcal{C}|$ & $\eta$ & $|\mathcal{P}|$ & None & UP & CB & Both \\ \hline
\multirow{4}{*}{25} & \multirow{4}{*}{2} & 30  & 7231.4  & 1242.8  & 105.6   & \textbf{43.4} \\
                    &                    & 75  & 10800.0 & 10800.0 & 447.3   & \textbf{241.7} \\
                    &                    & 150 & 10800.0 & 10800.0 & 1439.0  & \textbf{785.1} \\
                    &                    & 300 & 10800.0 & 10800.0 & 10800.0 & \textbf{3998.1} \\
\hline
\multirow{3}{*}{50} & \multirow{3}{*}{2} & 30  & 8909.5  & 1698.4  & 254.8   & \textbf{105.3} \\
                    &                    & 75  & 10800.0 & 10800.0 & 1356.4  & \textbf{585.0} \\
                    &                    & 150 & 10800.0 & 10800.0 & 8900.4  & \textbf{2469.0} \\
\hline
\multirow{2}{*}{100} & \multirow{2}{*}{2} & 30  & 4756.7  & 1539.3  & 164.7   & \textbf{80.8} \\
                     &                    & 75  & 10800.0 & 10800.0 & 2076.9  & \textbf{796.1} \\
\hline
\multirow{3}{*}{25} & \multirow{3}{*}{3} & 30  & 8183.6  & 6143.6  & 165.0   & \textbf{82.0} \\
                    &                    & 75  & 10800.0 & 10800.0 & 1061.0  & \textbf{549.2} \\
                    &                    & 150 & 10800.0 & 10800.0 & 6128.5  & \textbf{1676.2} \\
\hline
\multirow{3}{*}{50} & \multirow{3}{*}{3} & 30  & 10800.0 & 10800.0 & 426.0   & \textbf{205.6} \\
                    &                    & 75  & 10800.0 & 10800.0 & 4906.4  & \textbf{1378.9} \\
                    &                    & 150 & 10800.0 & 10800.0 & 10800.0 & \textbf{6273.4} \\
\hline
\multirow{2}{*}{100} & \multirow{2}{*}{3} & 30  & 10800.0 & 7177.0  & 394.7   & \textbf{215.2} \\
                     &                    & 75  & 10800.0 & 10800.0 & 7071.4  & \textbf{2766.5} \\
\hline
\multirow{3}{*}{25} & \multirow{3}{*}{4} & 30  & 10800.0 & 7403.2  & 394.3   & \textbf{187.6} \\
                    &                    & 75  & 10800.0 & 10800.0 & 6878.5  & \textbf{2724.6} \\
                    &                    & 150 & 10800.0 & 10800.0 & 10800.0 & \textbf{7106.1} \\
\hline
\multirow{2}{*}{50} & \multirow{2}{*}{4} & 30  & 10800.0 & 10800.0 & 1176.6  & \textbf{570.2} \\
                    &                    & 75  & 10800.0 & 10800.0 & 8886.3  & \textbf{4537.9} \\
\hline
\multirow{1}{*}{100} & \multirow{1}{*}{4} & 30  & 10800.0 & 10800.0 & 1607.0  & \textbf{525.8} \\
\hline
\multirow{1}{*}{25} & \multirow{1}{*}{5} & 30  & 10800.0 & 8766.5  & 1012.6  & \textbf{512.6} \\
\hline
\multirow{2}{*}{50} & \multirow{2}{*}{5} & 30  & 10800.0 & 10800.0 & 4668.7  & \textbf{1801.1} \\
                    &                    & 75  & 10800.0 & 10800.0 & 10800.0 & \textbf{7337.4} \\
\hline
\multirow{2}{*}{100} & \multirow{2}{*}{5} & 30  & 10800.0 & 10800.0 & 3310.7  & \textbf{1218.8} \\
                     &                    & 75  & 10800.0 & 10800.0 & 10800.0 & \textbf{9021.2} \\
\hline
\end{tabular}
\end{table}

\begin{table}[!hbtp]
\caption{Computational time taken in seconds as a function of the two acceleration techniques available under multi-purchase problems using the $\ell_1$-error model (shifted geometric mean of the duration in seconds) for problems with $|\mathcal{N}| = 15$ and a market share of 80\%.}
\label{tab:cpu_times_cb_up_exp1:mp:l1:share80}
\centering
\begin{tabular}{cccrrrr}
$|\mathcal{C}|$ & $\eta$ & $|\mathcal{P}|$ & None & UP & CB & Both \\ \hline
\multirow{3}{*}{25} & \multirow{3}{*}{2} & 30  & 7438.2  & 1875.7  & 171.0   & \textbf{66.4} \\
                    &                    & 75  & 10800.0 & 10800.0 & 619.6   & \textbf{317.9} \\
                    &                    & 150 & 10800.0 & 10800.0 & 3277.9  & \textbf{1035.1} \\
\hline
\multirow{3}{*}{50} & \multirow{3}{*}{2} & 30  & 10800.0 & 6138.1  & 419.8   & \textbf{171.8} \\
                    &                    & 75  & 10800.0 & 10800.0 & 4768.4  & \textbf{944.9} \\
                    &                    & 150 & 10800.0 & 10800.0 & 10800.0 & \textbf{9026.9} \\
\hline
\multirow{2}{*}{100} & \multirow{2}{*}{2} & 30  & 10800.0 & 8778.6  & 715.0   & \textbf{282.6} \\
                     &                    & 75  & 10800.0 & 10800.0 & 10800.0 & \textbf{1420.0} \\
\hline
\multirow{2}{*}{25} & \multirow{2}{*}{3} & 30  & 10800.0 & 10800.0 & 429.2   & \textbf{205.5} \\
                    &                    & 75  & 10800.0 & 10800.0 & 1990.8  & \textbf{864.4} \\
\hline
\multirow{2}{*}{50} & \multirow{2}{*}{3} & 30  & 10800.0 & 10800.0 & 1276.1  & \textbf{557.4} \\
                    &                    & 75  & 10800.0 & 10800.0 & 10800.0 & \textbf{9018.3} \\
\hline
\multirow{1}{*}{100} & \multirow{1}{*}{3} & 30  & 10800.0 & 10800.0 & 2767.6  & \textbf{650.5} \\
\hline
\multirow{3}{*}{25} & \multirow{3}{*}{4} & 30  & 10800.0 & 10800.0 & 806.7   & \textbf{472.4} \\
                    &                    & 75  & 10800.0 & 10800.0 & 9011.9  & \textbf{4957.2} \\
                    &                    & 300 & 10800.0 & 10800.0 & 10800.0 & \textbf{8992.5} \\
\hline
\multirow{1}{*}{50} & \multirow{1}{*}{4} & 30  & 10800.0 & 10800.0 & 7077.5  & \textbf{1159.7} \\
\hline
\multirow{1}{*}{100} & \multirow{1}{*}{4} & 30  & 10800.0 & 10800.0 & 8892.9  & \textbf{2179.6} \\
\hline
\multirow{1}{*}{25} & \multirow{1}{*}{5} & 30  & 10800.0 & 10800.0 & 3410.4  & \textbf{1271.5} \\
\hline
\multirow{1}{*}{50} & \multirow{1}{*}{5} & 30  & 10800.0 & 10800.0 & 8827.3  & \textbf{6778.2} \\
\hline
\multirow{1}{*}{100} & \multirow{1}{*}{5} & 30  & 10800.0 & 10800.0 & 10800.0 & \textbf{7106.5} \\
\hline
\end{tabular}
\end{table}


\newpage

\section{Comparison between DP and the MIP for the multi-purchase EM estimation problem}\label{sec:detailled_multipurchase_setting}

In this section, we provide a more detailed comparison between DP and the MIP for the multi-purchase EM estimation problem.
Tables~\ref{tab:mip-dp:n4:c100:s80}-\ref{tab:mip-dp:n18:c25:s80} report results grouped by the number of products, $|\mathcal{N}|$, the cardinality of the set of consumer types in the ground-truth model, $|\mathcal{C}|$, the market share (\textit{Share}), the maximum number of purchases per consumer type, $\eta$, and the number of periods.
In the main body of the manuscript, results are aggregated only with respect to \textit{Share} and $\eta$.
\begin{table}[!hbtp]
\centering
\caption{Comparison between the DP and the MIP for the multi-purchase EM estimation problem for $|\mathcal{N}|=10, |\mathcal{C}|=25, \text{Share}=20$ \label{tab:mip-dp:n10:c25:s20}}
\begin{tabular}{cccc}
\hline
$\eta$ & $|\mathcal{P}|$ & DP & MIP \\ \hline
2 & 300 & \textbf{6.8} & 821.9 \\
2 & 600 & \textbf{13.7} & 1245.3 \\
3 & 600 & \textbf{40.5} & 898.6 \\ \hline
\end{tabular}
\end{table}

\begin{table}[!hbtp]
\centering
\caption{Comparison between the DP and the MIP for the multi-purchase EM estimation problem for $|\mathcal{N}|=10, |\mathcal{C}|=50, \text{Share}=20$ \label{tab:mip-dp:n10:c50:s20}}
\begin{tabular}{cccc}
\hline
$\eta$ & $|\mathcal{P}|$ & DP & MIP \\ \hline
2 & 150 & \textbf{3.1} & 625.1 \\
2 & 300 & \textbf{11.2} & 1814.1 \\
2 & 600 & \textbf{41.3} & 4643.3 \\
3 & 300 & \textbf{22.4} & 820.6 \\
3 & 600 & \textbf{67.7} & 3190.6 \\
4 & 600 & \textbf{103.4} & 970.1 \\ \hline
\end{tabular}
\end{table}

\begin{table}[!hbtp]
\centering
\caption{Comparison between the DP and the MIP for the multi-purchase EM estimation problem for $|\mathcal{N}|=10, |\mathcal{C}|=100, \text{Share}=20$ \label{tab:mip-dp:n10:c100:s20}}
\begin{tabular}{cccc}
\hline
$\eta$ & $|\mathcal{P}|$ & DP & MIP \\ \hline
2 & 150 & \textbf{5.3} & 827.9 \\
2 & 300 & \textbf{17.1} & 2111.1 \\
2 & 600 & \textbf{67.3} & 8598.7 \\
3 & 150 & \textbf{11.7} & 617.5 \\
3 & 300 & \textbf{31.1} & 1276.6 \\
3 & 600 & \textbf{84.1} & 6401.4 \\
4 & 600 & \textbf{118.1} & 2245.7 \\ \hline
\end{tabular}
\end{table}

\begin{table}[!hbtp]
\centering
\caption{Comparison between the DP and the MIP for the multi-purchase EM estimation problem for $|\mathcal{N}|=10, |\mathcal{C}|=25, \text{Share}=50$ \label{tab:mip-dp:n10:c25:s50}}
\begin{tabular}{cccc}
\hline
$\eta$ & $|\mathcal{P}|$ & DP & MIP \\ \hline
2 & 150 & \textbf{4.2} & 630.9 \\
2 & 300 & \textbf{11.3} & 1902.8 \\
2 & 600 & \textbf{26.5} & 3544.0 \\
3 & 300 & \textbf{27.7} & 792.2 \\
3 & 600 & \textbf{56.2} & 1218.1 \\ \hline
\end{tabular}
\end{table}

\begin{table}[!hbtp]
\centering
\caption{Comparison between the DP and the MIP for the multi-purchase EM estimation problem for $|\mathcal{N}|=10, |\mathcal{C}|=50, \text{Share}=50$ \label{tab:mip-dp:n10:c50:s50}}
\begin{tabular}{cccc}
\hline
$\eta$ & $|\mathcal{P}|$ & DP & MIP \\ \hline
2 & 150 & \textbf{7.2} & 1077.6 \\
2 & 300 & \textbf{26.5} & 3201.7 \\
2 & 600 & \textbf{67.6} & 10800.0 \\
3 & 300 & \textbf{44.4} & 2194.3 \\
3 & 600 & \textbf{110.0} & 8329.2 \\
4 & 300 & \textbf{58.2} & 1264.7 \\
4 & 600 & \textbf{130.0} & 1508.4 \\ \hline
\end{tabular}
\end{table}

\begin{table}[!hbtp]
\centering
\caption{Comparison between the DP and the MIP for the multi-purchase EM estimation problem for $|\mathcal{N}|=10, |\mathcal{C}|=100, \text{Share}=50$ \label{tab:mip-dp:n10:c100:s50}}
\begin{tabular}{cccc}
\hline
$\eta$ & $|\mathcal{P}|$ & DP & MIP \\ \hline
2 & 150 & \textbf{9.1} & 1271.4 \\
2 & 300 & \textbf{30.5} & 7200.1 \\
2 & 600 & \textbf{79.3} & 10800.0 \\
3 & 150 & \textbf{17.0} & 1494.0 \\
3 & 300 & \textbf{37.5} & 4944.8 \\
3 & 600 & \textbf{135.2} & 10800.0 \\
4 & 150 & \textbf{27.7} & 632.5 \\
4 & 300 & \textbf{72.7} & 1794.3 \\
4 & 600 & \textbf{162.0} & 9577.5 \\
5 & 600 & \textbf{236.6} & 1001.5 \\ \hline
\end{tabular}
\end{table}

\begin{table}[!hbtp]
\centering
\caption{Comparison between the DP and the MIP for the multi-purchase EM estimation problem for $|\mathcal{N}|=10, |\mathcal{C}|=25, \text{Share}=80$ \label{tab:mip-dp:n10:c25:s80}}
\begin{tabular}{cccc}
\hline
$\eta$ & $|\mathcal{P}|$ & DP & MIP \\ \hline
2 & 150 & \textbf{6.0} & 961.2 \\
2 & 300 & \textbf{13.7} & 1560.7 \\
2 & 600 & \textbf{37.1} & 4890.3 \\
3 & 300 & \textbf{48.2} & 783.3 \\
3 & 600 & \textbf{96.4} & 3794.0 \\
4 & 600 & \textbf{108.7} & 886.3 \\ \hline
\end{tabular}
\end{table}

\begin{table}[!hbtp]
\centering
\caption{Comparison between the DP and the MIP for the multi-purchase EM estimation problem for $|\mathcal{N}|=10, |\mathcal{C}|=50, \text{Share}=80$ \label{tab:mip-dp:n10:c50:s80}}
\begin{tabular}{cccc}
\hline
$\eta$ & $|\mathcal{P}|$ & DP & MIP \\ \hline
2 & 150 & \textbf{11.4} & 1673.7 \\
2 & 300 & \textbf{37.6} & 8209.4 \\
2 & 600 & \textbf{87.4} & 10800.0 \\
3 & 150 & \textbf{21.6} & 1167.2 \\
3 & 300 & \textbf{61.7} & 3391.5 \\
3 & 600 & \textbf{146.5} & 8679.9 \\
4 & 300 & \textbf{91.7} & 2581.6 \\
4 & 600 & \textbf{193.0} & 5110.4 \\
5 & 600 & \textbf{229.9} & 871.7 \\ \hline
\end{tabular}
\end{table}


\begin{table}[!hbtp]
\centering
\caption{Comparison between the DP and the MIP for the multi-purchase EM estimation problem for $|\mathcal{N}| = 10, |\mathcal{C}|=100, \text{Share}=80$ \label{tab:mip-dp:n4:c100:s80}}
\begin{tabular}{cccc}
\hline
$\eta$ & $|\mathcal{P}|$ & DP & MIP \\ \hline
2 & 75 & \textbf{3.5} & 851.9 \\
2 & 150 & \textbf{12.8} & 2333.7 \\
2 & 300 & \textbf{58.3} & 10800.0 \\
2 & 600 & \textbf{140.2} & 10800.0 \\ \hline
\end{tabular}
\end{table}


\begin{table}[!hbtp]
\centering
\caption{Comparison between the DP and the MIP for the multi-purchase EM estimation problem for  $|\mathcal{N}| = 10, |\mathcal{C}|=100, \text{Share}=80$ \label{tab:mip-dp:n6:c100:s80}}
\begin{tabular}{cccc}
\hline
$\eta$ & $|\mathcal{P}|$ & DP & MIP \\ \hline
3 & 75 & \textbf{9.3} & 913.0 \\
3 & 150 & \textbf{24.8} & 3418.6 \\
3 & 300 & \textbf{74.5} & 10800.0 \\
3 & 600 & \textbf{186.0} & 10800.0 \\
4 & 150 & \textbf{40.6} & 1164.2 \\
4 & 300 & \textbf{111.6} & 6367.0 \\
4 & 600 & \textbf{259.8} & 10800.0 \\
5 & 300 & \textbf{160.8} & 740.0 \\
5 & 600 & \textbf{449.1} & 2003.0 \\ \hline
\end{tabular}
\end{table}



\begin{table}[!hbtp]
\centering
\caption{Comparison between the DP and the MIP for the multi-purchase EM estimation problem for $|\mathcal{N}|=15, |\mathcal{C}|=25, \text{Share}=20$ \label{tab:mip-dp:n15:c25:s20}}
\begin{tabular}{cccc}
\hline
$\eta$ & $|\mathcal{P}|$ & DP & MIP \\ \hline
2 & 30 & \textbf{2.3} & 9414.4 \\
2 & 75 & \textbf{11.8} & 10800.0 \\
2 & 150 & \textbf{32.8} & 10800.0 \\
2 & 300 & \textbf{95.1} & 10800.0 \\
2 & 600 & \textbf{282.9} & 10800.0 \\
3 & 30 & \textbf{7.1} & 617.1 \\
3 & 75 & \textbf{31.5} & 9224.9 \\
3 & 150 & \textbf{143.1} & 10800.0 \\
3 & 300 & \textbf{377.5} & 8733.0 \\
3 & 600 & \textbf{1073.2} & 10800.0 \\
4 & 150 & \textbf{144.7} & 1996.4 \\
4 & 300 & \textbf{501.8} & 3206.1 \\
4 & 600 & \textbf{1290.6} & 10800.0 \\
5 & 150 & \textbf{314.6} & 620.0 \\
5 & 300 & \textbf{941.6} & 2857.7 \\
5 & 600 & \textbf{1825.1} & 2729.9 \\ \hline
\end{tabular}
\end{table}

\begin{table}[!hbtp]
\centering
\caption{Comparison between the DP and the MIP for the multi-purchase EM estimation problem for $|\mathcal{N}|=15, |\mathcal{C}|=50, \text{Share}=20$ \label{tab:mip-dp:n15:c50:s20}}
\begin{tabular}{cccc}
\hline
$\eta$ & $|\mathcal{P}|$ & DP & MIP \\ \hline
2 & 30 & \textbf{2.7} & 10800.0 \\
2 & 75 & \textbf{19.0} & 10800.0 \\
2 & 150 & \textbf{79.2} & 10800.0 \\
2 & 300 & \textbf{488.8} & 10800.0 \\
2 & 600 & \textbf{786.4} & 10800.0 \\
3 & 30 & \textbf{8.9} & 792.9 \\
3 & 75 & \textbf{57.1} & 6926.7 \\
3 & 150 & \textbf{275.6} & 10800.0 \\
3 & 300 & \textbf{897.3} & 10800.0 \\
3 & 600 & \textbf{1652.4} & 10800.0 \\
4 & 75 & \textbf{136.1} & 697.4 \\
4 & 150 & \textbf{442.7} & 4490.8 \\
4 & 300 & \textbf{1347.9} & 10800.0 \\
4 & 600 & \textbf{1836.8} & 10800.0 \\
5 & 150 & \textbf{843.4} & 1238.1 \\
5 & 300 & \textbf{1715.2} & 3814.3 \\
5 & 600 & \textbf{1872.5} & 10800.0 \\ \hline
\end{tabular}
\end{table}

\begin{table}[!hbtp]
\centering
\caption{Comparison between the DP and the MIP for the multi-purchase EM estimation problem for $|\mathcal{N}|=15, |\mathcal{C}|=100, \text{Share}=20$ \label{tab:mip-dp:n15:c100:s20}}
\begin{tabular}{cccc}
\hline
$\eta$ & $|\mathcal{P}|$ & DP & MIP \\ \hline
2 & 30 & \textbf{4.9} & 10800.0 \\
2 & 75 & \textbf{33.6} & 10800.0 \\
2 & 150 & \textbf{104.1} & 10800.0 \\
2 & 300 & \textbf{397.3} & 10800.0 \\
2 & 600 & \textbf{1313.9} & 10800.0 \\
3 & 75 & \textbf{65.1} & 7028.7 \\
3 & 150 & \textbf{280.9} & 10800.0 \\
3 & 300 & \textbf{754.0} & 10800.0 \\
3 & 600 & \textbf{1774.5} & 10800.0 \\
4 & 75 & \textbf{161.1} & 683.4 \\
4 & 150 & \textbf{597.8} & 7082.2 \\
4 & 300 & \textbf{1642.4} & 10800.0 \\
4 & 600 & \textbf{1835.9} & 10800.0 \\
5 & 150 & \textbf{844.1} & 2011.2 \\
5 & 300 & \textbf{1750.6} & 7592.1 \\
5 & 600 & \textbf{1834.7} & 10800.0 \\ \hline
\end{tabular}
\end{table}

\begin{table}[!hbtp]
\centering
\caption{Comparison between the DP and the MIP for the multi-purchase EM estimation problem for $|\mathcal{N}|=15, |\mathcal{C}|=25, \text{Share}=50$ \label{tab:mip-dp:n15:c25:s50}}
\begin{tabular}{cccc}
\hline
$\eta$ & $|\mathcal{P}|$ & DP & MIP \\ \hline
2 & 30 & \textbf{4.6} & 10800.0 \\
2 & 75 & \textbf{15.3} & 10800.0 \\
2 & 150 & \textbf{48.7} & 10800.0 \\
2 & 300 & \textbf{107.3} & 10800.0 \\
2 & 600 & \textbf{439.6} & 10800.0 \\
3 & 30 & \textbf{10.5} & 3370.8 \\
3 & 75 & \textbf{60.5} & 2355.0 \\
3 & 150 & \textbf{178.1} & 8603.6 \\
3 & 300 & \textbf{639.7} & 10800.0 \\
3 & 600 & \textbf{1425.1} & 10800.0 \\
4 & 75 & \textbf{117.4} & 4460.3 \\
4 & 150 & \textbf{228.6} & 5872.0 \\
4 & 300 & \textbf{685.2} & 10800.0 \\
4 & 600 & \textbf{1502.4} & 10800.0 \\
5 & 75 & \textbf{305.8} & 692.9 \\
5 & 150 & \textbf{590.2} & 3252.8 \\
5 & 300 & \textbf{1202.2} & 9479.0 \\
5 & 600 & \textbf{1821.1} & 10800.0 \\ \hline
\end{tabular}
\end{table}

\begin{table}[!hbtp]
\centering
\caption{Comparison between the DP and the MIP for the multi-purchase EM estimation problem for $|\mathcal{N}|=15, |\mathcal{C}|=50, \text{Share}=50$ \label{tab:mip-dp:n15:c50:s50}}
\begin{tabular}{cccc}
\hline
$\eta$ & $|\mathcal{P}|$ & DP & MIP \\ \hline
2 & 30 & \textbf{9.9} & 10800.0 \\
2 & 75 & \textbf{48.7} & 10800.0 \\
2 & 150 & \textbf{123.6} & 10800.0 \\
2 & 300 & \textbf{381.9} & 10800.0 \\
2 & 600 & \textbf{1339.3} & 10800.0 \\
3 & 30 & \textbf{29.9} & 5214.4 \\
3 & 75 & \textbf{133.4} & 10800.0 \\
3 & 150 & \textbf{381.2} & 10800.0 \\
3 & 300 & \textbf{1138.1} & 10800.0 \\
3 & 600 & \textbf{1608.8} & 10800.0 \\
4 & 75 & \textbf{254.0} & 8317.8 \\
4 & 150 & \textbf{598.5} & 9339.4 \\
4 & 300 & \textbf{1524.4} & 10800.0 \\
4 & 600 & \textbf{1833.9} & 10800.0 \\
5 & 75 & \textbf{523.9} & 1184.1 \\
5 & 150 & \textbf{871.3} & 8918.3 \\
5 & 300 & \textbf{1667.6} & 10800.0 \\
5 & 600 & \textbf{1837.7} & 10800.0 \\ \hline
\end{tabular}
\end{table}

\begin{table}[!hbtp]
\centering
\caption{Comparison between the DP and the MIP for the multi-purchase EM estimation problem for $|\mathcal{N}|=15, |\mathcal{C}|=100, \text{Share}=50$ \label{tab:mip-dp:n15:c100:s50}}
\begin{tabular}{cccc}
\hline
$\eta$ & $|\mathcal{P}|$ & DP & MIP \\ \hline
2 & 30 & \textbf{10.2} & 10800.0 \\
2 & 75 & \textbf{75.3} & 10800.0 \\
2 & 150 & \textbf{322.4} & 10800.0 \\
2 & 300 & \textbf{1201.9} & 10800.0 \\
2 & 600 & \textbf{1821.3} & 10800.0 \\
3 & 30 & \textbf{30.3} & 9118.1 \\
3 & 75 & \textbf{168.3} & 10800.0 \\
3 & 150 & \textbf{620.6} & 10800.0 \\
3 & 300 & \textbf{1652.3} & 10800.0 \\
3 & 600 & \textbf{1833.1} & 10800.0 \\
4 & 30 & \textbf{87.6} & 1210.7 \\
4 & 75 & \textbf{413.5} & 9050.2 \\
4 & 150 & \textbf{1308.7} & 10800.0 \\
4 & 300 & \textbf{1823.0} & 10800.0 \\
4 & 600 & \textbf{1843.4} & 10800.0 \\
5 & 75 & \textbf{714.0} & 2309.5 \\
5 & 150 & \textbf{1654.4} & 10800.0 \\
5 & 300 & \textbf{1827.7} & 10800.0 \\
5 & 600 & \textbf{1840.3} & 10800.0 \\ \hline
\end{tabular}
\end{table}

\begin{table}[!hbtp]
\centering
\caption{Comparison between the DP and the MIP for the multi-purchase EM estimation problem for $|\mathcal{N}|=18, |\mathcal{C}|=25, \text{Share}=80$ \label{tab:mip-dp:n18:c25:s80}}
\begin{tabular}{cccc}
\hline
$\eta$ & $|\mathcal{P}|$ & DP & MIP \\ \hline
2 & 30 & \textbf{6.4} & 5405.5 \\
2 & 75 & \textbf{19.8} & 10800.0 \\
2 & 150 & \textbf{61.8} & 10800.0 \\
2 & 300 & \textbf{182.6} & 10800.0 \\
2 & 600 & \textbf{334.3} & 10800.0 \\
3 & 30 & \textbf{25.5} & 1929.8 \\
3 & 75 & \textbf{76.3} & 7451.8 \\
3 & 150 & \textbf{236.2} & 7863.6 \\
3 & 300 & \textbf{826.4} & 10800.0 \\
3 & 600 & \textbf{1351.2} & 10800.0 \\
4 & 75 & \textbf{123.2} & 2186.1 \\
4 & 150 & \textbf{291.0} & 8443.1 \\
4 & 300 & \textbf{729.7} & 10800.0 \\
4 & 600 & \textbf{1479.2} & 10800.0 \\
5 & 75 & \textbf{251.4} & 842.7 \\
5 & 150 & \textbf{487.3} & 6328.6 \\
5 & 300 & \textbf{1613.1} & 10800.0 \\
5 & 600 & \textbf{1847.1} & 10800.0 \\ \hline
\end{tabular}
\end{table}

\begin{table}[!hbtp]
\centering
\caption{Comparison between the DP and the MIP for the multi-purchase EM estimation problem for $|\mathcal{N}|=15, |\mathcal{C}|=50, \text{Share}=80$ \label{tab:mip-dp:n15:c50:s80}}
\begin{tabular}{cccc}
\hline
$\eta$ & $|\mathcal{P}|$ & DP & MIP \\ \hline
2 & 30 & \textbf{13.0} & 10800.0 \\
2 & 75 & \textbf{69.6} & 10800.0 \\
2 & 150 & \textbf{155.2} & 10800.0 \\
2 & 300 & \textbf{504.2} & 10800.0 \\
2 & 600 & \textbf{1052.8} & 10800.0 \\
3 & 30 & \textbf{37.0} & 8942.3 \\
3 & 75 & \textbf{193.3} & 10800.0 \\
3 & 150 & \textbf{363.9} & 10800.0 \\
3 & 300 & \textbf{1195.4} & 10800.0 \\
3 & 600 & \textbf{1829.3} & 10800.0 \\
4 & 30 & \textbf{71.3} & 675.6 \\
4 & 75 & \textbf{371.1} & 5315.3 \\
4 & 150 & \textbf{754.7} & 10800.0 \\
4 & 300 & \textbf{1719.3} & 10800.0 \\
4 & 600 & \textbf{1841.1} & 10800.0 \\
5 & 75 & \textbf{506.6} & 2677.9 \\
5 & 150 & \textbf{1210.7} & 10800.0 \\
5 & 300 & \textbf{1766.1} & 10800.0 \\
5 & 600 & \textbf{1858.3} & 10800.0 \\ \hline
\end{tabular}
\end{table}

\begin{table}[!hbtp]
\centering
\caption{Comparison between the DP and the MIP for the multi-purchase EM estimation problem for $|\mathcal{N}|=15, |\mathcal{C}|=100, \text{Share}=80$ \label{tab:mip-dp:n15:c100:s80}}
\begin{tabular}{cccc}
\hline
$\eta$ & $|\mathcal{P}|$ & DP & MIP \\ \hline
2 & 30 & \textbf{18.1} & 10800.0 \\
2 & 75 & \textbf{111.3} & 10800.0 \\
2 & 150 & \textbf{407.8} & 10800.0 \\
2 & 300 & \textbf{1198.7} & 10800.0 \\
2 & 600 & \textbf{1822.9} & 10800.0 \\
3 & 30 & \textbf{63.3} & 9288.6 \\
3 & 75 & \textbf{264.9} & 10800.0 \\
3 & 150 & \textbf{751.2} & 10800.0 \\
3 & 300 & \textbf{1674.4} & 10800.0 \\
3 & 600 & \textbf{1823.4} & 10800.0 \\
4 & 30 & \textbf{144.4} & 4094.2 \\
4 & 75 & \textbf{614.3} & 10800.0 \\
4 & 150 & \textbf{1502.1} & 10800.0 \\
4 & 300 & \textbf{1818.6} & 10800.0 \\
4 & 600 & \textbf{1845.3} & 10800.0 \\
5 & 75 & \textbf{937.4} & 7404.6 \\
5 & 150 & \textbf{1745.5} & 10800.0 \\
5 & 300 & \textbf{1830.0} & 10800.0 \\
5 & 600 & \textbf{1872.1} & 10800.0 \\ \hline
\end{tabular}
\end{table}

\end{document}